\documentclass[12pt]{article}
\usepackage{amsmath,amsthm,amsfonts}
\usepackage{graphicx,psfrag,epsf}
\usepackage{enumerate}
\usepackage{url}
\usepackage{subfig}
\usepackage{multirow}
\usepackage{natbib}
\usepackage{amssymb}
\usepackage{bbm}
\usepackage{placeins}
\usepackage{float}
\usepackage{afterpage}
\usepackage{array}
\usepackage{multirow}
\usepackage{setspace}
\usepackage{epsfig}
\usepackage{booktabs}
\usepackage{comment}
\usepackage{xcolor}
\usepackage{float}
\usepackage{hyperref}

\usepackage{xr}
\makeatletter
\newcommand*{\addFileDependency}[1]{
  \typeout{(#1)}
  \@addtofilelist{#1}
  \IfFileExists{#1}{}{\typeout{No file #1.}}
}
\makeatother

\usepackage[colorinlistoftodos]{todonotes}
\usepackage{algorithm2e}
	\SetKwComment{Comment}{\#\#\#}{}
	\DontPrintSemicolon
	\SetAlgoLined
	\RestyleAlgo{boxruled}
\usepackage{titlesec}

\DeclareMathOperator*{\argmin}{arg\,min}

\newcommand{\di}{i}
\newcommand{\tr}{\text{tr}}

\theoremstyle{plain}
\newtheorem{theorem}{Theorem}

\newtheorem{lemma}{Lemma}

\newtheorem{assumption}{Assumption}

\theoremstyle{definition}
\newtheorem{definition}{Definition}

\newtheorem{proposition}{Proposition}[section]

\theoremstyle{remark}
\newtheorem{remark}[theorem]{Remark}

\newcommand{\mbf}{\mathbf}

\newcommand{\blind}{1}

\begin{document}

\def\spacingset#1{\renewcommand{\baselinestretch}%
{#1}\small\normalsize} \spacingset{1}


\if1\blind
{
  \title{\bf Localized Sparse Principal Component Analysis of Multivariate Time Series in the Frequency Domain}
  \author{Jamshid Namdari \\
    Department of Biostatistics \& Bioinformatics, Emory University\\
    Amita Manatunga \\
    Department of Biostatistics \& Bioinformatics, Emory University\\
    Fabio Ferrarelli  \\
    Department of Psychiatry, University of Pittsburgh\\
    and\\
    Robert T. Krafty\thanks{
    Corresponding author Robert T. Krafty, Department of Biostatistics \& Bioinformatics, Emory University, Atlanta, GA 30322 (e-mail:rkrafty@emory.edu). This work is supported by National Institutes of Health grants R01GM140476, R01HL159213 and R01MH125816.} \\  
    Department of Biostatistics \& Bioinformatics, Emory University}
  \maketitle
} \fi

\if0\blind
{
  \bigskip
  \bigskip
  \bigskip
  \begin{center}
    {\LARGE\bf Localized Sparse Principal Component Analysis of Multivariate Time Series in Frequency Domain}
\end{center}
  \medskip
} \fi

\bigskip
\begin{abstract}
Principal component analysis has been a main tool in multivariate analysis for estimating
a low dimensional linear subspace that explains most of the variability in the
data. However, in high-dimensional regimes, naive estimates of the principal loadings
are not consistent and difficult to interpret. In the context of time series, principal
component analysis of spectral density matrices can provide valuable, parsimonious
information about the behavior of the underlying process, particularly if the principal
components are interpretable in that they are sparse in coordinates and localized in
frequency bands. In this paper, we introduce a formulation and consistent estimation
procedure for interpretable principal component analysis for high-dimensional time
series in the frequency domain. An efficient frequency-sequential algorithm is developed
to compute sparse-localized estimates of the low-dimensional principal subspaces
of the signal process. The method is motivated by and used to understand neurological
mechanisms from high-density resting-state EEG in a study of first episode
psychosis.
\end{abstract}

\noindent%
\vfill

\newpage
\spacingset{1.9} 
\section{Introduction} \label{sec:Introduction}
Since its first descriptions by  \cite{pearson1901liii}, 
principal component analysis (PCA) has been one of the main multivariate analysis techniques for dimension reduction and feature extraction. PCA has become an essential tool for not just independent and identically distributed (iid) multivariate data, but also for serially correlated multivariate time series data in both the time and frequency domains.  In the frequency domain, PCA as a sequential method for finding directions of maximum variability 
appeared in the work of \cite{Brillinger1964a}.  The principal component series are formulated through an optimal linear filtering that transmits a $p$-dimensional signal through a $d$-dimensional channel and recovers it with minimum loss of information.  A foundational discussion of theory and applications of PCA in frequency domain can be found in \cite{brillinger2001time}; recent applications of this framework include uncovering non-coherent block structures \citep{sundararajan2021}, time-frequency analysis \citep{ombao2005} and change point detection \citep{jiao2021}.

PCA for the frequency domain analysis of high-dimensional multivariate time series faces several challenges.  The first challenge, which is not unique to frequency domain PCA and is a challenge for PCA in general, is high-dimensionality. When the dimension is fixed, sample eigenvectors, and consequently sample estimates of the principal components, are consistent and asymptotically normally distributed. 
However, in high-dimensional regimes, where the dimension of the random variable grows, sample PCs fail to be consistent. For the single spiked covariance model in the iid multivariate setting,
it has been shown that 
the leading eigenvector of the sample covariance matrix can actually be orthogonal to the leading eigenvector of the population covariance matrix if its corresponding eigenvalue is not sufficiently large \citep{paul2007asymptotics}. 

To obtain consistent estimates of PCs, \cite{johnstone2009consistency} proposed to obtain PCs that have sparse representation in an orthonormal system. Reviews of sparse PCA can be found in \cite{johnstone2018pca} and in \cite{zou2018selective}.  Aside from the theoretical benefits or necessities for PCA in high-dimensions, sparsification also provides interpretation that is essential for effective data analysis.  For example, consider our motivating application of resting state EEG data that is discussed in  Section \ref{sec:Data_Analysis}.  Figure \ref{timeseriesfigure} displays subsets of EEG signals from 9 channels, or locations in the brain, that are part of 64 channel recording, from two individuals, one who is experiencing a first psychotic episode (FEP) and one who is a healthy control (HC), for one minute while resting with their eyes open. We desire a frequency domain analysis of each of these data that can provide insights into underlying dependence structure of the signals at each frequency across and within each location in the brain. This inherently requires low-dimensional representations that are interpretable as combinations of power at important frequencies, which requires localization in frequency, and within certain channels or regions of the brain, which is equivalent to sparsity within coordiantes.

\FloatBarrier
\begin{figure}[t!] 
\begin{center}
\begin{tabular}{cc}
\includegraphics[width=2.8in, height=3in, bb = 50 50 600 600]{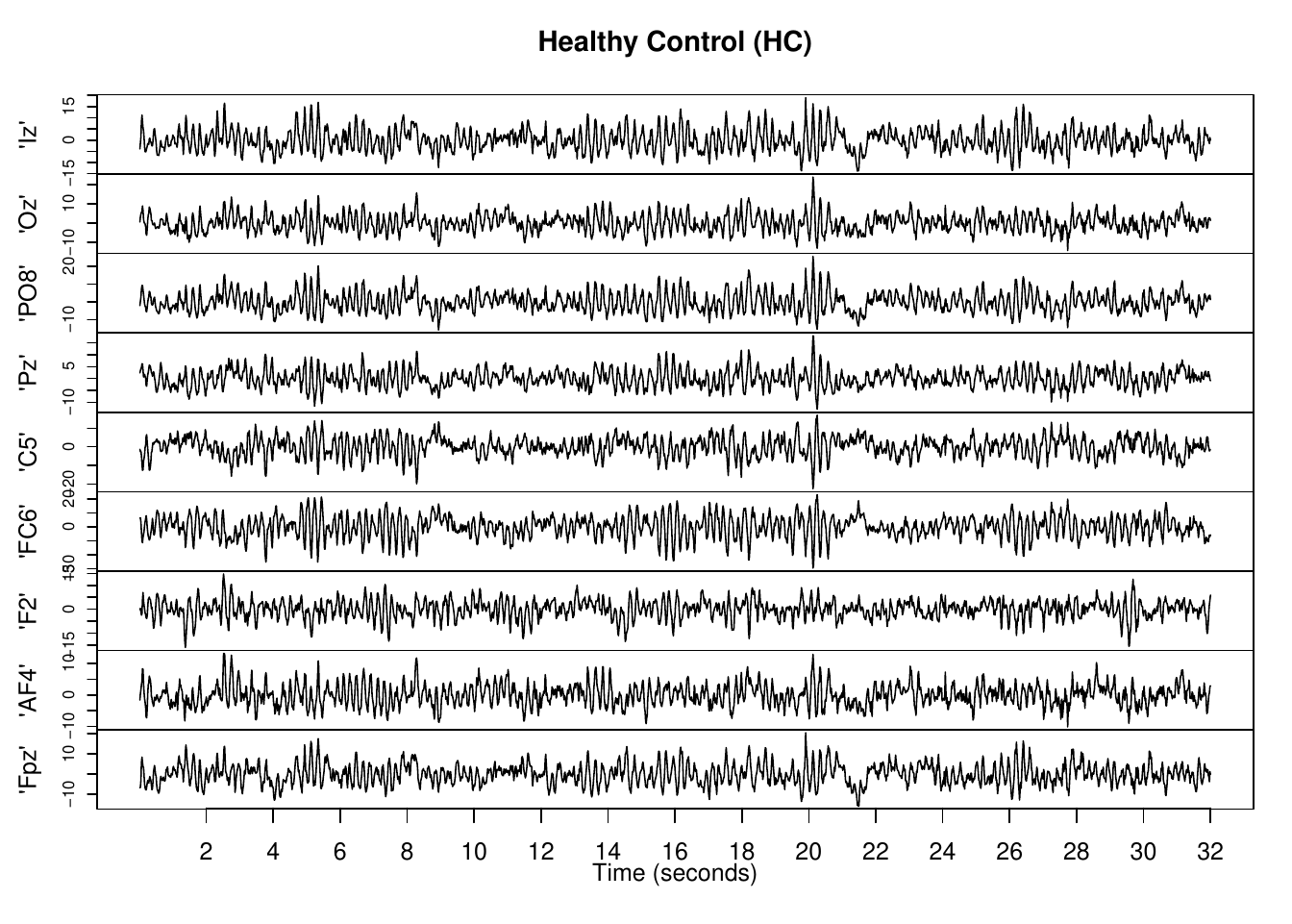} 
\hspace{12pt}
& \includegraphics[width=2.8in, height=3in, bb = 50 50 600 600]{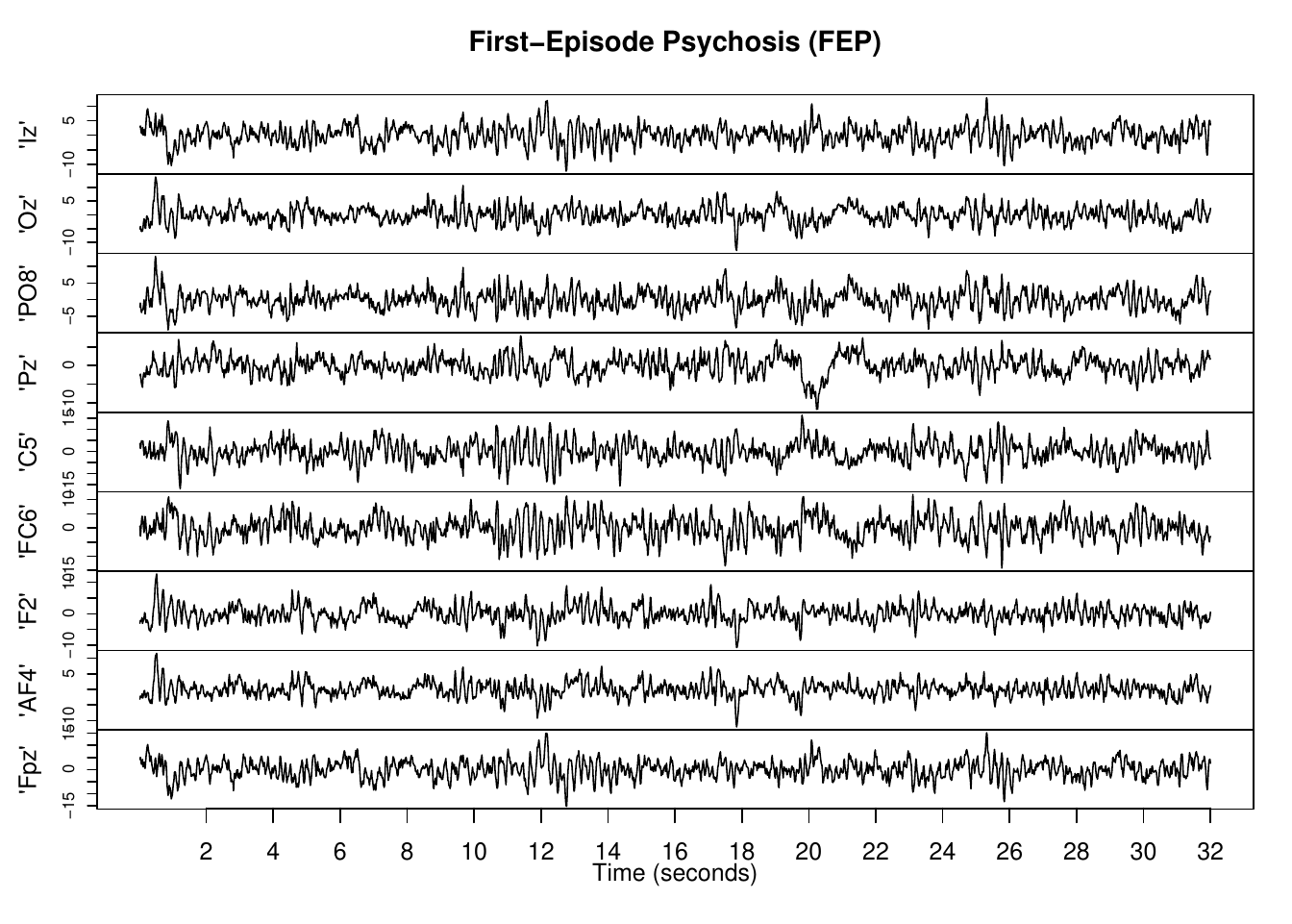}
\end{tabular}
         \caption{Resting state EEG data from two individuals, a healthy control (HC, left) and a person experiencing a first psychotic episode (FEP, right), from a subset of 9 channels from a 64-channel montage.}
         \label{timeseriesfigure}
\end{center}

\end{figure}


Various formulations of the PC problem, alternative methods of imposing sparsity, and convex relaxation of the corresponding optimization problems have inspired a notable volume of research.  For iid multivarite data, these methods include diagonal thresh-holding for spiked covariance models \citep{johnstone2009consistency},
regularized regression \citep{zou2006sparse}, 
low rank matrix approximation with sparsity constraint \citep{shen2008sparse, witten2009penalized}, 
maximizing variability over sparse subspaces 
\citep{d2004direct,vu2013fantope}, and 
subspace estimation via orthogonal iteration with sparsification \citep{ma2013sparse,yuan2013truncated}.  
%
%
Literature on the analysis of sparse PCA for serially dependent data is dearth compared to that for iid multivarite observations. In the time domain, sparse principal subspace estimation under vector autoregressive models were studied by \cite{wang2013sparse_s}. \cite{wang2014nonconvex} developed a two step method of Fantope projection and selection (FPS) \citep{vu2013fantope} followed by an orthogonal iteration method with sparsification (SOAP) and studied its theoretical properties for non-Gaussian and dependent data. Sparse PCA is even less explored in the frequency domain.  To the best of our knowledge, the only previous method for sparse PCA in the frequency domain is a two-step FPS followed by SOAP approach introduced by \citep{lusparse}.

%


Being equivalent to finding a sequence of principal subspaces of spectral matrices across frequency, high-dimensional spectral domain PCA has an additional layer of complexity compared to the iid multivaraite case that only considers the principal subspace of a single matrix.  This additional complexity presents two types of challenges.  First, the ability to consistently estimate a $d$-dimensional principal subspace relies on the first $d$ eigenvalues being sufficiently larger than the remaining, which is often referred to as having a sufficient eigengap.  This presents theoretical challenges in a PCA for the entire spectrum across all frequencies since, although there can be a sufficient eigengap at certain frequencies of interest, one is unable to reliably estimate principal subspaces of spectral matrices at frequencies with low power where the eigengap is small. A second challenge emerges with regards to interpretation.  The challenge of interpretation is typically addressed by applied researchers by summarizing frequency-domain information by collapsing power within a finite number of pre-defined frequency bands. 
Although historically derived, pre-defined frequency bands have been shown to be associated with a variety of scientific mechanisms, they are not optimal for parsimoniously summarizing and describing information in any given signal.  There has been considerable recent research into methods to address this challenge and to learn frequency bands for a given time series that are optimal in some sense \citep{bruce2020, tuft2023}.  However, to the best of our knowledge, there exist no approach to provide a frequency-domain PCA of a stationary time series that is interpretable in that principal components are localized in frequency band.  

The concept of localization to improve interpretation has been developed within the context of functional data analysis, where it is often referred to as ``interpretable functional data analysis'' \citep{james2009functional, zhang2021}.  Although functional data analytic methods have been developed for the frequency-domain analysis of replicated time series where several independent time series realizations are observed \citep{krafty2011, krafty2015}, it should be noted that  a different setting and question is considered in this article.  This article considers the analysis of a single realization of a multivariate time series, so that we consider PCA in the sense of \cite{Brillinger1964a} that involves principal subspaces of spectral matrices as operators on finite complex vector spaces, and not PCA in the functional sense that considers principal subspaces of functional operators over a continuous domain of frequency. 


The broad contribution of this paper is the introduction of the first method for conducting interpretable frequency-domain PCA that is both sparse among variables as well as localized in {bands of} frequency.  
We formulate a definition for the PCA of multivariate time series that contains a low-rank signal of interest with sparse{, continuous,} localized principal subspaces.  
We propose a sequential algorithm for estimating the principal subspaces. The approach estimates the sparse $d$-dimensional principal subspace itself, thus avoiding some of the challenges associated with deflation that is required in approaches to sparse PCA that sequentially estimate one-dimensional subspaces. Through simulation, we have studied the performance of the proposed algorithms on estimation of the underlying principal subspaces as well as parameter selection. On the theoretical side, we have established the consistency of the estimated principal subspaces in high-dimensions. The proof builds upon and extends the arguments presented in \cite{wang2014nonconvex} and \cite{lusparse} to account for smoothness of the principle subspace, including a new concentration inequality that shares information across frequency. 

The remainder of this paper is organized as follows. In Section \ref{sec:Main_Results}, a formulation of the localized and sparse principal subspaces of the underlying process is described. Section \ref{sec:Estimation} is devoted to the development of the estimation procedure. The theoretical analysis of the estimator is covered in Section \ref{subsec:theoretical_analysis} and the selection of tuning parameters is discussed in Section \ref{sec:model_selection}.  Empirical properties are investigated through simulation studies in Section \ref{sec:Simulation} and through the analysis of the motivating EEG data in Section \ref{sec:Data_Analysis}. Section \ref{sec:Discussion} provides a discussion of limitations and future directions.  Proofs of theoretical results and additional details with regards to the algorithms, simulation results and data analysis are provided in Supplementary Material.

\subsection*{Notation: }
Let $\mbf A=[A_{i,j}]\in\mathbb{C}^{p\times p}$. We denote conjugate transpose of $\mbf A$ by $\mbf A^\dagger$ and will use it to represent transpose of a real valued matrix as well. Let $u\in\mathbb{C}^p$, the $\ell_2$-norm of $u$ is defined as $\|u\|_2=\sqrt{u^\dagger u}$ and the $\ell_0$-norm of $u$ is the number of non-zero elements of $u$. 
For matrices $\mbf A_1$ and $\mbf A_2$, we define the inner product as $\langle \mbf A_1 , \mbf A_2 \rangle = \tr(\mbf A_1^\dagger \mbf A_2)$ and $\|\mbf A\|_F=\sqrt{\langle \mbf A,\mbf A\rangle}$, where $\tr(\mbf A)$ is the trace of  $A$. 
In this paper, $\|\mbf A\|_{2,0}$ determines the number of non-zero rows of $\mbf A$ and $\|\mbf A\|_{1,1}:=\sum_{i,j} |\mbf A_{i,j}|$. In addition, for Hermitian matrices $\mbf A, \mbf B\in\mathbb{C}^{p\times p}$, $\mbf A\preccurlyeq \mbf B$ if and only if $\mbf B-\mbf A$ is positive definite.  We denote the real and imaginary parts by $\Re(\cdot)$ and $\Im(\cdot)$, respectively, and we define the unit ball in $\mathbb{C}^p$ by $\mathbb{S}^{p-1}(\mathbb{C})=\{v\in\mathbb{C}^p \mid \|v\|_2=1\}$.  

\section{Localized Sparse Principal Components} \label{sec:Main_Results}

\subsection{Principal Components in Frequency Domain} \label{subsec:PCA_definition}

Let $\{X(t): t\in\mathbb{Z}\}$ be a $p$-dimensional stationary time series with mean vector $\mathbb{E}[X] = \mu$, auto-covariance matrix $\mbf\Gamma(h) = \mathbb{E}\left\{ \left[ X(t+h)-\mu \right] \left[ X(t)-\mu \right]^\dagger\right\}$, $h=0,\pm 1,\dots$, and spectral density matrix 
    $\mbf f(\omega) = (2\pi)^{-1}\sum_{h=-\infty}^{\infty}\mbf\Gamma(h)\exp\{-\di\omega h\}$, $-\infty < \omega < \infty,$
that is continuous as a function of frequency. 
Consider the decomposition $X(t) = \vartheta(t) + \varepsilon(t),\; t\in\mathbb{Z},$
where $\vartheta$ is the time series that is the closest time series to $X$ in terms of mean square error that can be obtained after compressing then reconstructing $X$ through a $d$-dimensional linear filter.  Formally, $\vartheta$ is defined by the $d\times p$ filter $\{\mbf b(h)\}$ and the $p\times d$ filter $\{\mbf c(h)\}$ such that $\vartheta(t) = \mu_{\vartheta} + \sum_h \mbf c(t-h)\zeta(h)$ and $\zeta(t) = \sum_h \mbf b(t-h)X(h)$  
minimizes 
\begin{equation} \label{eq:l2_est_loss}
    \mathbb{E}\left\{ \left[ X(t)-\vartheta(t) \right]^\dagger \left[X(t)-\vartheta(t) \right] \right\}
\end{equation} 
over all possible $d\times p$ and $p\times d$ filters.

Let $\mbf B(\omega) = \sum_h \mbf b(h)\exp\{\di\omega h\}$ and $\mbf C(\omega) = \sum_h \mbf c(h)\exp\{\di\omega h\}$ be the corresponding transfer functions. The next theorem, presented in Brillinger (2001), identifies the optimal transfer functions that minimizes Equation (\ref{eq:l2_est_loss}).
 
\begin{theorem} \label{thm:PC_series}
    Let $\{X(t)\}$ be a $p$-dimensional weakly stationary time series with mean vector $\mu$, an absolutely summable autocovariance function $\mbf \Gamma(h)$, and spectral density matrix $\mbf f(\omega), -\infty < \omega < \infty$. Then the $\mu_{\vartheta}, \{\mbf b(h)\}$, and $\{\mbf c(h)\}$ that minimizes (\ref{eq:l2_est_loss}) are given by $\mu_{\vartheta} = \mu - \left[\sum_h \mbf c(h)\right]\left[\sum_h \mbf b(h)\right]\mu$, $\mbf b(h) = (2\pi)^{-1}\int_0^{2\pi} \mbf B(\alpha)\exp\{ih\alpha\} \;d\alpha$, and $\mbf c(h) = (2\pi)^{-1}\int_0^{2\pi} \mbf C(\alpha)\exp\{ih\alpha\} \;d\alpha,$
    where $\mbf C(\omega) = [U_1(\omega) \dots U_d(\omega)]$, $\mbf B(\omega) = \mbf C^\dagger(\omega)$, and $U_j(\omega)$ is the $j$-th eigenvector of $\mbf f(\omega), j=1,\dots,p$. In addition, if $\lambda_j(\omega)$ denotes the corresponding eigenvalue, $j=1,\dots,p$, then the minimum obtained is $\int_0^{2\pi}\sum_{j>d}\lambda_j(\alpha) \;d\alpha$.
\end{theorem}
Note that if we denote $\mbf A(\omega)=\mbf C(\omega)\mbf B(\omega)$, then
\begin{align*}
    \mathbb{E}\left\{ \left[ X(t)-\vartheta(t) \right]^\dagger\left[X(t)-\vartheta(t)\right] \right\} &= \int_0^{2\pi} \tr\left\{ \left[I-\mbf A\left(\omega\right)\right] \mbf f(\omega) \left[ I-\mbf A(\omega) \right]^\dagger\right\}\; d(\omega) \\
    \mbf A(\omega) &= U_1(\omega) U_1(\omega)^\dagger + \dots + U_d(\omega) U_d(\omega)^\dagger.
\end{align*}
In other words, for each $\omega\in[0,1)$, $\mbf A(\omega)$ is a rank $d$ projection matrix. This indicates that the minimizer of
    $\int_0^{2\pi} \tr\left\{\left[I-\mbf A(\omega)\right] \mbf f(\omega) \left[I-\mbf A(\omega) \right]^\dagger\right\}\; d(\omega)$ 
over the space of rank-$d$ projection matrices $\mathcal{G}_d$ is equivalent to the solution to the maximization problem 
\begin{equation} \label{eq:max_prob}
    \max_{A\in\mathcal{G}_d} \int_0^{2\pi} \tr \left[ \mbf f(\omega)\mbf A(\omega)\right]\; d(\omega). 
\end{equation}

The focus of this article is the interpretation and estimation of the principal subspace spanned by the orthogonal directions $U_j\left(\omega\right)$, considered with reference to the eigenvalues $\lambda_j\left(\omega\right)$. It should be noted that the power spectrum of $\vartheta$ can be represented as
$\mbf f_{\vartheta}\left(\omega\right) = \sum_{j=1}^d \lambda_j\left(\omega\right) U_j\left(\omega\right) U_j\left(\omega\right)^{\dagger}.$
The principal time series $\zeta_{j}\left(t\right)$, $j=1,\dots,d$, are uncorrelated time series with power spectra $\lambda_j\left(\omega\right)$ that represent parsimonious underlying latent mechanisms that account for most of the information in $X\left(t\right)$. The orthogonal directions $U_j\left(\omega\right)$ describe how these latent time series relate to and can be interpreted from the perspective of the $p$-dimensional space.  For example, in our analysis of the EEG data that is presented in Section \ref{sec:Data_Analysis}, $U_j$ and $\lambda_j$ are dominated by information within a subset of the slow delta frequencies less than 4 Hz, which has been shown to be most prominent during times of and can be used to electrophysiologically quantify rest, and within a subset of the theta frequencies between 4 - 7 Hz, which has been shown to be associated with attention control. The principal time series $\zeta_j$ represent uncorrelated relative expression of these two mechanisms, and the principal directions $U_j$ indicate how these mechanisms are expressed in each location of the brain.

\subsection{Sparsity} \label{subsec:sparse_def}

This article is concerned with PCA where  principal subspaces are sparse in variates. This assumption is essential both to make estimation tractable in the high-dimensional setting, as well as for interpretation. For example, in our motivating application, we desire a parsimonious interpretation where a component represents power in only certain regions of the brain by estimating principal subspaces that are sparse in the following sense.

\begin{definition}[Subspace Sparsity]  \label{def:sparsity_level}
    Let $\mathcal{U}$ be a $d$-dimensional subspace of $\mathbb{C}^p$ and $\mathbb{U}$ be the set of $p\times d$ orthonormal matrices whose columns span $\mathcal{U}$. Let $\mbf\Pi=\mbf U\mbf U^\dagger, \mbf U\in\mathbb{U}$ be the unique (orthogonal) projection matrix onto the subspace $\mathcal{U}$. We define the sparsity level of $\mathcal{U}$ as $s = \left|\mbox{supp}\left[\mbox{diag}\left(\mbf\Pi\right)\right]\right|$.  
\end{definition}

We desire a principal component analysis under the assumption that the principal subspace  that is spanned by $U_1\left(\omega\right), \dots, U_d \left(\omega\right)$ is sparse with sparsity level $s^*$. 

\subsection{Frequency Localization} \label{subsec:local_def}

In addition to sparsity, we also desire a PCA that is localized in the frequency domain in that only the most relevant frequencies are retained.  Frequency localization is important for two reasons.  In terms of estimation, the ability to consistently estimate the $d$-dimensional principal subspace of a matrix depends on the difference between the $d$th and $d+1$st eigenvalues. In many applications, including our motivating EEG example where there is low signal at higher frequencies, this eigengap is not sufficiently large for spectral matrices at many frequencies.  In terms of interpretation and applications, we desire a parsimonious decomposition of information in that principal components can be interpretable as having support withing certain ranges or bands of frequency.  This desire also mitigates the issues caused by the inability to consistently estimate principal subspaces at frequencies with insufficient eigengaps as said information is not of practical interest.  Formally, we desire a frequency localization procedure that identifies frequencies with sufficient power by finding $\Omega \subset \left[0, 2\pi \right)$ such that the power of $X$ at frequency $\omega$ that is accounted for by $\vartheta$, or $\sum_{j=1}^d \lambda_j(\omega)$, is greater than some threshold for all $\omega \in \Omega$.  Although this threshold can be selected either subjectively or based on existing scientific knowledge, in Section \ref{sec:model_selection} we discuss a data driven procedure for selecting this threshold relative to the variance of the remainder 
$\epsilon_t$
and present details in Appendix A. 

{
\subsection{Frequency Bands and Continuity of  Principal Subspaces}
Continuity of the power spectrum is a common theoretical assumption that is justifiable in most practical applications, including for EEG.  We utilize the assumption that  principal subspaces are continuous as a function of frequency in two ways.  First, it allows for the sharing of information in adjacent frequencies, which improves estimation, especially in areas with small eigen-gap.  Second, we utilize the continuity of the principal subspaces to improve interpretation.  When combined with frequency localization, regularizing the continuity of the principal subspace results in estimates of the principal subspaces that are not localized in individual frequencies, but in bands of frequencies.   
%
}

\subsection{Optimization Problem}\label{subsec:opt_def}
{We combine Equation (\ref{eq:max_prob}) to formulate the PCA problem with appropriate regularizations to obtain estimates of the principal subspaces that are localized, sparse, and preserve smoothness presented in the underlying principal subspaces, as a function of frequency. In order to obtain a sparse solution, we can add the $\ell_0$ penalty term for each frequency component, i.e. via $\|\mbf V(\omega_\ell)\|_{2,0} \leq s^*$. To obtain a localized solution, we propose to discretize the objective function and add a shrinkage parameter, $\beta_\ell$, for each frequency $\omega_\ell$. To control variation in estimated principal subspaces at consecutive fundamental frequencies, an intuitive approach can be to add the constraint $\|\mbf\Pi_\ell-\mbf\Pi_{\ell+1}\|_F < \varepsilon$ for all $\ell=1,\dots,n/2-1$.  We relax this constraint to obtain a computationally feasible, sequential optimization problem.

Given a realization $X(1),\dots,X(n)$ of the time series $\{X(t):t\in\mathbb{Z}\}$, we consider the frequency transformation of the data given by
\begin{equation}\label{def:fn}
    \mbf f_n(\omega_\ell) = \sum_{t=-M}^{M}\hat{\mbf R}_t\exp\{-2\pi i\omega t\},
\end{equation}
where $\hat{\mbf R}_t=\frac{1}{n}\sum_{k=1}^{n-t}X(k+t)X(t)^\dagger$ with the dependence of $M$ on $n$ implicit so that $M$ can grow with $n$ as $n\to\infty$. Note that, when $M=n/2$, $\mbf f_n$ is the standard periodogram and when $M<n/2$, $\mbf f_n$ is a truncated periodogram. In establishing the theoretical properties of our proposed estimator, we specify an appropriate rate at which $M$ can grow relative to $n$ and $p$ such that we maintain consistency of the estimator.

With the aforementioned considerations, we can consider estimating the sparse and localized principal components through solving the following optimization problem.
\begin{equation} \label{eq:LSPC_opt0}
    \begin{split}
    \underset{\ell=1,\dots,n/2}{\underset{\beta_\ell,\mbf V(\omega_\ell)}{\mbox{maximize}}}\; & \left\{\sum_{\ell=1}^{n/2} \beta_\ell \tr\left[\mbf{f}_n^{}(\omega_\ell)\mbf V(\omega_\ell)\mbf V(\omega_\ell)^\dagger\right] \right\}\\
    \mbox{subject to: } &\;\; \mbf V(\omega_\ell) \mbox{ orthonormal, } \|\mbf V(\omega_\ell)\|_{2,0} \leq s^* \mbox{ and } \\
    &\; \|\mbf V(\omega_\ell)\mbf V(\omega_\ell)^\dagger-\mbf V(\omega_{\ell+1})\mbf V(\omega_{\ell+1})^\dagger\|_F < \varepsilon, \ell=1,\dots,n/2-1 \mbox{ and } \\
    &\; \sum_{\ell=1}^{n/2} \beta_\ell \leq \eta \; \mbox{ and } 0\leq\beta_\ell\leq 1, \ell=1,\dots,n/2.
    \end{split}
\end{equation}
Note that the constraint set in Equation (\ref{eq:LSPC_opt0}) is non-convex and solving it directly is a challenging problem. We approach the problem through relaxing the constraint set, reformulating the problem as a different regularization problem and propose a sequential procedure to obtain a solution. We show in Section \ref{subsec:theoretical_analysis} that, for large enough $n$, the solution obtained through the proposed sequential procedure is feasible for (\ref{eq:LSPC_opt0}), with high probability, and is a consistent estimate of the underlying principal subspaces. 

Let $\Theta = \{(\mbf\Pi_1,\dots,\mbf\Pi_{n/2}): \|\mbf\Pi_\ell-\mbf\Pi_{\ell+1}\|_F < \varepsilon$ for all $\ell=1,\dots,n/2-1\}$. We can show that $\|\mbf\Pi_\ell-\mbf\Pi_{\ell+1}\|_F=\sqrt{2}\|\mbf\Pi_{\ell+1} - \mbf\Pi_\ell\mbf\Pi_{\ell+1}\mbf\Pi_\ell\|_F$, see Appendix B.2.1 for details. Thus, for any $(\mbf\Pi_1,\dots,\mbf\Pi_{n/2})\in\Theta$, $\sum_\ell\tr\{\beta_\ell \mbf f_n(\omega_{\ell+1})[\mbf\Pi_{\ell+1} - \mbf\Pi_\ell\mbf\Pi_{\ell+1}\mbf\Pi_\ell]\}\leq E$ for an appropriate $E$ (e.g, $E=2n\sqrt{d}\sup_{\ell}\{\lambda_1(\omega_\ell)\}\varepsilon$), and we define the relaxed constraint set $\tilde{\Theta}:=\{(\mbf\Pi_1,\dots,\mbf\Pi_{n/2}): \sum_\ell\tr\{\beta_\ell \mbf f_n(\omega_{\ell+1})[\mbf\Pi_{\ell+1} - \mbf\Pi_\ell\mbf\Pi_{\ell+1}\mbf\Pi_\ell]\}\leq E\}$.  While the constrained set $\Theta$ restricts the distance between principle subspaces at all pairwise adjacent frequencies, this relaxed set $\tilde{\Theta}$ restricts the sum of the distances across selected frequencies between the projection of the periodogram at a frequency onto its principle subspace and onto the principle subspace at the previous frequency.  Relaxing the constraint to $\tilde{\Theta}$ enables us to consider (\ref{eq:LSPC_opt0}) as 
\begin{equation*} \label{eq:LSPC_opt1}
    \begin{split}
    \underset{\ell=1,\dots,n/2}{\underset{\beta_\ell,V(\omega_\ell)}{\mbox{maximize}}}\; & \left\{\sum_{\ell=1}^{n/2} \beta_\ell \tr\left[\mbf{f}_n^{}(\omega_\ell)\mbf V(\omega_\ell)\mbf V(\omega_\ell)^\dagger\right] - \theta\sum_{\ell=1}^{n/2-1}\tr\left[\beta_\ell \mbf f_n(\omega_{\ell+1})\left(\mbf\Pi_{\ell+1} - \mbf\Pi_\ell\mbf\Pi_{\ell+1}\mbf\Pi_\ell\right)\right] \right\}\\
    \mbox{subject to: } &\;\; \mbf V(\omega_\ell) \mbox{ orthonormal, } \|\mbf V(\omega_\ell)\|_{2,0} \leq s^* \mbox{ and } \\
    &\; \sum_{\ell=1}^{n/2} \beta_\ell \leq \eta \; \mbox{ and } 0\leq\beta_\ell\leq 1, \ell=1,\dots,n/2,
    \end{split}
\end{equation*}
where the relation between $\theta$ and $E$ depends on the data. Simple calculations show that
$\tr\{\mbf f_n(\omega_{\ell+1})\mbf\Pi_{\ell+1} - \theta\mbf f_n(\omega_{\ell+1})[\mbf\Pi_{\ell+1}-\mbf\Pi_{\ell}\mbf\Pi_{\ell+1}\mbf\Pi_{\ell}]\} = \tr\{[(1-\theta)\mbf f_n(\omega_{\ell+1}) + \theta\mbf\Pi_\ell\mbf f_n(\omega_{\ell+1})\mbf\Pi_\ell]\mbf\Pi_{\ell+1}\}$. This enables us to rewrite 
the above optimization problem as the following problem 
\begin{equation} \label{eq:LSPC_opt_relaxed}
    \begin{split}
    \underset{\ell=1,\dots,n/2}{\underset{\beta_\ell,V(\omega_\ell)}{\mbox{maximize}}}\; & \left\{\sum_{\ell=1}^{n/2} \beta_\ell \tr\left[\mbf{f}^{(\theta)}(\omega)\mbf V(\omega_\ell)\mbf V(\omega_\ell)^\dagger\right] \right\}\\
    \mbox{subject to: } &\;\; \mbf V(\omega_\ell) \mbox{ orthonormal, } \|\mbf V(\omega_\ell)\|_{2,0} \leq s^* \mbox{ and } \\
    &\; \sum_{\ell=1}^{n/2} \beta_\ell \leq \eta \; \mbox{ and } 0\leq\beta_\ell\leq 1, \ell=1,\dots,n/2,
    \end{split}
\end{equation}
where $\mbf f^{(\theta)}(\omega_{\ell+1}) = (1-\theta)\mbf f_n(\omega_{\ell+1}) + \theta\mbf\Pi_\ell\mbf f_n(\omega_{\ell+1})\mbf\Pi_\ell$ and $\mbf f^{(\theta)}(\omega_1) = \mbf f_n(\omega_1)$, and we suppressed the dependence of $\mbf f^{(\theta)}$ on $n$ to ease notation. Note that, since $\Theta\subset\tilde{\Theta}$, a solution to (\ref{eq:LSPC_opt_relaxed}) may not be feasible for (\ref{eq:LSPC_opt0}). However, as will be shown in Section \ref{subsec:theoretical_analysis}, for large enough $n$ the solution obtained falls in $\Theta$ with high probability.

Observe that $\mbf f^{(\theta)}$ regularizes the object of interest, the principal subspace, by shrinking the principal subspace at frequency $\omega_{\ell+1}$ toward the principal subspace at frequency $\omega_\ell$, which implicitly encourages the estimated principal subspaces at consecutive frequencies to stay close together. The parameter $\theta$ controls the amount of information that is pooled from the previous frequency and implicitly controls the total variation in the estimated principal subspaces. In the extreme case where $\theta=0$, no sharing of information is incorporated into subspace estimation, and when $\theta=1$, the principal subspace at each frequency is projected on the space spanned by the corresponding previous frequency resulting in recovering the estimated principal subspace obtained at $\omega_1$.  

}


\section{Estimation Procedure} \label{sec:Estimation}
{
We propose to solve the optimization problem in (\ref{eq:LSPC_opt_relaxed}) sequentially as follows. First, we estimate $\mbf V(\omega_1),\dots,\mbf V(\omega_{n/2})$ sequentially, where, at frequency $\omega_{\ell+1}$ and conditional on estimates $\hat{\mbf V}(\omega_1),\dots,\hat{\mbf V}(\omega_\ell)$ for $\mbf V(\omega_1),\dots,\mbf V(\omega_\ell)$, we obtain an estimate for $\mbf V(\omega_{\ell+1})$ through solving 
    \begin{equation}\label{eq:spc_opt1}
    \underset{\mbf V(\omega_{\ell+1})}{\max}\; \tr\left[ \hat{\mbf{f}}^{(\theta)}(\omega_{\ell+1})\mbf V(\omega_{\ell+1})\mbf V(\omega_{\ell+1})^\dagger\right],\quad\mbox{s.t. }  \mbf V(\omega_{\ell+1}) \mbox{ orthonormal, } \|\mbf V(\omega_{\ell+1})\|_{2,0} \leq s^*,
\end{equation} 
where 
    $\hat{\mbf f}^{(\theta)}(\omega_{\ell+1}) = (1-\theta)\mbf f_n(\omega_{\ell+1}) + \theta\hat{\mbf\Pi}_\ell\mbf f_n(\omega_\ell)\hat{\mbf\Pi}_\ell, $
 $\hat{\mbf\Pi}_\ell = \hat{\mbf V}(\omega_\ell)\hat{\mbf V}^\dagger(\omega_\ell)$, and $\hat{\mbf f}^{(\theta)}(\omega_1)=\mbf f_n(\omega_1)$. 
Note that substituting $\mbf\Pi_\ell$ by its estimate is justified by the consistency of the estimated principal subspaces established in Theorem \ref{thm:consistency}. Next, we solve
\begin{equation}\label{eq:opt_beta}
    \underset{ \sum_{\ell=1}^{n/2} \beta_\ell \leq \eta \; \mbox{ and } 0\leq\beta_\ell\leq 1, \ell=1,\dots,n/2}{\underset{\beta_1,\dots,\beta_{n/2}}{\mbox{maximize}}}\; \left\{\sum_{\ell=1}^{n/2} \beta_\ell  \tr\left[ \hat{\mbf {f}}^{(\theta)}(\omega_\ell)\hat{\mbf V}(\omega_\ell)\hat{\mbf V}(\omega_\ell)^\dagger\right] \right\},
\end{equation}
which as shown in Proposition \ref{prop:linear_programming}, maintains only the $\eta$ frequency components with highest objective function. 
We solve the former problem sequentially by utilizing the sparse orthogonal iterated pursuit (SOAP) proposed by \cite{wang2014nonconvex}, where the initial value at $\omega_1$ is obtained by the Fantope projection and selection (FPS) and the estimate obtained at each frequency is used as the initial value for the next frequency component afterwards.

\subsection{Solution of the Sparse Principal Components} \label{subsec:sparse_PCA}

Although the solution to (\ref{eq:spc_opt1}) for each $\ell=1,\dots,n/2$ can attain the optimal statistical rate of convergence \citep{vu2013minimax}, it is NP-hard to compute \citep{moghaddam2005spectral}. Extensive research has been done to design a computationally feasible algorithm that enjoys optimal statistical convergence rate, a review of which can be found in Section D of  \cite{zou2018selective}. 

The optimization problem (\ref{eq:spc_opt1}) can be solved directly by the orthogonal iteration algorithm \citep{golub2013matrix} combined with a sparsification step. In addition, to ensure the solution attains optimal statistical rate of convergence, the initial value should fall within an appropriate distance from the solution. The initial estimate is obtained by applying the FPS method that solves a convex relaxation of the problem (\ref{eq:spc_opt1}). The two steps of (relaxed step) FPS followed by (tightened step) SOAP are explained below.

We first introduce the FPS algorithm for estimating the sparse PCs of a real matrix $\mbf{\Sigma}$, where we maximize $\tr\left[ \mbf{\Sigma}\mbf V\mbf V^\dagger\right]$ subject to $\mbf V$ being orthonormal and $\|\mbf V\|_{2,0} \leq s^*$. We then extend it so that it can estimate the sparse principal subspace of a complex valued spectral density matrix. Note that we can relax the constraint set of (\ref{eq:spc_opt1}) to obtain a convex relaxation of the problem. To do so, let $\mbf \Pi = \mbf V \mbf V^T$ and note that since $\mbf V$ is an orthonormal matrix, $\mbf \Pi$ is the projection matrix onto a $d$-dimensional subspace of $\mathbb{R}^p$ (in the real case). In addition, we know that $\mbf \Pi$ has exactly two eigenvalues, $1$ with multiplicity $d$ and $0$ with multiplicity $p-d$. Such constraint on eigenvalues of $\mbf \Pi$ can be relaxed to $\tr\left(\mbf \Pi\right)=d$ and $0\preccurlyeq \mbf \Pi\preccurlyeq \mathbb{I}_p$. In addition, we relax the constraint $\|\mbf \Pi\|_{2,0}\leq s^*$ to $\|\mbf \Pi\|_{1,1}\leq s^*$. Note that the constraint set $\{\|\mbf V\|_{2,0} \leq s^*\}$ is not convex, while the set $\mathcal{A} = \{\mbf \Pi: \mbf \Pi\in\mathbb{R}^{p\times p}, \tr({\mbf \Pi})=d, 0\preccurlyeq \mbf \Pi\preccurlyeq \mathbb{I}_p\}$ is convex.

The relaxed convex optimization problem can be equivalently expresses as
\begin{equation} \label{eq:relaxed_convex_opt}
    \mbox{minimize } \{\tr({\mbf \Sigma}\;\mbf \Pi) + \rho \|\mbf \Phi\|_{1,1} \;|\; \mbf \Pi=\mbf \Phi, \mbf \Pi\in\mathcal{A}, \mbf \Phi\in\mathbb{R}^{p\times p} \},
\end{equation}
with the Lagrangian $\mathcal{L}(\mbf \Pi,\mbf \Phi,\mbf \Theta) = \tr({\mbf \Sigma}\;\mbf \Pi) + \rho \|\mbf \Phi\|_{1,1} - \tr\left[{\mbf \Theta \left(\mbf \Pi-\mbf \Theta \right)}\right], \quad \mbf \Pi\in\mathcal{A}, \mbf \Phi\in\mathbb{R}^{p\times p}, \mbf\Theta\in\mathbb{R}^{p\times p}$,
and be solved by the alternating direction of multiplier (ADMM), which iteratively minimizes the augmented Lagrangian, 
\begin{equation} \label{eq:augmented_Lagrangian}
    \mathcal{L}(\mbf \Pi,\mbf \Phi,\mbf \Theta) +\beta/2\|\mbf \Pi-\mbf \Phi\|_F^2
\end{equation}
with respect to $\mbf \Pi$ and $\mbf \Phi$ and updating the dual variable $\mbf \Theta$. We only need to iterate the algorithm enough so that the calculated $\mbf \Pi$ at iteration $T$, $\mbf \Pi^{(T)}$, falls within the basin of attraction of the SOAP algorithm. A detailed description of the ADMM step can be found in Appendix A. Then, the top $d$ leading eigenvectors of $\mbf \Pi^{(T)}$ are used as the initial value in the SOAP algorithm.

To apply the FPS algorithm to complex valued metrics we invoke to Lemma B.4.3 from the Appendix B
that describes the isomorphism between complex matrices and real matrices. Let ${^{(R)}}{ \hat{\mbf {f}}^{(\theta)}}$ be the associated $(2p)\times(2p)$ real matrix to $\hat{\mbf f}^{(\theta)}$. We propose to apply the FPS algorithm to ${^{(R)}}{ \hat{\mbf {f}}^{(\theta)}}(\omega_\ell)$ and estimate the $2d$-dimensional principal subspace of ${^{(R)}}{ \hat{\mbf {f}}^{(\theta)}}(\omega_\ell)$ to obtain the initial estimate of $d$ leading eigenvectors of $\mbf {f}(\omega_\ell)$. As will be shown in Theorem (\ref{thm:consistency}), part (I), the estimated subspace obtained in this manner will be consistent.

In SOAP, the orthogonal iteration method is followed by a truncation step to enforce row-sparsity and further followed by taking another re-normalization step to enforce orthogonality. More precisely, at the $t$-th iteration of the algorithm the following operations are performed
\begin{itemize}
    \item Orthogonal iteration:  $\Tilde{\mbf V}^{(t+1)}\leftarrow\hat{\mbf {f}}^{(\theta)}(\omega_\ell)\mbf U^{(t)}; \; \mbf V^{(t+1)},\mathsf{\mbf R}_1^{(t+1)}\leftarrow\mbox{QR}(\Tilde{\mbf V}^{(t+1)})$
    \item Truncation/re-normalization: $\Tilde{\mbf U}^{(t+1)}\leftarrow\mbox{Truncate}(\mbf V^{(t+1)},\hat{s});\; \mbf U^{(t+1)},\mathsf{\mbf R}_2^{(t+1)}\leftarrow\mbox{QR}(\Tilde{\mbf U})^{(t+1)}$
\end{itemize}
where columns of $\mbf U^{(t)}$ contain the estimated first $d$ eigenvectors of $\mbf {f}(\omega_\ell)$ and the truncation operator sets the $p-\hat{s}$ rows with the smallest modulus to zero.

\subsection{Solution of the Linear Programming Problem}\label{subsec:localization}
Let $\hat{\mbf V}(\omega_\ell),\ell=1,\dots,n/2$ be the maximizer of $\tr\{[\hat{\mbf {f}}^{(\theta)}(\omega_\ell)]\mbf V(\omega_\ell)\mbf V(\omega_\ell)^\dagger\}$ such that $V(\omega_\ell)$ is orthonormal and $\|\mbf V(\omega_\ell)\|_{2,0} \leq s^*$.
Since $\hat{\mbf {f}}^{(\theta)}(\omega_\ell),\ell=1,\dots,n/2$ are positive definite, $h_\ell := \tr\left[ \hat{\mbf {f}}^{(\theta)}(\omega_\ell)\hat{\mbf V}(\omega_\ell)\hat{\mbf V}(\omega_\ell)^\dagger\right]>0$. Thus, we can write (\ref{eq:opt_beta}) as
\begin{equation}\label{eq:sum_beta_h}
    \underset{ \sum_{\ell=1}^{n/2} \beta_\ell \leq \eta \; \mbox{ and } 0\leq\beta_\ell\leq 1, \ell=1,\dots,n/2}{\underset{\beta_1,\dots,\beta_{n/2}}{\mbox{maximize}}}\; \left\{\sum_{\ell=1}^{n/2} \beta_\ell h_\ell \right\}.
\end{equation}
Note that, since $h_\ell>0, \ell=1,\dots,n/2$, the objective function is monotonically increasing in $\beta_\ell,\ell=1,\dots,n/2$ and therefore attains its maximum on the boundary of the constraint set. The algorithm selects the $\eta$ largest $h_j$'s and set the coefficients of the $n/2-\eta$ smallest $h_j$'s to zero. 
\begin{proposition} \label{prop:linear_programming}
Let $\eta\in\mathbb{N}$, $\eta < K$ for some $K\in\mathbb{N}$, and $h_1,\dots,h_K\in\mathbb{R}^+$. In addition, let $h_{(1)}\geq\dots\geq h_{(K)}$ be the sorted $h_j's$ in decreasing order and $\beta_{(1)},\dots,\beta_{(K)}$ be the corresponding coefficients in (\ref{eq:sum_beta_h}). Then
    \begin{equation}
    \underset{ \sum_{\ell=1}^{K} \beta_\ell \leq \eta \; \mbox{ and } 0\leq\beta_\ell\leq 1, \ell=1,\dots,K}{\underset{\beta_1,\dots,\beta_{K}}{\mbox{maximum}}}\; \left\{\sum_{\ell=1}^{K} \beta_\ell h_\ell \right\}.
\end{equation}
is attained at $\beta_{(1)}=\dots=\beta_{(\eta)}=1,\beta_{(\eta+1)}=\dots=\beta_{(K)}=0$.
\end{proposition}    

{
\subsection{LSPCA Algorithm}
Below, we summarize the estimation procedure described above and will refer to it as the localized sparse principal component analysis (LSPCA) algorithm. 

\BlankLine

\begin{algorithm}[H]
\caption{LSPCA} \label{alg:LSPCA}

\SetKwInput{KwInput}{Input}                
\SetKwInput{KwOutput}{Output}              
\SetKwInOut{Parameter}{Parameters}
\DontPrintSemicolon
\SetAlgoLined
\RestyleAlgo{boxruled}
  
  \KwInput{ $\{\mbf {f}_n(\omega_\ell)\}_{\ell=1}^{n/2}$ }
\Parameter{Regularization parameter $\rho>0$, penalty parameter $\beta>0$, maximum number of iterations of ADMM $T$, Sparsity parameter $\hat{s}$, Maximum number of iteration of SOAP $\Tilde{T}$, Frequency parameter $\eta$, Smoothing parameter $\theta$.}

\BlankLine
    \textbf{ADMM:} $\mbf \Pi_1\leftarrow\mbox{ADMM}\left[\mbf {f}_n(\omega_1)\right]$\;
    Set the columns of $\mbf U^{init}$ to be the top $d$ leading eigenvectors of $\mbf \Pi_1$\;
    $\tilde{\mbf U}^{(0)}\leftarrow\mbox{Truncate}(\mbf U^{init},\hat{s})$, $\mbf U^{init}\leftarrow\mbox{Thin-QR}(\tilde{\mbf U}^{(0)})$\;
    $\hat{\mbf U}^{(0)}\leftarrow\mbox{SOAP}\left[\hat{\mbf {f}}^{(\theta)}(\omega_1),\mbf U^{init}\right]$\;
    \For{$\ell\in\{0,\dots,n/2-1\}$}{
        $\quad \hat{\mbf U}(\omega_{\ell+1})\leftarrow\mbox{SOAP}\left[\hat{\mbf {f}}^{(\theta)}(\omega_{\ell+1}),\hat{\mbf U}(\omega_\ell)\right]$\;
    }
    $\hat{\beta}_1,\dots,\hat{\beta}_{n/2}\leftarrow\mbox{Solve (\ref{eq:sum_beta_h}) by linear programming.}$\; 
  \KwOutput{$\{\hat{\mbf U}(\omega_\ell),\hat{\beta}_\ell\}_{\ell=1}^{n/2}$}
\end{algorithm}

\section{Theoretical Analysis}\label{subsec:theoretical_analysis}
To investigate theoretical properties, we first introduce the notion of distance between subspaces, in addition to several key quantities that will be used in the theoretical analysis, then we present 
{
the model and assumptions under which theoretical properties were derived. Finally, we present
}
the rate of convergence of the estimated localized sparse principal components. 

\begin{itemize}
    \item  \textbf{Subspace distance: } Let $\mathcal{U}$ and $\mathcal{U}^\prime$ be two $d$-dimentional subspaces of $\mathbb{C}^p$. Denote the projection matrices onto them by $\mbf \Pi$ and $\mbf \Pi^\prime$, respectively. We define and denote the distance between $\mathcal{U}$ and $\mathcal{U}^\prime$ by $\mathcal{D}(\mathcal{U},\mathcal{U}^\prime)=\|\mbf \Pi-\mbf \Pi^\prime\|_F$.
    \item \textbf{Principal subspace notations: } Let $\mathcal{U}^*_\ell$ be the $d$-dimensional principal subspace of $\mbf f(\omega_\ell)$ for each fundamental frequency $\omega_\ell=\ell/n,\ell=2,\dots,n/2$ and $\mathcal{U}^{(t)}(\omega_1)$ be the $d$-dimensional subspace spanned by the top $d$ eigenvectors of $\bar{\mbf \Pi}^{(t)}$ obtained at the $t$-th iteration of the ADMM algorithm presented in the Appendix A.  
    \item \textbf{Minimum number of iterations and data points: } Let $\gamma = \sup_{\omega\in[0,1]} \frac{3\lambda_{d+1}(\omega)+\lambda_d(\omega)}{\lambda_{d+1}(\omega)+3\lambda_d(\omega)}$ and $R = \min\left\{\sqrt{\frac{d\gamma(1-\gamma^{1/2})}{2}},\frac{\sqrt{2\gamma}}{4}\right\}$. The minimum number of iterations of the ADMM, $T_{min}$, the SOAP, $\tilde{T}_{min}$, and the minimum data points, $n_{min}$, are
    \begin{equation}
    T_{\min} =\left\lceil\frac{\zeta_1^2}{(R-\zeta_2)^2}\right\rceil ;\;
    \tilde{T}_{\min} = \left\lceil\frac{4\log(R/\xi)}{\log(1/\gamma)}\right\rceil ;\; n_{\min} = C\frac{(s^*)\log(p)}{R^2}\left(\frac{\lambda_1(\omega_1)}{\lambda_d(\omega_1)-\lambda_{d+1}(\omega_1)}\right)^2,
    \end{equation}
    where $
     \zeta_1 = \frac{\tilde{C}^\prime\lambda_1(\omega_1)}{\lambda_d(\omega_1)-\lambda_{d+1}(\omega_1)}.s^*\sqrt{\frac{\log(p)}{n}}$, and $
     \zeta_2 = \frac{\tilde{C}^{\prime\prime}\sqrt{M(n)\lambda_1(\omega_1)}}{\sqrt{\lambda_d(\omega_1)-\lambda_{d+1}(\omega_1)}}\left(\frac{d.p^2\log(p)}{n}\right)^{1/4}\frac{1}{\sqrt{t}}$.
\end{itemize}

{
\subsubsection{Model Assumptions} \label{subsec:model_assumptions}
let $\mathcal{M}_d(f,d,s^*)$ be the class of $p$-dimensional stationary time series $\{X(t): t\in\mathbb{Z}\}$ satisfying the following assumptions. 

\begin{assumption}\label{as:sparsity}
    For all $\omega\in[0,1)$, the $d$-dimensional principal subspace of $\mbf f(\omega)$ is continuous as a function of $\omega$, is $s^*$-sparse and these principal subspaces share the same support. In addition, we assume that $\inf_{\omega\in[0,.5]}\lambda_d(\omega)-\lambda_{d+1}(\omega)>\delta>0$ for some constant $\delta$.
\end{assumption}

\begin{assumption}\label{as:mixing}
    There exists constants $c_1$ and $\gamma_1\geq 1$ such that for all $h\geq 1$, the $\alpha$-mixing coefficient satisfies 
        $\alpha(h) \leq \exp\{-c_1h^{\gamma_1}\}.$
\end{assumption}

\begin{assumption} \label{as:concentration}
    There exists positive constants $c_2$ and $\gamma_2$ such that for all $v\in\mathbb{S}^{p-1}(\mathbb{C})$ and all $\lambda\geq 0$, we have 
        $\mathbb{P}\left(|v^*X(t)|\geq\lambda\right)\leq 2\exp\{-c_2\lambda^{\gamma_2}\}$  for all $t\in\mathbb{Z}.$
\end{assumption}

\begin{assumption} \label{as:gamma}
    Define $\gamma$ via $\frac{1}{\gamma}=\frac{1}{\gamma_1}+\frac{2}{\gamma_2}$, where $\gamma_1$ and $\gamma_2$ are given in Assumptions (\ref{as:mixing}) and (\ref{as:concentration}). We assume that $\gamma<1$.
\end{assumption}

Assumption 1 is made to ensure consistency of the estimated principal subspaces as well as enhancing interpretability. Moreover, continuity of the principal subspace allows us to solve the optimization problem in (\ref{eq:LSPC_opt_relaxed}) sequentially, where the solution obtained are feasible with high probability for large enough $n$. Assumption 2 ensures the process we consider has short range dependence. An example of processes satisfying this assumption is the class of 1-Lipschitz functions of linear processes with absolutely regular innovations   \citep{merlevede2011bernstein}. Assumption 3 ensures the processes we consider are not heavy tailed. We leave spectrum estimation in the frequency domain for heavy tailed and long memory processes for a later investigation. Assumption 4 is required for applying the deviation inequality derived in \cite{merlevede2011bernstein}. 
}

\begin{theorem} \label{thm:consistency} 
 Let $\{X(t): t=1,\dots,n\}$ be a realization of a weakly stationary time series that follows $\mathcal{M}_d(f,d,s^*)$ with $n>n_{min}$. Let the regularization parameter in (\ref{eq:relaxed_convex_opt}) be $\varrho=C\lambda_1(\omega_1)\sqrt{\log(p)/n}$ for a sufficiently large constant $C$, and the penalty parameter $\beta$ in (\ref{eq:augmented_Lagrangian}) be $\beta=\sqrt{2}p.\varrho/\sqrt{d}$. 
 \begin{itemize}
     \item[(I)] The iterative sequence of $d$-dimensional subspace $\{\mathcal{U}^{(t)}(\omega_1)\}_{t=1}^{T}$ satisfies
    \begin{equation}
    \footnotesize
                \mathcal{D}\left[\mathcal{U}^{(t)}(\omega_1), \mathcal{U}^*(\omega_1)\right]\leq \frac{\tilde{\tilde{C}}^\prime\lambda_1(\omega_1)}{\lambda_d(\omega_1)-\lambda_{d+1}(\omega_1)}.s^*\sqrt{\frac{\log(p)}{n}}
                + \frac{\tilde{\tilde{C}}^{\prime\prime}\sqrt{M\lambda_1(\omega_1)}}{\sqrt{\lambda_d(\omega_1)-\lambda_{d+1}(\omega_1)}}\left[\frac{d.p^2\log(p)}{n}\right]^{1/4}\frac{1}{\sqrt{t}}
    \end{equation}
    with high probability, where $\tilde{\tilde{C}}^\prime $ and $\tilde{\tilde{C}}^{\prime\prime}$ are constants. 
    \item[(II)] 
    Let $\mathcal{U}^{(T+\tilde{T})}_\ell$ be the space spanned by the columns of the estimator obtained from the Algorithm LSPCA after $T \geq T_{min}$ iterations in Algorithm ADMM followed by $\tilde{T}\geq\tilde{T}_{min}$ iterations of Algorithm SOAP. By taking the sparsity parameter $\hat{s}$ in Algorithm SOAP such that  $ \hat{s}=C\max\left\{\left[\frac{4d}{(\gamma^{-1/2}-1)^2}\right],1\right\}s^*$,
    for some integer constant $C\geq 1$, and a fixed $\theta\in[0,1)$,  
    the final estimator $\hat{\mathcal{U}}_\ell=\mathcal{U}^{(T+\tilde{T})}_\ell$ satisfies 
    \begin{equation} \label{eq:D_U_U_hat}
        \mathcal{D}(\mathcal{U}^*_\ell,\hat{\mathcal{U}}_\ell) \leq C^{\prime\prime\prime}\frac{\gamma^{1/2}}{1-\gamma^{1/4}}\Delta(2\hat{s})
    \end{equation}
    with high probability, for all $\ell=1,2,\dots,n/2$, where
        \begin{equation}\label{eq:Delta_theta}
    \footnotesize
        \Delta(s) := \sup_{\omega_\ell}\frac{\sqrt{2d}\left[\left(\exp(-c_0M)\vee M\sqrt{\frac{s^*\log(p)}{n}}\right) +2\theta\lambda_1(\omega_{\ell+1})\left[\mathcal{D}(\mathcal{U}^*_\ell,\mathcal{U}^*_{\ell+1})+\alpha\right] \right]}{\frac{1}{2}\left[\lambda_d(\omega_{\ell+1})-(1-\theta)\lambda_{d+1}(\omega_{\ell+1})\right]-2\theta\lambda_1(\omega_{\ell+1})\alpha},
    \end{equation}   
$\alpha = \sup_{\omega_\ell} \frac{c_1\sqrt{2d}}{\lambda_d(\omega_\ell)-\lambda_{d+1}(\omega_\ell)}\left\{\exp\left[-c_0M\right]\vee M\sqrt{\frac{s^*\log(p)}{n}}\right\},$ and  $c_0,c_1$ are constants.
\end{itemize}
       
\end{theorem}

{
    Theorem \ref{thm:consistency} can be seen as an extension of the result in \cite{wang2014nonconvex} that is also an extension of the Davis-Kahan theorem quantifying the precision of the estimated principal subspaces and its dependence on the sample size, dimension, and the magnitude of the perturbation relative to the eigengap of the covariance type matrix in an iid sampling scheme. In particular, 
    Equations (\ref{eq:D_U_U_hat}) and (\ref{eq:Delta_theta}) reveal how convergence of the estimated principal subspaces depends on the eigengap $\lambda_d(\omega)-\lambda_{d+1}(\omega)$, sample size, and dimension. Observe that the numerator of (\ref{eq:Delta_theta}) is a combination of a term proportional to $2\theta\lambda_1(\omega_{\ell+1})\mathcal{D}(\mathcal{U}^*_\ell,\mathcal{U}^*_{\ell+1})$ and a term proportional to $\left\{\exp\left[-c_0M\right]\vee M\sqrt{\frac{s^*\log(p)}{n}}\right\}$. As the proof of Lemma B.2.1 in Appendix B indicates, the former term is an upper bound for $\|\mbf f^{(\theta)}(\omega)-\mbf f(\omega)\|_{op}$, which represents the extend of the bias introduced by the information sharing step in the proposed principal component analysis. Note that the bias vanishes as $n\to\infty$, since $\mathcal{D}(\mathcal{U}^*_\ell,\mathcal{U}^*_{\ell+1})$ converges to zero by the continuity of the principal subspaces. The later term is an upper bound on the sparse operator norm of $\mbf f_n(\omega)-\mbf f(\omega)$. The upper bound obtained holds with high probability and vanishes when $M\to\infty$ and $M\sqrt{\log(p)}/n\to 0$ as $n,p\to\infty$, with the dependence of $M$ on $n$ implicit. This indicates that, to achieve consistency in estimation of the principal subspaces, the parameter $M$, should grow with $n$ and $p$ at the rate such that $M\sqrt{\log(p)}/n\to 0$. Finally, Theorem \ref{thm:consistency} guaranties that for a sufficiently large $n,p$ and $M$, $\mathcal{D}(\mbf\Pi_\ell,\hat{\mbf\Pi}_\ell) <\epsilon/3$ for a given $\epsilon$. This, along with the smoothness of principal subspaces as a function of frequency, implies $\|\hat{\mbf\Pi}_\ell-\hat{\mbf\Pi}_{\ell+1}\| < \epsilon$ with high probability. This confirms that, for sufficiently large $n,p$, and $M$ the estimates obtained by the LSPCA is feasible for (\ref{eq:LSPC_opt0}), with high probability.  
}

\section{Tuning Parameter Selection}\label{sec:model_selection}
In this section we outline a procedure for the selection of four parameters:  the dimension of the principal subspaces $d$, the sparsity parameter $\hat{s}$,  the localization parameter $\eta$, and  the smoothing parameter $\theta$. A detailed description is provided in Appendix A. Given the complex nature of the problem that makes the empirical joint selection of the parameters infeasible, we propose the selection of each parameter individually, conditional on the other parameters, and iteratively. 

For the dimension of the principal subspaces $d$, we inspect the scree plot or equivalently the plot of the proportion of variance explained. For the localization parameter $\eta$,  we propose to use information criteria based on the log-Whittle likelihood. More precisely, let $h_{\ell}$ be the maximum total power in a d-dimensional subspace at frequency $\ell$ defined in Section \ref{subsec:localization}, $\left\{h_{(1)},\dots,h_{(n/2)}\right\}$ be their order statistics, $\check{L} = \{\ell_{(1)},\dots,\ell_{(n/2)} \}$ be the corresponding indices of the Fourier frequencies of these order statistics, and $\check{L}_\eta = \{\ell_{(1)},\dots,\ell_{(\eta)}\}\subset \check{L}$ be the the Fourier frequency indices of the top $\eta$ order statistics. The log-Whittle likelihood is estimated by $\log(\mathcal{L}) = -\sum_{\ell=1}^{n/2} \{p\log{\pi} + \log{|\hat{\mbf G}_\ell|} + [d_X(\omega_\ell)][\hat{\mbf G}_\ell]^{-1}[d_X(\omega_\ell)]^\dagger\}$, 
    where $\hat{\mbf G}_{\ell} =  \mathbb{I}\{\ell \in \check{L}_\eta\}\hat{\mbf f}_\vartheta(\omega_\ell) + \hat{\mbf\Sigma}$, $\hat{\mbf f}_\vartheta(\omega_\ell) = {\mbf f}_n(\omega_\ell)\hat{U}_1(\omega_\ell)\hat{U}_1(\omega_\ell)^\dagger + \dots + {\mbf f}_n(\omega_\ell)\hat{U}_d(\omega_\ell)\hat{U}_d(\omega_\ell)^\dagger$, $\hat{U}_j(\omega_\ell)$ are obtained from the LSPCA algorithm, $\hat{\mbf \Sigma} = \sum_{\ell\in \check{L}/\check{L}_\eta}d_{X}(\omega_\ell)d_{X}(\omega_\ell)^\dagger/\left(|\check{L}| - |\check{L}_\eta|\right)$, and $d_{X}(\omega_\ell)$ is the discrete Fourier transform of the data at $\omega_\ell$. We use this to define standard information criteria for $\eta$ including $AIC = -2\log(\mathcal{L}) + 2\eta$, $AICc = -2\log(\mathcal{L}) + 2\eta + \frac{2\eta^2 + 2\eta}{n-\eta-1}$, and $BIC = -2\log(\mathcal{L}) + \log(n)\eta$. The localization parameter $\eta$ is selected to minimize the information criteria.  For both the sparsity parameter $s$ and the smoothing parameter $\theta$, we propose to use $k$-folds cross validation, with the Mahalanobis distance to evaluate the performance of fitted model in the validation step. For time series with short length, one might consider alternative methods such as information criterion for selecting $s$ and $\theta$.

\section{Simulation Studies}\label{sec:Simulation}

\subsection{Setting}
In this section, we explore, numerically, the performance of the LSPCA algorithm, effect of smoothing, and performance of the parameter selection procedures and compare the results with the classical frequency domain PCA developed in \cite{brillinger2001time}.  
Let $a(t)$ be a linear filter with frequency response $\mathbf{I}(\Omega)$, where $\mathbf{I}(\Omega)$ is the indicator function of the set $\Omega=[.05,.25]$.  We consider processes $X(t) = \left[X_1(t),\dots,X_p(t)\right]^T, \; t=1,\dots,n$ that are constructed from independent processes $Y_1,\dots,Y_5$ such that $\{Y_k(t):t\in\mathbb{Z}\}\sim AR(4)$. More precisely, let $\star$ represents the convolution operator, we define $X_1(t) = a(t)\star (1/c)Y_1(t)$, $X_j(t) = a(t)\star[\pi_j^{(x)}X_1(t) + (1/c)Y_j(t)]+  W_{j1}(t), j=2,\dots,5$ with $\pi_{2}^{(x)} = 2.2, \pi_{3}^{(x)} = 1.2, \pi_{4}^{(x)} = 1.25, \pi_{5}^{(x)} = 2.25$ and $X_j(t) \sim WN(0,1),\; j=6,\dots,p$, where 
$Y_j(t) = \sum_{k=1}^{4}\pi_{jk}Y_{j}(t-k)+W_{j0}(t), W_{j0}(t)\sim N(0,1/4), j=1,\dots,5$ such that 
$\pi_{j1}=\alpha_1+\phi_{j1}, \pi_{j2}=\alpha_2-\alpha_1\phi_{j1}+\phi_2, \pi_{j3}=-(\phi_{j1}\alpha_2+\phi_2\alpha_1), \pi_{j4}=-\phi_{j2}\alpha_2$, $\alpha_1=1/20, \alpha_2=-1/1.15, \phi_2=-0.75, \phi_{11}=1.5, \phi_{21}=1.55, \phi_{31}=1.45, \phi_{41}=1.65, \phi_{51}=1.35$, and $W_{j0}(t), j=1,\dots,5$ are independent white noise processes with mean zero and variance $1$ and are independent of $W_{j1}(t), j=2,\dots,5$.

\begin{figure}[t] 
\begin{center}
\vspace{-.7in}
\begin{tabular}{cc}
\includegraphics[width=2.5in, height=3.4in, bb = 50 50 650 650]{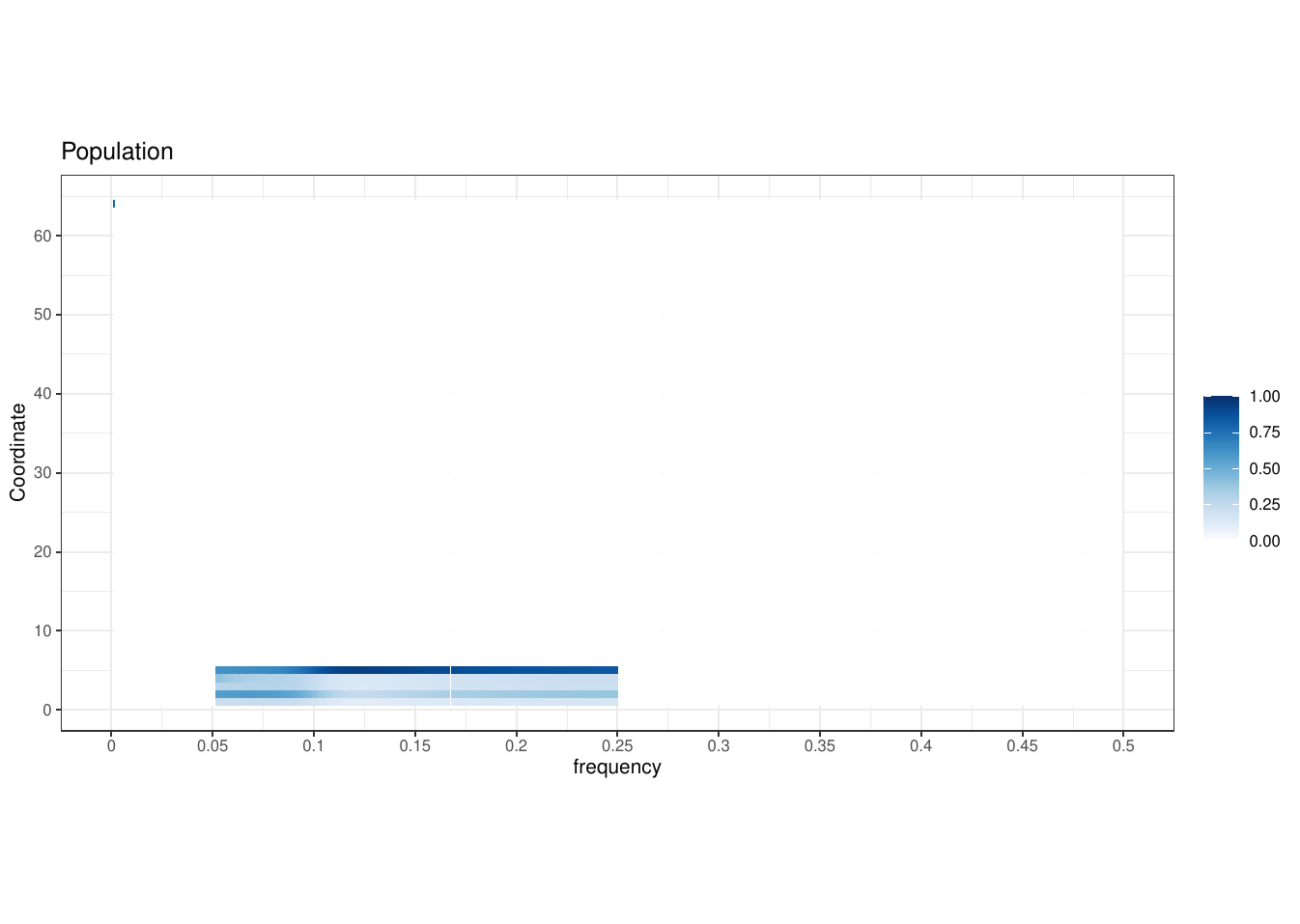} 
& \includegraphics[width=2.5in, height=3.4in, bb = 50 50 650 650]{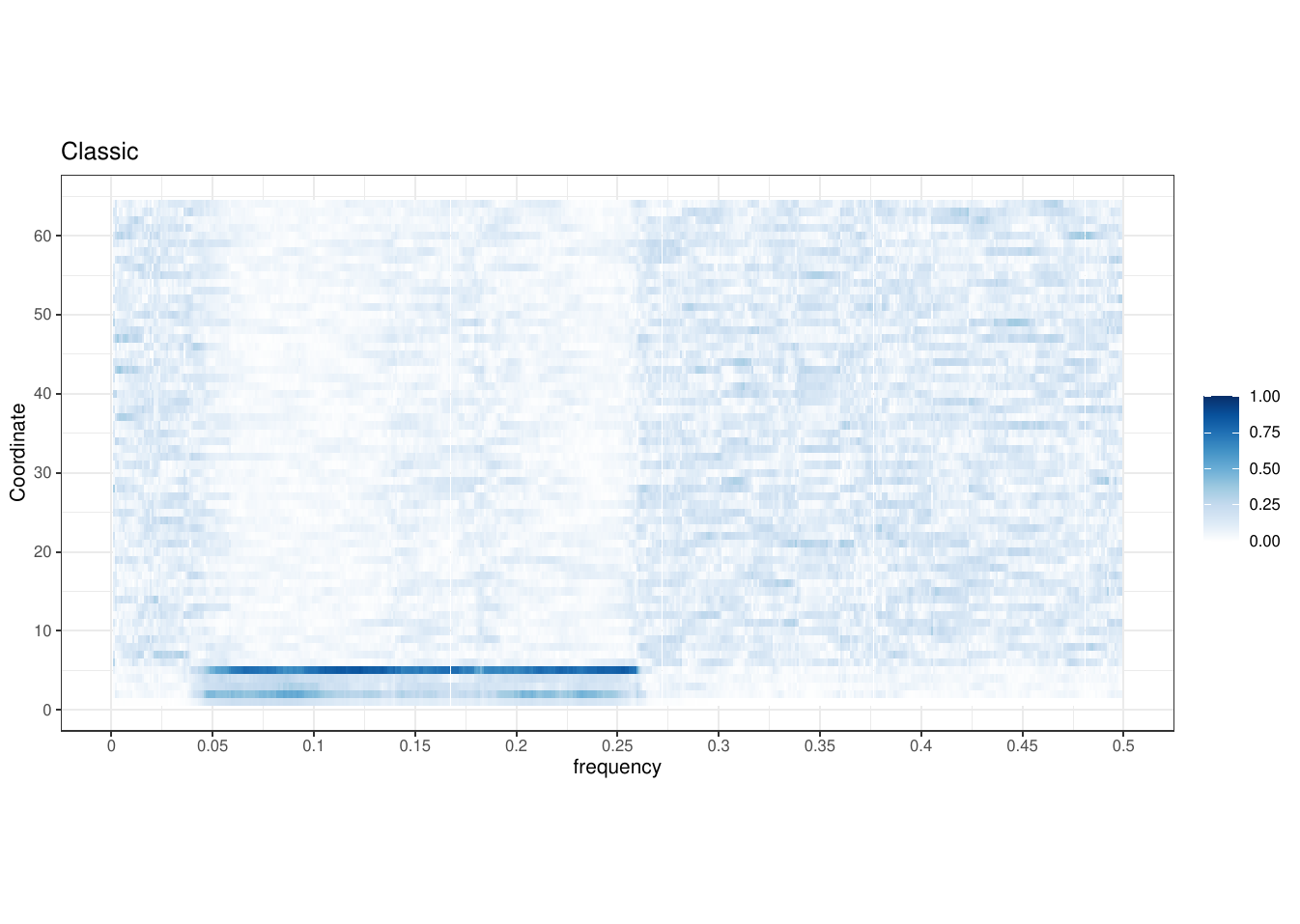}\vspace{-100pt} \\
\includegraphics[width=2.5in, height=3.4in, bb = 50 50 650 650]{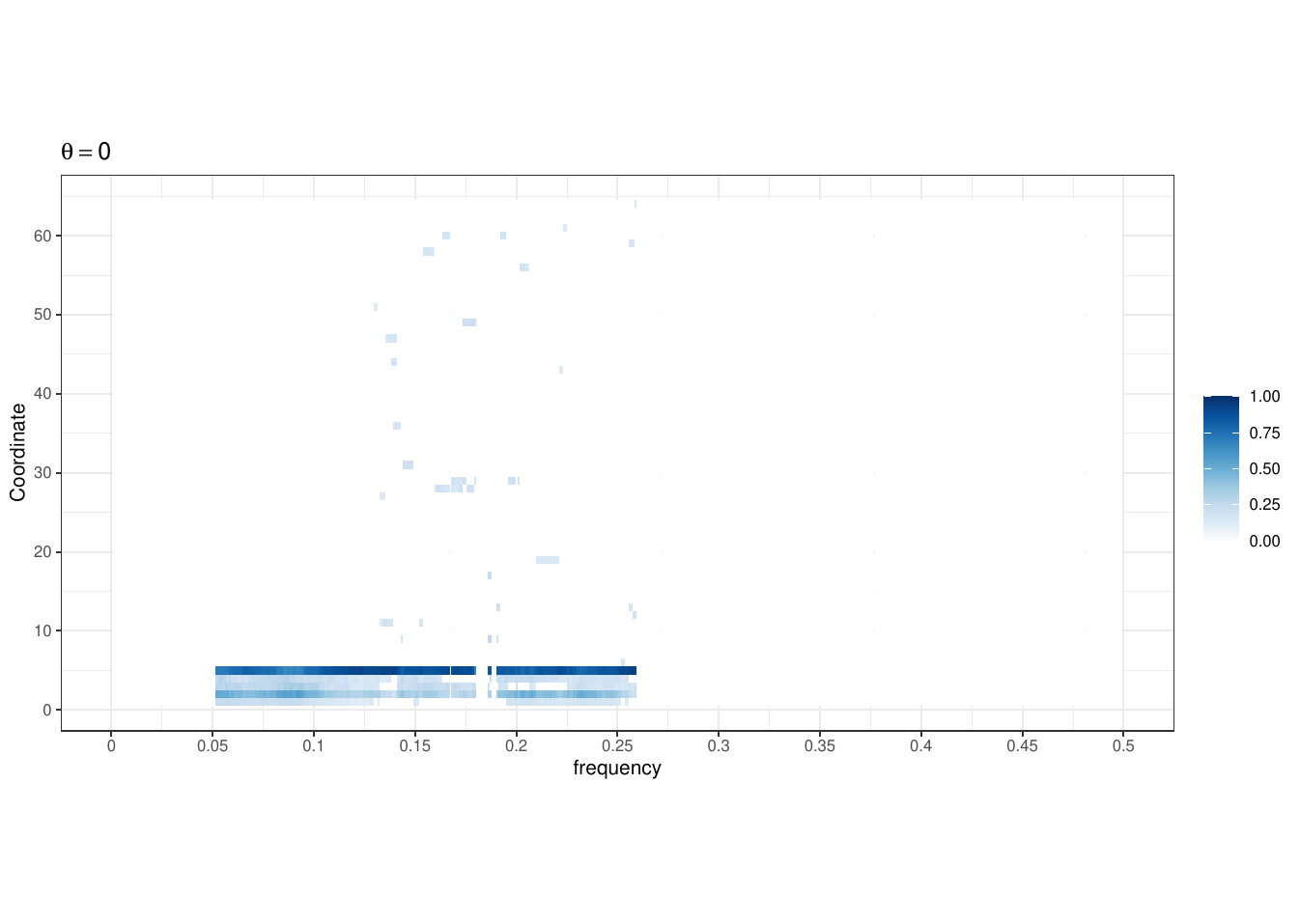}
& \includegraphics[width=2.5in, height=3.4in, bb = 50 50 650 650]{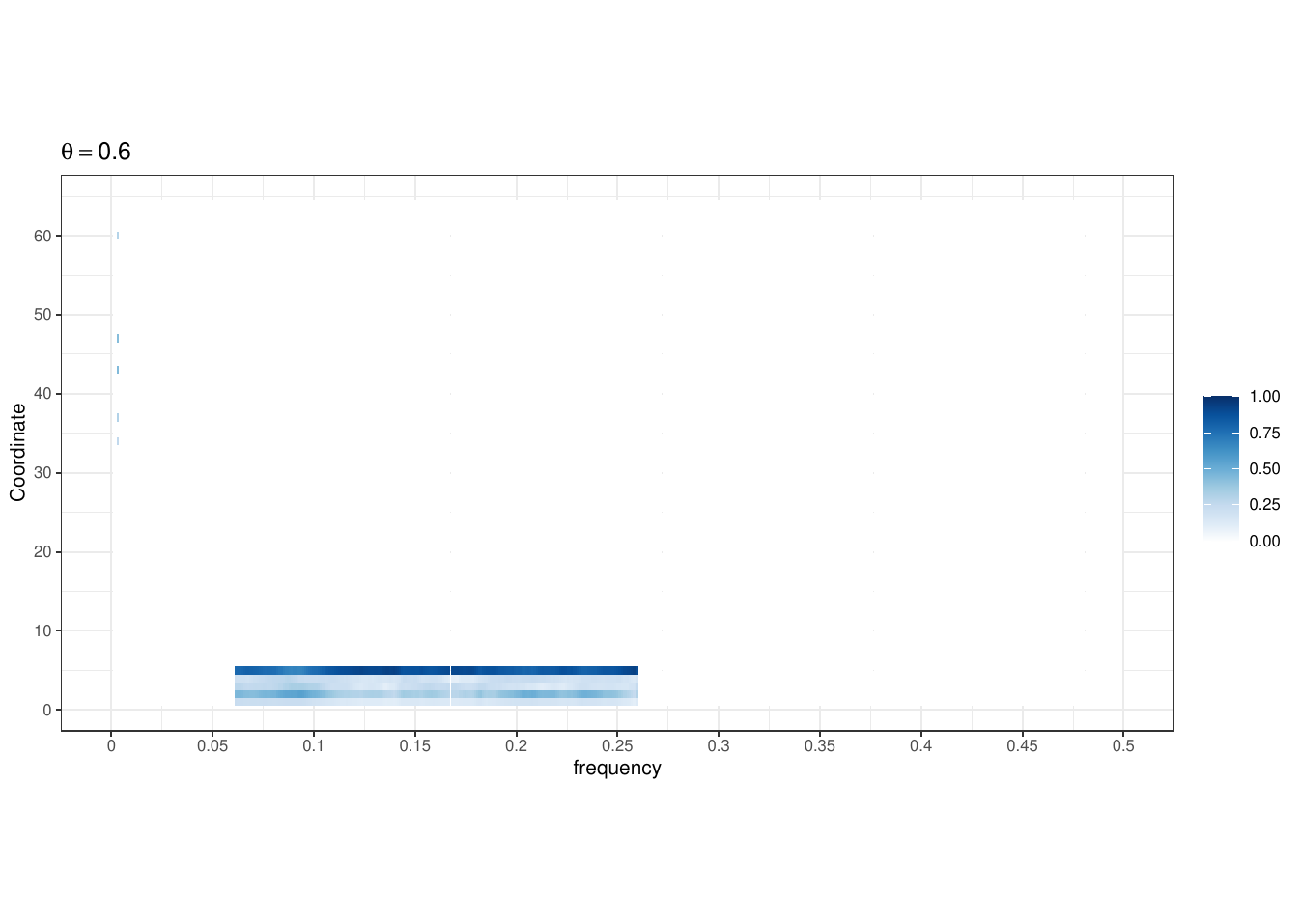}
\end{tabular}
         \caption{\textbf{Top left} panel represents the population leading eigenvector; \textbf{Top right} panel represents the classical estimate of the leading eigenvector; \textbf{Bottom left} panel represents the sparse estimate of the leading eigenvector with $\theta=0$; \textbf{Bottom right} panel represents the sparse estimate of the leading eigenvector with $\theta=0.6$.}
         \label{fig:evec_trajectories_2D_truncated}
\end{center}
\end{figure}

The processes have $d=1$ dimensional principle subspaces. The parameter $c$ controls the eigengap such that increases in $c$ reduces the power in the first $5$ coordinates, which weakens the signal strength and decreases the eigengap. The modulus of its leading eigenvector is displayed in the top left panel of Figure \ref{fig:evec_trajectories_2D_truncated} for $c=3$. The accuracy of the estimated principal subspace at any frequency, say $\omega$, depends on various factors including the dimension ($p$), sample size ($n$), and the eigengap. We considered two values of the dimension $p=64, 128$,  three sample sizes $n=1024, 2048, 4096$, and two values of signal strength $c = 1, 3$. One hundred realizations of the process for all combinations of $n$, $p$, and $c$ were generated, and the first eigenvector of the spectral density matrix at each frequency were estimated using the LSPCA algorithm and the classical method of \cite{brillinger2001time}. We evaluated the performance of the estimators using mean squared estimation error (MSEE) defined as $\sum_{k=1}^{100}\mathcal{D}(\mathcal{U}-\hat{\mathcal{U}}^{(k)})/100$, where $\mathcal{U}$ is the true 1-dimensional principal subspace and $\hat{\mathcal{U}}^{(k)}$ is the estimated one obtained from the $k$-th run of the simulation, and $\mathcal{D}$ is the distance defined in Section \ref{subsec:theoretical_analysis}.

\begin{figure}[t] 
\centering
\begin{tabular}{c}
 \includegraphics[scale = .55]{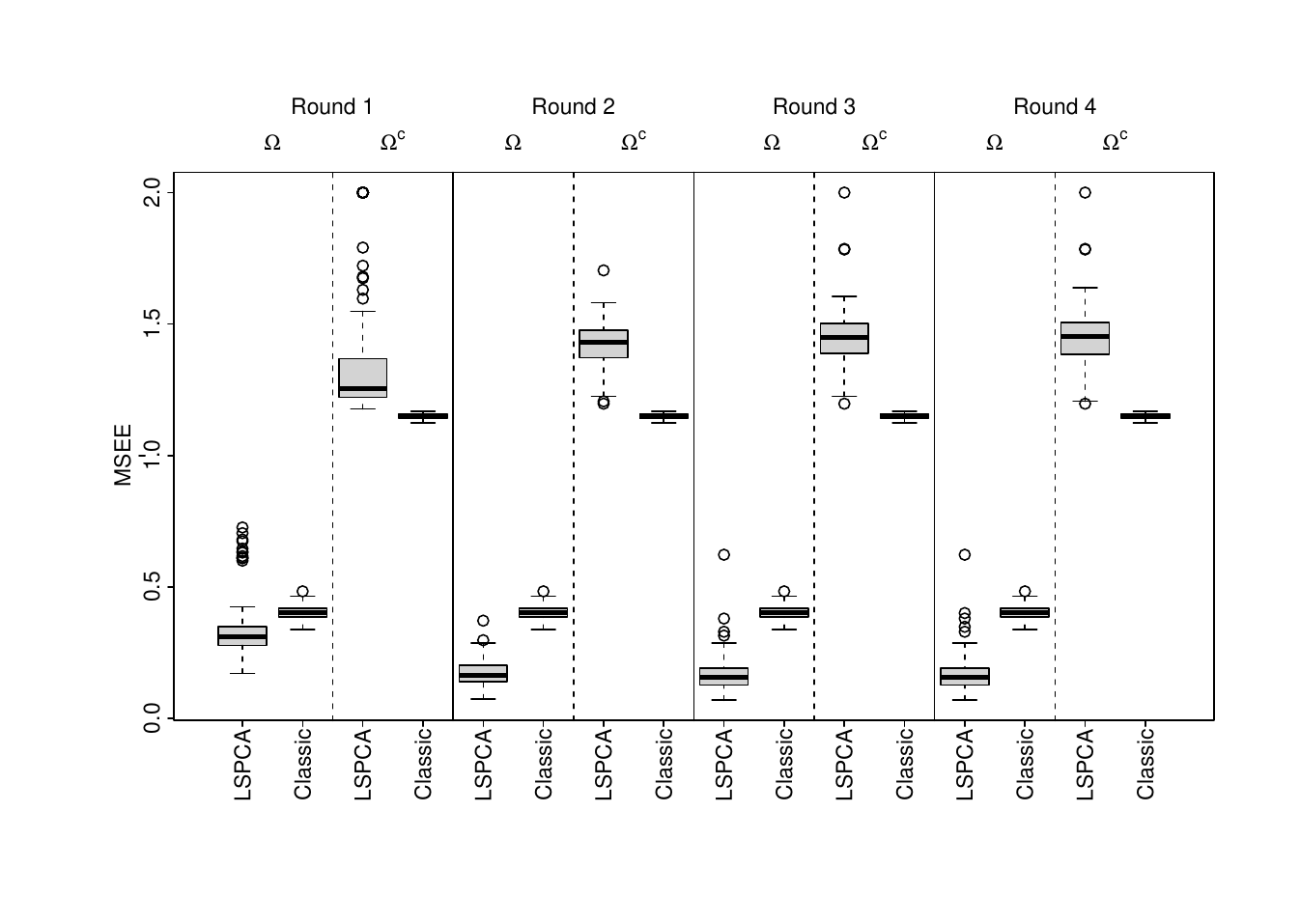}
\end{tabular}
         \caption{Side by side box plots of the MSEE of the LSPCA and the classical PCA for 4 times iteration of the parameter selection and estimation of the 1-dimensional principal subspaces over $\Omega$ and $\Omega^c$.}
         \label{fig:5_iters}
\end{figure}

\subsection{Results}
We present three sets of results here; additional results related to computation time, sparsity, and localization are provided in Appendix C.  The first set of results are illustrative results from one realization with $p=64$, $n=1024$ and $c=3$.  The top right panel of Figure \ref{fig:evec_trajectories_2D_truncated} displays the estimated modulus of the first principal component estimated from the classical estimator of \cite{brillinger2001time}. The high variability of the classical estimator is illustrated, which inhibits the separation of signal within the five channels in the support $\Omega$ from noise.  The bottom panels of Figure \ref{fig:evec_trajectories_2D_truncated} provide a visual illustration of the role of the sharing of information across principle subspaces at adjacent frequencies via $\theta$. Although the results of Section \ref{subsec:theoretical_analysis} indicate consistency even with $\theta = 0$, we see that without sharing information across frequency in the lower left panel where $\theta = 0$ at certain frequencies, other coordinates are selected rather than the five that have true signal.  The result is that supported subspaces cannot be interpreted as a continuous bands.  The estimate under $\theta = 0.6$ is provided in the lower right panel, where it can be seen that the sharing of information not only improves estimation in finite samples, but provides essential interpretation where frequency is localized into bands.  

\begin{figure}[t] 
\centering
\begin{tabular}{c}
 \includegraphics[scale = .5]{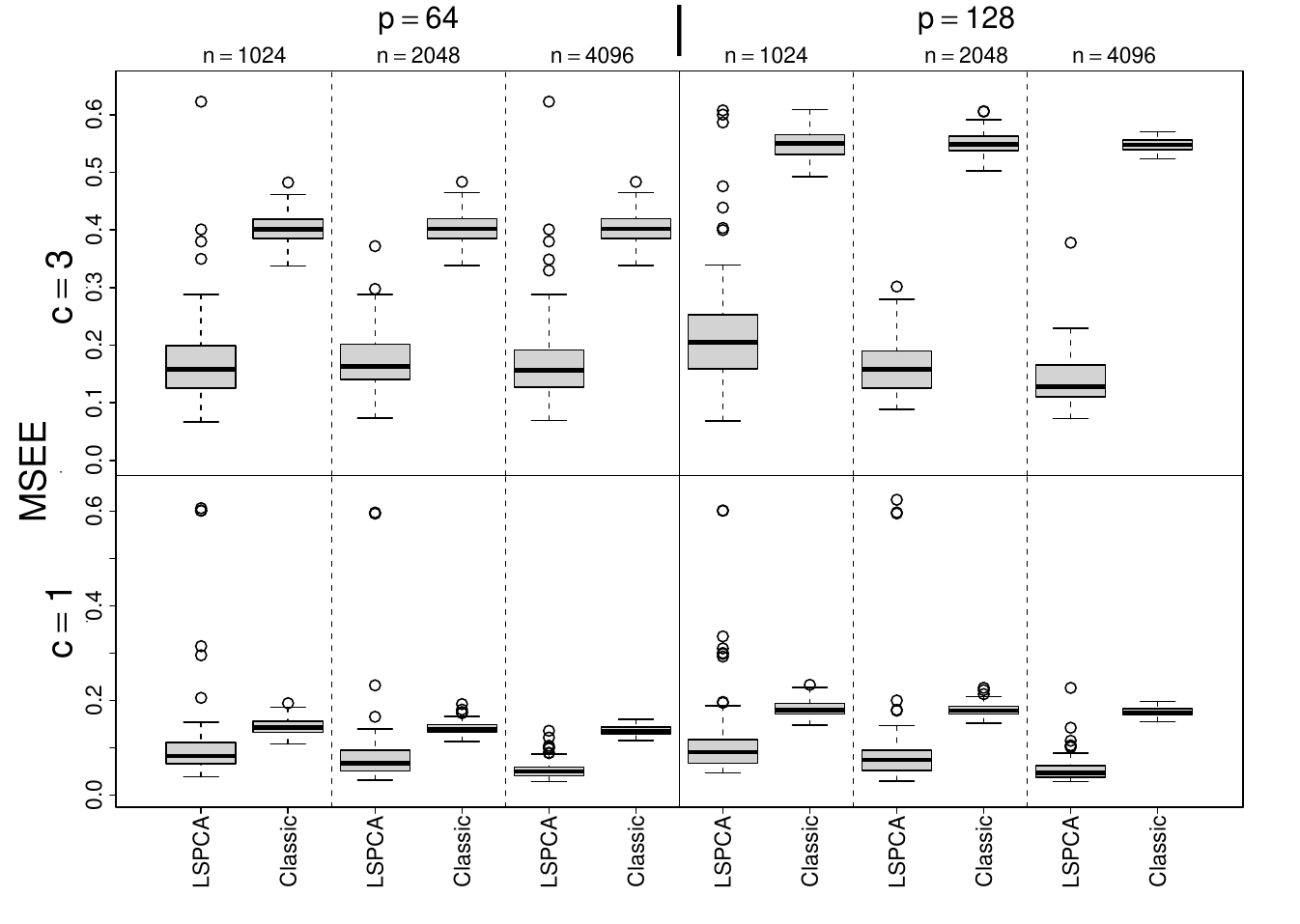}
\end{tabular}
         \caption{Side by side box plots of the MSEE of the LSPCA illustrating the effect of increasing sample size, dimension, and signal strength on estimation of the 1-dimensional principal subspaces over $\Omega$ and in comparison with the the classical principal subspace estimation.}
         \label{fig:increasing_n_&_p}
\end{figure}

The second set of results investigates the stability of the iterative procedure for selecting tuning parameters and the relative estimation at frequencies within the support $\Omega$ and outside the support $\Omega^c = [0,.5] \backslash \Omega$ for $p=64, n=1024, c=3$. Figure \ref{fig:5_iters} illustrates side by side boxplots of the mean estimation error of the LSPCA and the classical PCA over $\Omega$ and $\Omega^c$ after 1-4 iterations of the tuning parameter selection procedure.  Plots illustrates that, after two iterations of the parameter selection and estimation procedure, the estimated principal subspaces stabilizes, suggesting performing the iterative scheme in LSPCA for two iterations. In addition, all plots confirm that the LSPCA improves estimation of the underlying principal subspaces in $\Omega$ compared to the classical approach. Moreover, over $\Omega^c$, where the underlying principal subspaces are not sparse and signal strength is low, both estimates obtained by the classical PCA and the LSPCA do not perform well.  It can be seen that imposing sparsity through LSPCA  improves estimation within regions of scientific interest $\Omega$ while increasing error in regions that are not of scientific interest $\Omega^c$.

\begin{figure}[t]
    \centering
    \includegraphics[width=6.5in, height=3.5in]{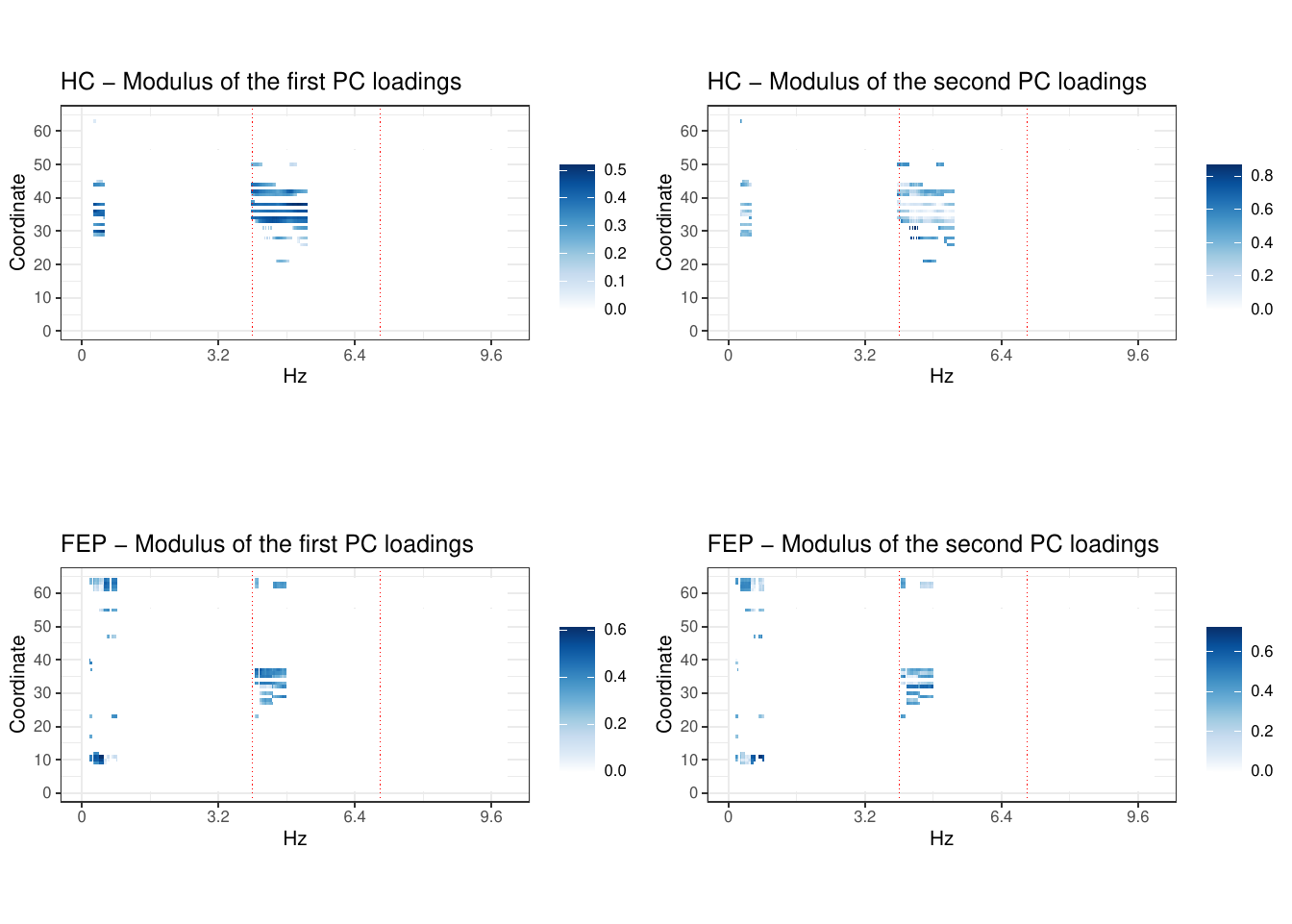}
    \caption{\textbf{Top left panel} Modulus of the first PC loadings of HC subject; \textbf{Top right pane}: Modulus of the second PC loadings of HC subject. \textbf{Bottom left panel} Modulus of the first PC loadings of FEP subject; \textbf{Bottom right pane}: Modulus of the second PC loadings of FEP subject. Vertical lines denote the boundaries of the traditional delta and theta bands.}
    \label{fig:Subjects_evecs}
\end{figure}

The third set of results investigates the relative effects of sample size $n$, dimension $p$ and signal strength $c$ for estimation within $\Omega$ as summarized in Figure \ref{fig:increasing_n_&_p}.  We found that LSPCA outperforms classical PCA in all settings, and confirms the results in Section \ref{subsec:theoretical_analysis} in that estimation error improves for higher sample sizes $n$, smaller dimensions $p$ and stronger larger eigengaps/smaller $c$.


\section{Data Analysis} \label{sec:Data_Analysis}
Evidence suggests that electrophysiological activity at different frequencies and locations of the brain can be biomarkers for schizophrenia \citep{renaldi2019predicting, zhang2021}. To illustrate the use of LSPCA to obtain interpretable frequency-channel analyses, we apply it in separate analyses of 64-channel EEG recording from two individuals. One is from a patient who is experiencing an episode of psychosis for the first time (FEP), which can be a precursor to the eventual development of schizophrenia, and has been emitted to a psychiatric emergency department. The second is from a healthy control with no history of mental illness (HC). During the recording, participants sat in a chair and relaxed with their eyes open. Data were recorded using a f10-10 system and were initially sampled at a rate of 250 Hz for one minute. Pre-processing consisted of down-sampling to 64 Hz and filtering using a 1 Hz high-pass filter and 58 Hz low-pass filter; removal of segments with large artifacts such as muscle activity or movements by a trained EEG data manager; and further removal of subtle artifacts such as ocular movement and cardiac signals via independent component analysis \citep{delorme2004}.

We applied LSPCA to understand personal electrophysiological activity by analyzing each subject's data separately.  We made exploratory comparisons and discuss future formal group analyses in Section \ref{sec:Discussion}.  The parameter selection procedure selected a sparsity level of $\hat{s}=8$ for both subjects, smoothing parameters of $\theta=0.2$ and $\theta=0.6$,  and localization parameters of $\eta=41$ and $\eta=52$ for the FEP and HC participants, respectively. Inspection of scree plots at all frequencies suggests that $d=2$. Two aspects of the results are presented here; additional explorations, including analyses of real and complex structures and of coherence, are provided in Appendix D. 

\begin{figure}[t] 
\begin{center}
\includegraphics[scale = .48]{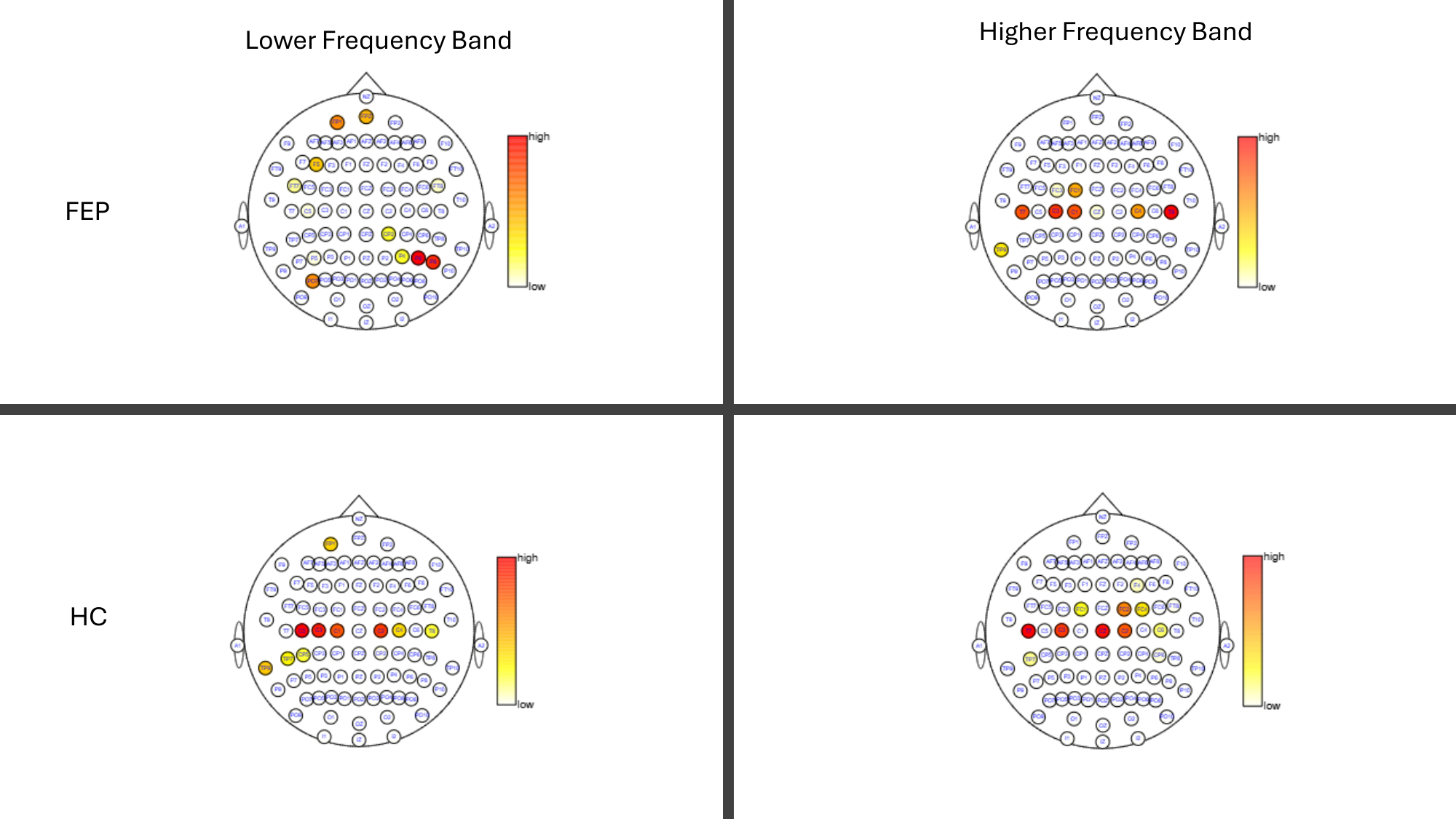}
\caption{Active channels in the FEP (top row) and HC (bottom row) participants for the lower frequency band (left column) and higher frequency band (right column).}
         \label{fig:eegcap}
\end{center}
\end{figure}

First, we investigated the modulus of PC loadings of the $d=2$ components as functions of coordinate/channel and frequency for the FEP participant and the healthy control in Figure \ref{fig:Subjects_evecs}.
Principal subspaces for both participants are localized within the union of a band of very-low frequencies that is contained within the traditional delta band of frequencies less than 4 Hz, and with a band that is contained in the traditional theta band of frequencies between 4 - 8 Hz. As power within the delta band is characteristic of unconscious processes and elevated during rest, and power within the theta band is involved in cognitive processes such as attention control that is elevated when eyes are open, this localization is not unexpected.  However, as opposed to collapsing power within the historically defined delta and theta bands, the boundaries of which are displayed in Figure \ref{fig:Subjects_evecs}, the data-driven LSPCA identified narrower, more parsimonious, person-specific bands. 

Next, we explored the spatial localization of power within these identified bands. Figure \ref{fig:eegcap} displays the diagonal elements of $\sum_{\omega\in B} \hat{\mbf f}_{\vartheta}(\omega)$ where B is the localized frequency band for both subjects. We see that power in both the delta and theta band are concentrated in the central regions of the HC; theta power in the FEP participant is also concentrated in the central region.  However, delta power for the FEP participant is primarily located in the frontopolar and parietal regions. This difference in distribution of delta brain activity between the FEP and HC participants is consistent with the findings of \cite{renaldi2019predicting}, who reported significantly higher delta power in the frontal and posterior regions in FEP compared to HC.

\section{Discussion} \label{sec:Discussion}
This article introduced what is, to the best of our knowledge, the first approach to conducting a PCA on a high-dimensional stationary time series whose principal subspace is sparse among variates, localized within frequency, and smooth as a function of frequency.  The method is by no means exhaustive and can potentially be extended to more complex scenarios.  The first of these is to nonstationary time series. Although the developed LSPCA routine could be applied directly using quadratic time-frequency transformations such as the local Fourier periodogram or SLeX periodogram, it is not yet obvious how to impose sparsity, frequency localization, or smoothness on the time-frequency subspaces that accounts for temporal ordering and information. A second extension is to the replicated time series setting in which a joint analysis is conducted on data where multivariate time series are observed for multiple subjects. As opposed to the analysis presented in Section \ref{sec:Data_Analysis}, where separate PCAs were conducted individually for two separate subjects, a PCA for replicated time series will find optimal eigenspaces for describing mutual and subject specific information.  The proposed LSPCA offers a regularized estimation approach. 
One might desire a confidence-based procedure that provides inference with regards to included frequency bands and retained channels.  Although the excursion set method that has been used to obtain inference for spatial clusters in image data appears to provide a natural solution for conducting inference with regards to frequency localization \citep{maullinsapey2023}, how to do this while imposing sparsity and extract the low-dimensional principal subspace could prove to be challenging.  Lastly, PCA is just one of many popular tools for extracting low-dimensional structures in time series data, each with different goals. Another popular tool, dynamic factor modeling, finds parsimonious dynamic factors that represent comovements of the variates. It has an elegant formulation that involves the eigendecomposition of spectral matrices and has favorable statistical properties when eigenvalues diverge \citep{forni2000}. Future work can investigate the extension of LSPCA to dynamic factor models when largest eigenvalues are bounded.

\begin{center}
{\large\bf CODE}
\end{center}
\appendix

 An R package to implement LSPCA and reproduce results presented in the manuscript is provided at \href{github.com/jamnamdari/LSPCA}{github.com/jamnamdari/LSPCA}. 

\newpage

\begin{center}
{\large\bf APPENDICES}
\end{center}
\appendix

\renewcommand{\thetable}{S\arabic{table}}
\renewcommand{\thefigure}{S\arabic{figure}}

\section{Algorithms}\label{Appendix:Algorithms} 
In Section \ref{subsec:LSPCA} the ADMM, SOAP, and the LSPCA algorithms are illustrated. Detailed description of the localization, sparsity, and smoothing parameter selection are provided in Sections \ref{subsec:localization_par}, \ref{subsec:sparsity_par}, and \ref{subsec:smoothing_par}, respectively. 
\subsection{LSPCA} \label{subsec:LSPCA}
\begin{algorithm}[H]
\caption{ADMM : Solving (\ref{eq:relaxed_convex_opt})} \label{alg:ADMM}

\SetKwInput{KwInput}{Input}                
\SetKwInput{KwOutput}{Output}              
\SetKwInOut{Parameter}{Parameters}
\SetKwInOut{Initialization}{Initialization}
\DontPrintSemicolon
\SetAlgoLined
\RestyleAlgo{boxruled}
  
  \KwInput{Spectral density estimator $\hat{\mbf\Sigma}$}
  \KwOutput{$\mbf U^{init}$}
\Parameter{Regularization parameter $\rho>0$, penalty parameter $\beta>0$, maximum number of iterations $T$}
\Initialization{$\mbf\Pi^{(0)} \leftarrow 0$, $\mbf\Phi^{(0)} \leftarrow 0$, $\mbf\Theta^{(0)} \leftarrow 0$ \;
}
\BlankLine
    \For{$t=0,\dots,T-1$}{
    $\mbf\Pi^{t+1} \leftarrow \mbox{argmin } \{\mathcal{L}(\mbf\Pi,\mbf\Phi^{(t)},\mbf\Theta^{(t)} +\beta/2\|\mbf\Pi-\mbf\Phi^{(t)}\|_F^2 \;|\; \mbf\Pi\in\mathcal{A}\}$\;
    $\mbf\Phi^{t+1} \leftarrow \mbox{argmin } \{\mathcal{L}(\mbf\Pi^{(t+1)},\mbf\Phi,\Theta^{(t)}) +\beta/2\|\mbf\Pi^{(t+1)}-\mbf\Phi\|_F^2 \;|\; \mbf\Phi\in\mathbb{R}^{p\times p}\}$\;
    $\mbf\Theta^{(t+1)} \leftarrow \mbf\Theta^{(t)} - \beta(\mbf\Pi^{(t+1)}-\mbf\Phi^{(t+1)})$\;
    }
    $\bar{\mbf\Pi}^{(T)} = \frac{1}{T}\sum_{t=0}^{T}\mbf\Pi^{(t)}$ \;
    Set the columns of $U^{init}$ to be the top $q$ leading eigenvectors of $\bar{\mbf \Pi}^{(T)}$ \;

  \KwOutput{$\mbf U^{init}$}
\end{algorithm}

\begin{algorithm}[H]
\caption{Projection} 

\SetKwInput{KwInput}{Input}                
\SetKwInput{KwOutput}{Output}              
\DontPrintSemicolon
\SetAlgoLined
\RestyleAlgo{boxruled}
  
  \KwInput{$\mbf \Phi^{(t)},\mbf \Theta^{(t)},\hat{\mbf \Sigma},\beta$}
  \KwOutput{$\mbf \Pi^{(t+1)}$}
\BlankLine
  \SetKwFunction{FMain}{Projection}
  \SetKwProg{Fn}{Function}{($\mbf \Phi^{(t)},\mbf \Theta^{(t)},\hat{\mbf \Sigma},\beta$)}{\KwRet}
  \Fn{\FMain}{
     \textbf{Eigenvalue Decomposition:} $\mbf A\mbf \Lambda \mbf Q^* \leftarrow \mbf \Phi^{(t)} + \mbf \Theta^{(t)}/\beta+\hat{\mbf \Sigma}/\beta$\;
     $(v_1^\prime,\dots,v_p^\prime)=\mbox{argmin}\{\|v-\text{diag}(\mbf \Lambda^{(t)})\|_2^2 \;|\; v\in\mathbb{R}^p, \sum_jv_j=q,v_j\in[0,1] \mbox{ for all } j \}$\;
     $\mbf \Pi^{(t+1)} \leftarrow \mbf Q\mbox{diag}\{v_1^\prime,\dots,v_p^\prime\}\mbf Q^*$\;
    \KwRet $\mbf \Pi^{(t+1)}$
  } 
\end{algorithm}

\begin{algorithm}[H]
\caption{Soft-Thresholding} 

\SetKwInput{KwInput}{Input}                
\SetKwInput{KwOutput}{Output}              
\DontPrintSemicolon
\SetAlgoLined
\RestyleAlgo{boxruled}
  
  \KwInput{$\mbf \Pi^{(t+1)},\mbf \Theta^{(t)},\rho,\beta$}
  \KwOutput{$\mbf \Phi^{(t+1)}$}

\BlankLine
    \For{$i,j\in\{1,\dots,d\}$}{
    $\mbf \Phi^{(t+1)}\leftarrow \begin{cases}
            0 & \mbox{if } |\mbf \Pi_{i,j}^{(t+1)}-\mbf \Theta_{i,j}^{(t)}/\beta|\leq \rho/\beta \\
            \substack{\mbox{sign}\left(\mbf \Pi_{i,j}^{(t+1)}-\mbf \Theta_{i,j}^{(t)}/\beta\right) \\ \times \left(|\mbf \Pi_{i,j}^{(t+1)}-\mbf \Theta_{i,j}^{(t)}/\beta|-\rho/\beta\right) } & \mbox{if } |\mbf \Pi_{i,j}^{(t+1)}-\mbf \Theta_{i,j}^{(t)}/\beta| > \rho/\beta
    \end{cases}$
    }
  \KwOutput{$\mbf \Phi^{(t+1)}$}
\end{algorithm}
\BlankLine

In the SOAP, the orthogonal iteration method is followed by a truncation step to enforce row-sparsity and further followed by taking another re-normalization step to enforce orthogonality. 

\begin{algorithm}[H]
\caption{SOAP} \label{alg:SOAP}

\SetKwInput{KwInput}{Input}                
\SetKwInput{KwOutput}{Output}              
\SetKwInOut{Parameter}{Parameters}
\SetKwInOut{Initialization}{Initialization}
\DontPrintSemicolon
\SetAlgoLined
\RestyleAlgo{boxruled}
  
  \KwInput{$\mbf {f}_n(\omega)$, initialization $U^{init}$}
\Parameter{Sparsity parameter $\hat{s}$, Maximum number of iteration $\Tilde{T}$}
\Initialization{$\tilde{\mbf U}^{(T+1)}\leftarrow\mbox{Truncate}(\mbf U^{init},\hat{s})$, $\mbf U^{(T+1)}$, $\mathsf{\mbf R}_2^{(T+1)}\leftarrow\mbox{Thin-QR}(\Tilde{\mbf U}^{(T+1)})$}
\BlankLine
    \For{$t=T+1\dots,T+\Tilde{T}-1$}{
    $\quad \Tilde{\mbf V}^{(t+1)}\leftarrow{\mbf f_n}(\omega)\mbf U^{(t)}$\;
    $\quad \mbf V^{(t+1)},\mathsf{\mbf R}_1^{(t+1)}\leftarrow\mbox{Thin-QR}(\Tilde{\mbf V}^{(t+1)})$\;
    $\quad \Tilde{\mbf U}^{(t+1)}\leftarrow\mbox{Truncate}(\mbf V^{(t+1)},\hat{s})$ \label{eq:Truncate}\;
    $\quad \mbf U^{(t+1)},\mathsf{\mbf R}_2^{(t+1)}\leftarrow\mbox{Thin-QR}(\Tilde{\mbf U})^{(t+1)}$\;
    }
  \KwOutput{$\mbf U^{T+\tilde{T}}$}
\end{algorithm}

\begin{algorithm}[H]
\caption{$\Tilde{\mbf U}^{(t+1)}\leftarrow\mbox{Truncate}(\mbf V^{(t+1)},\hat{s})$ : \\ The Truncate function, where $V_{i,.}$ denotes the $i$-th row vector of $\mbf V$.} 

\SetKwInput{KwInput}{Input}                
\SetKwInput{KwOutput}{Output}              
\DontPrintSemicolon
\SetAlgoLined
\RestyleAlgo{boxruled}
  
\BlankLine
  \SetKwFunction{FMain}{Truncate}
  \SetKwProg{Fn}{Function}{($\mbf V^{(t+1)},\hat{s}$)}{\KwRet}
  \Fn{\FMain}{
     \textbf{Row Sorting:} $\mathcal{I}_{\hat{s}}\leftarrow\mbox{The set of row index $i$`s corresponding to the top $\hat{s}$ largest $\|V_{i,.}^{(t+1)}\|_2$'s }$\;
      \For{$i\in\{1,\dots,p\}$}{
        $\quad \Tilde{U}_{i,.}^{(t+1)}\leftarrow\mathbbm{1}[i\in\mathcal{I}_{\hat{s}}]V_{i,.}^{(t+1)}$\;
    }
    \KwRet $\Tilde{\mbf U}^{(t+1)}$
  } 
\end{algorithm}
\BlankLine

The combined algorithm is as follows
\BlankLine

\begin{algorithm}[H]
\caption{LSPCA} \label{alg:LSPCA}

\SetKwInput{KwInput}{Input}                
\SetKwInput{KwOutput}{Output}              
\SetKwInOut{Parameter}{Parameters}
\DontPrintSemicolon
\SetAlgoLined
\RestyleAlgo{boxruled}
  
  \KwInput{{ $\{\mbf {f}_n(\omega_\ell)\}_{\ell=1}^{n/2}$ }}
\Parameter{Regularization parameter $\rho>0$, penalty parameter $\beta>0$, maximum number of iterations of ADMM $T$, Sparsity parameter $\hat{s}$, Maximum number of iteration of SOAP $\Tilde{T}$, Frequency parameter $\eta$, Smoothing parameter $\theta$.}

\BlankLine
    \textbf{ADMM:} $\mbf \Pi_1\leftarrow\mbox{ADMM}\left[\mbf {f}_n(\omega_1)\right]$\;
    Set the columns of $\mbf U^{init}$ to be the top $d$ leading eigenvectors of $\mbf \Pi_1$\;
    $\tilde{\mbf U}^{(0)}\leftarrow\mbox{Truncate}(\mbf U^{init},\hat{s})$, $\mbf U^{init}\leftarrow\mbox{Thin-QR}(\tilde{\mbf U}^{(0)})$\;
    $\hat{\mbf U}^{(0)}\leftarrow\mbox{SOAP}\left[\hat{\mbf {f}}^{(\theta)}(\omega_1),\mbf U^{init}\right]$\;
    \For{$\ell\in\{1,\dots,n/2-1\}$}{
        $\quad \hat{\mbf U}(\omega_{\ell+1})\leftarrow\mbox{SOAP}\left[\hat{\mbf {f}}^{(\theta)}(\omega_{\ell+1}),\hat{\mbf U}(\omega_\ell)\right]$\;
    }
    $\hat{\beta}_1,\dots,\hat{\beta}_{n/2}\leftarrow\mbox{Solve (\ref{eq:sum_beta_h}) by linear programming.}$\; 
  \KwOutput{$\{\hat{\mbf U}(\omega_\ell),\hat{\beta}_\ell\}_{\ell=1}^{n/2}$}
\end{algorithm}

\BlankLine
\BlankLine

\subsection{Localization Parameter Selection} \label{subsec:localization_par}
For selecting the localization parameter $\eta$, we propose to use information criteria based on the Whittle likelihood, or the large sample complex Gaussian distribution of the {periodogram} $d_X(\omega_\ell)\overset{D}{\rightarrow}N^c\left[0,\mbf f(\omega_\ell)\right]$. This is done given $d$ and initially without smoothing $\theta = 0$. It can be done iteratively after updating other parameters; simulations suggest that it is robust to the selection of $s$ and $\theta$.  We apply the LSPCA algorithm to obtain the eigenvectors $\hat{U}_j(\omega)$, $j=1,\dots,d$, which are used to compute the estimator of the spectrum of the $d$-dimensional principle series 
{
\[
\hat{\mbf f}_\vartheta(\omega_\ell) = \mbf {f}_n(\omega_\ell)\hat{U}_1(\omega_\ell)\hat{U}_1(\omega_\ell)^\dagger + \dots + \mbf {f}_n(\omega_\ell)\hat{U}_d(\omega_\ell)\hat{U}_d(\omega_\ell)^\dagger.   
\]
}
Let $h_\ell = \tr \left[\mbf{f}_n(\omega_\ell)\hat{\mbf V}(\omega_\ell)\hat{\mbf V}(\omega_\ell)^\dagger\right]$ be the maximum total d-dimensional power at frequency $\omega_{\ell}$, $\left\{ h_{(1)}, \dots, h_{(n/2)} \right\}$ be their order statistics, $\check{L} = \left\{\ell_{(1)}, \dots, \ell_{(n/2)} \right\}$ be the corresponding Fourier frequencies where $h_{(j)} = h_{\ell_{(j)}}$, and $\check{L}_\eta = \{\ell_{(1)},\dots,\ell_{(\eta)}\}$ be the $\eta$  fundamental frequencies indexing the location of the $\eta$ d-dimensional principal subspaces with highest total power.   Our desire is to select $\eta$ such that the power from the principle series that are not distinguishable from average power from the residual series are not maintained. We estimate average power in the residual series $\epsilon(t)$ as $\hat{\mbf \Sigma} = \sum_{\ell\in \check{L}/\check{L}_\eta}d_{X}(\omega_\ell)d_{X}(\omega_\ell)^\dagger/\left(|\check{L}| - |\check{L}_\eta|\right)$,
and consider the estimated spectral density matrices  $\hat{\mbf G}_{\ell} =  \mathbb{I}\{\ell \in \check{L}_\eta\}\hat{\mbf f}_\vartheta(\omega_\ell) + \hat{\mbf\Sigma}$.
The log-Whittle likelihood are estimated by 
\[
\log(\mathcal{L}) = -\sum_{\ell=1}^{n/2} \{p\log{\pi} + \log{|\hat{\mbf G}_\ell|} + [d_X(\omega_\ell)][\hat{\mbf G}_\ell]^{-1}[d_X(\omega_\ell)]^\dagger\},
\]
which we use to define standard information criteria for $\eta$ including $AIC = -2\log(\mathcal{L}) + 2\eta$, $AICc = -2\log(\mathcal{L}) + 2\eta + \frac{2\eta^2 + 2\eta}{n-\eta-1}$, and $BIC = -2\log(\mathcal{L}) + \log(n)\eta$.
The localization parameter $\eta$ is selected to minimizing an information criteria. 

\subsection{Sparsity Parameter Selection} \label{subsec:sparsity_par}
For selecting the sparsity level of the underlying process, we propose to use $k$-folds cross validation, with the Mahalanobis distance to evaluate the performance of fitted model in the validation step.  The procedure splits the data into $k$ blocks, or folds, of length $n/k$, and define the time series from the $r$th fold as $X^{(r)}(t) = X\left[t-\left(r-1\right)n/k\right]$, $t=(r-1)n/k+1$, $r=1,\dots,k$, { with discrete Fourier transforms $\{d_{X^{(r)}}(\omega_\ell), \ell = 1,\dots,n/(2k)\}$}. To make sure fundamental frequencies in the training and validation step match, $n$ should be divisible by $k$. In addition, let { $\mbf {f}_n^{(1)}(\omega),\dots,\mbf {f}_n^{(k)}(\omega)$} be the estimated spectral density matrices obtained from each partition. We consider 
 \begin{equation}
     \mbf {f}_n^{(-r)}(\omega) = \frac{1}{r-1}\sum_{j\neq r}\mbf {f}_n^{(j)}(\omega)
 \end{equation}
as an estimate of the spectral density matrix $\mbf f(\omega)$ removing data from the $r$th fold. We then define the rank-d principal subspace spectrum estimate $\hat{\mbf f}^{(-r)}_\vartheta(\omega)$, 
average power from the residual series $\hat{\mbf \Sigma}^{(-r)}$, and spectral matrix estimate $\hat{\mbf G}^{(-r)}_{\ell}$ for data outside of the $r$th fold in a manner analogous to the definitions in Section \ref{subsec:localization_par} using all data.

We then consider the Mahalbanois distance between the discrete Fourier transform of the data from the $r$th fold from the estimated spectral matrix from the rest of the data at the $\ell$th frequency
\[
\mathcal{D}\left[d_{X^{(r)}}(\omega_\ell)\right] = d_{X^{(r)}}(\omega_\ell)[\hat{\mbf G}_\ell^{(-r)}]^{-1}d_{X^{(r)}}(\omega_\ell)^\dagger,
\]
and average over folds and frequencies to estimate the average Mahalbanois distance
\begin{equation}
    \bar{\mathcal{D}}(s) = \frac{1}{k} \sum_{1=1}^{k}\sum_{\ell=1}^{n/(2k)}\mathcal{D}\left[d_{X^{(r)}}(\omega_\ell)\right].
\end{equation}
We select the sparsity level $\hat{s}$ that minimizes this distance
\begin{equation}
    \hat{s} = \argmin_s \bar{\mathcal{D}}(s).
\end{equation}

\subsection{Smoothing Parameter Selection} \label{subsec:smoothing_par}
 We select the smoothing parameter $\theta$ using $k$-folds cross validation, with the Mahalanobis distance to evaluate the performance of fitted model in the validation step. The procedure is similar to the procedure for the selection of the sparsity, but considers a fixed value of sparsity and varying levels of $\theta$.

\section{Proofs}
\label{app:proofs}

Appendix B is devoted to the proof of the theorems presented in the main manuscript. 
We first reiterate the model and the assumptions under which we developed the theory in Section \ref{subsec:model_assumptions}, which are presented in the paper and restated here. This is followed by the proof of the main theoretical results provided in Section \ref{subsec:proof_main}. Section \ref{subsec:proof_linear_programming} contains a proof for Proposition \ref{prop:linear_programming}. In Section \ref{subsec:proof_thm_ADMM_complex} we present a proof for Theorem \ref{thm:consistency}, Part (I). Proof of Part (II) of Theorem \ref{thm:consistency} is presented in Section \ref{subsec:proof_thm_consistency_f_theta}. Proofs of the main results follow from the Theorems and Lemmas that are presented in Section \ref{subsec:preliminary_thm_lemma}. Appendix B is concluded with some background definitions and lemmas.  

\subsection{Model Assumptions} \label{subsec:model_assumptions}
Let $\mathcal{M}_d(f,d,s^*)$ be the class of $p$-dimensional stationary time series $\{X(t): t\in\mathbb{Z}\}$ satisfying the following assumptions. 

\begin{assumption}\label{as:sparsity}
    For all $\omega\in[0,1)$, the $d$-dimensional principal subspace of $\mbf f(\omega)$ is continuous as a function of $\omega$, is $s^*$-sparse and these principal subspaces share the same support. In addition, we assume that $\inf_{\omega\in[0,.5]}\lambda_d(\omega)-\lambda_{d+1}(\omega)>\delta>0$ for some constant $\delta$.
\end{assumption}

\begin{assumption}\label{as:mixing}
    There exists constants $c_1$ and $\gamma_1\geq 1$ such that for all $h\geq 1$, the $\alpha$-mixing coefficient satisfies 
        $\alpha(h) \leq \exp\{-c_1h^{\gamma_1}\}.$
\end{assumption}

\begin{assumption} \label{as:concentration}
    There exists positive constants $c_2$ and $\gamma_2$ such that for all $v\in\mathbb{S}^{p-1}(\mathbb{C})$ and all $\lambda\geq 0$, we have 
        $\mathbb{P}\left(|v^*X(t)|\geq\lambda\right)\leq 2\exp\{-c_2\lambda^{\gamma_2}\}$  for all $t\in\mathbb{Z}.$
\end{assumption}

\begin{assumption} \label{as:gamma}
    Define $\gamma$ via $\frac{1}{\gamma}=\frac{1}{\gamma_1}+\frac{2}{\gamma_2}$, where $\gamma_1$ and $\gamma_2$ are given in Assumptions (\ref{as:mixing}) and (\ref{as:concentration}). We assume that $\gamma<1$.
\end{assumption}

\subsection{Proof of the Main Results} \label{subsec:proof_main}
{
\subsubsection{Computations related to reformulation and relaxation of the optimization problem.}
To see why
$\|\mbf\Pi_\ell-\mbf\Pi_{\ell+1}\|_F=\sqrt{2}\|\mbf\Pi_{\ell+1} - \mbf\Pi_\ell\mbf\Pi_{\ell+1}\mbf\Pi_\ell\|_F$, note that by Lemma \ref{lemma:subspace_distance}, $\|\mbf\Pi_\ell-\mbf\Pi_{\ell+1}\|_F^2 = 2\|\mbf\Pi_\ell^\perp\mbf\Pi_{\ell+1}\|_F^2$ and
\begin{align*}
    \|\mbf\Pi_\ell^\perp\mbf\Pi_{\ell+1}\|_F^2 &= \tr\{\mbf\Pi_{\ell+1}\mbf\Pi_{\ell}^\perp\mbf\Pi_{\ell+1}\mbf\Pi_{\ell}^\perp\} \\
    &= \tr\{\mbf\Pi_{\ell}^\perp\mbf\Pi_{\ell+1}\mbf\Pi_{\ell}^\perp\mbf\Pi_{\ell}^\perp\mbf\Pi_{\ell+1}\mbf\Pi_{\ell}^\perp\} \\
    &= \|\mbf\Pi_\ell^\perp\mbf\Pi_{\ell+1}\mbf\Pi_\ell^\perp\|_F
\end{align*}
In addition, 
\begin{align*}
    \mbf\Pi_{\ell+1} = (\mbf\Pi_{\ell}+\mbf\Pi_{\ell}^\perp)\mbf\Pi_{\ell+1}(\mbf\Pi_{\ell}+\mbf\Pi_{\ell}^\perp) =\mbf\Pi_\ell\mbf\Pi_{\ell+1}\mbf\Pi_{\ell} + \mbf\Pi_\ell^\perp\mbf\Pi_{\ell+1}\mbf\Pi_{\ell}^\perp.
\end{align*}
Thus, $\mbf\Pi_{\ell+1} - \mbf\Pi_\ell\mbf\Pi_{\ell+1}\mbf\Pi_\ell = \mbf\Pi_\ell^\perp\mbf\Pi_{\ell+1}\mbf\Pi_\ell^\perp$, implying that, $\|\mbf\Pi_\ell-\mbf\Pi_{\ell+1}\|_F=\sqrt{2}\|\mbf\Pi_{\ell+1} - \mbf\Pi_\ell\mbf\Pi_{\ell+1}\mbf\Pi_\ell\|_F$.

The equality
\[\tr\{\mbf f_n(\omega_{\ell+1})\mbf\Pi_{\ell+1} - \theta\mbf f_n(\omega_{\ell+1})[\mbf\Pi_{\ell+1}-\mbf\Pi_{\ell}\mbf\Pi_{\ell+1}\mbf\Pi_{\ell}]\} = \tr\{(1-\theta)\mbf f_n(\omega_\ell+1)\mbf\Pi_{\ell+1} + \theta\mbf\Pi_\ell\mbf f_n(\omega_{\ell+1})\mbf\Pi_\ell\mbf\Pi_{\ell+1}\},\]
can be seen by noting that
\begin{align*}
    &\tr\{\mbf f_n(\omega_{\ell+1})\mbf\Pi_{\ell+1} - \theta\mbf f_n(\omega_{\ell+1})[\mbf\Pi_{\ell+1}-\mbf\Pi_{\ell}\mbf\Pi_{\ell+1}\mbf\Pi_{\ell}]\} \\
    &= \tr\{(1-\theta)\mbf f_n(\omega_{\ell+1})\mbf\Pi_{\ell+1}\} + \theta\tr\{\mbf f_n(\omega_{\ell+1})\mbf\Pi_\ell\mbf\Pi_{\ell+1}\mbf\Pi_{\ell}\} \\
    &= \tr\{(1-\theta)\mbf f_n(\omega_{\ell+1})\mbf\Pi_{\ell+1}\} + \theta\tr\{\mbf\Pi_{\ell}\mbf f_n(\omega_{\ell+1})\mbf\Pi_\ell\mbf\Pi_{\ell+1}\}.
\end{align*}
}

\subsubsection{Proof of Proposition \ref{prop:linear_programming}} \label{subsec:proof_linear_programming}
\begin{proof}
    We prove the claim by induction. First note that the objective function is increasing in $\beta_j$'s and thus the maximum is attained at $\sum_{\ell=1}^{K}\beta_\ell=\eta$.\\
    When $K=2$, and $\eta=1$, order $h_1,h_2$ in decreasing order as $h_{(1)},h_{(2)}$ and their corresponding coefficients as $\beta_{(1)}$ and $\beta_{(2)}$. Note that $\beta_1+\beta_2=1$ and 
    \[
    h_{(1)} = \beta_{(1)}h_{(1)}+(1-\beta_{(1)})h_{(1)} >\beta_{(1)}h_{(1)}+(1-\beta_{(1)})h_{(2)}=\beta_{(1)}h_{(1)}+\beta_{(2)}h_{(2)}.
    \]
    For any $\eta<K$ assume that the claim of the proposition hols for any $h_1,\dots,h_K\in\mathbb{R}^+$. Since $\sum_{\ell=1}^{K}\beta_\ell=\eta$ and $0\leq\beta_\ell\leq 1, \ell=1,\dots,K$, by the induction hypothesis we have
    \begin{align*} 
        \sum_{j=1}^{\eta}h_{(j)} &> \sum_{j=1}^{K-1}\beta_{(j)}h_{(j)}+(\beta_{(K)}+\beta_{(K+1)})h_{(K)} \\
        &> \sum_{j=1}^{K+1}\beta_{(j)}h_{(j)}.
    \end{align*}
    The second inequality holds since $h_{(K)}\geq h_{K+1}$. This shows that for $K+1$ the optimum is attained at $\beta_{(1)}=\dots=\beta_{(\eta)}=1,\beta_{(\eta+1)}=\dots=\beta_{(K)}=\beta_{(K+1)}=0$ hence, by induction, the claim follows for any $K\in\mathbb{N}$. This shows that when $\eta$ is fixed, for any $K$ the conclusion of the theorem holds. 
    
    Now, we show that the result holds for any $\eta$. For $K=2$ the result is trivial. For any $K\geq3$ set $\eta=1$, then
    \begin{align}
        (1-\beta_{(1)})h_{(1)} = (\beta_{(2)}+\beta_{(3)}+ \dots + \beta_{(K)})h_{(1)} \geq \beta_{(2)}h_{(2)} + \beta_{(3)}h_{(3)} + \dots + \beta_{(K)}h_{(K)}
    \end{align}
    showing that $h_{(1)} \geq  \beta_{(1)}h_{(1)} + \beta_{(2)}h_{(2)} + \dots + \beta_{(K)}h_{(K)}$, for any $0\geq\beta_1,\dots,\beta_K\geq 1$ such that $\beta_1+\dots+\beta_K=1$. Suppose the conclusion of the theorem holds for an $\eta \leq K$ we show that the conclusion also holds for $\eta+1$. 
    
    For $\eta+1$ By the assumptions of the theorem, $\beta_{(1)}+\dots+\beta_{(K)} = \eta+1$, therefore, $\eta+1\geq\beta_{(2)}+\dots+\beta_{(K)}\geq\eta$. Let $1\geq\alpha:=\beta_{(2)}+\dots+\beta_{(K)}-\eta\geq 0$. Note that $\beta_{(1)}+\alpha=1$. Since the for $K=2, \eta=1$ the conclusion of the theorem holds, we have 
    \begin{align}
        h_{(1)} \geq \beta_{(1)}h_{(1)} + \alpha(h_{(2)}+\dots+h_{(K)})
    \end{align}
    In addition, for any $\tilde{\beta}_{(2)},\dots,\tilde{\beta}_{(K)}$ such that $\tilde{\beta}_{(2)}+\dots+\tilde{\beta}_{(K)}=\eta$, by the indication hypothesis, we have 
    \begin{align}
        h_{(2)}+\dots+h_{(\eta+1)} \geq \tilde{\beta}_{(2)}h_{(2)},\dots,\tilde{\beta}_{(K)}h_{(K)}.
    \end{align}
    In particular, for any choice of $0\leq\tilde{\beta}_{(2)},\dots,\tilde{\beta}_{(K)}\leq 1$ such that  $\tilde{\beta}_{(2)}+\dots+\tilde{\beta}_{(K)}=\eta$ and 
    \[
    \tilde{\beta}_{(2)}+\dots+\tilde{\beta}_{(K)}+\alpha=\beta_{(2)}+\beta_{(3)}+ \dots + \beta_{(K)}
    \]
    we conclude that 
    \begin{align}
        h_{(1)}+h_{(2)}+\dots+h_{(\eta+1)} &\geq \beta_{(1)}h_{(1)} + (\tilde{\beta}_{(2)}+\alpha)h_{(2)},\dots,(\tilde{\beta}_{(K)}+\alpha)h_{(K)} \\
        &\geq \beta_{(1)}h_{(1)} + \beta_{(2)}h_{(2)} + \dots + \beta_{(K)}h_{(K)}.
    \end{align}
    Hence, by indication, the result follows for any $K$ and for any $\eta\leq K$.
\end{proof}

\subsubsection{Proof of Theorem \ref{thm:consistency}, Part (I)}  \label{subsec:proof_thm_ADMM_complex}
\begin{proof}
    Let $\mbf U_\ell^{(t)}$ and $\mbf U_\ell^*$ be orthonormal matrices whose columns span $\mathcal{U}^{(t)}(\omega_1)$ and $\mathcal{U}^*(\omega_1)$, respectively. In addition, let $\mbf \Pi_\ell^{(t)}=\mbf U_\ell^{(t)}[\mbf U_\ell^{(t)}]^\dagger$ and $\mbf \Pi_\ell^* = \mbf U_\ell^*[\mbf U_\ell^*]^\dagger$ be the corresponding projection matrices. Note that
    \begin{align*}
        \mathcal{D}(\mathcal{U}^{(t)}(\omega_1), \mathcal{U}^*(\omega_1)) &= \|\Pi_\ell^{(t)} - \Pi_\ell^*\|_F \\
       &= \|(\Re{(\mbf U_\ell^{(t)})}+i\Im{(\mbf U_\ell^{(t)})})(\Re{([\mbf U_\ell^{(t)}]^\dagger)}-i\Im{([\mbf U_\ell^{(t)}]^\dagger)}) \\
       &\quad - (\Re{(\mbf U_\ell^{*})}+i\Im{(\mbf U_\ell^{*})})(\Re{([\mbf U_\ell^{*}]^\dagger)}-i\Im{([\mbf U_\ell^{*}]^\dagger)}\|_F\\
       &=\|\Re{(\mbf U_\ell^{(t)})}\Re{([\mbf U_\ell^{(t)}]^\dagger)} - \Re{(\mbf U_\ell^{*})}\Re{([\mbf U_\ell^{*}]^\dagger)} \\
       &\quad + \Im{(\mbf U_\ell^{(t)})}\Im{([\mbf U_\ell^{(t)}]^\dagger)} - \Im{(\mbf U_\ell^{*})}\Im{([\mbf U_\ell^{*}]^\dagger)} \\
       &\quad + i \{-\Re{(\mbf U_\ell^{(t)})}\Im{([\mbf U_\ell^{(t)}]^\dagger)} + \Re{(\mbf U_\ell^{*})}\Im{([\mbf U_\ell^{*}]^\dagger)} \\
       &\quad + \Im{(\mbf U_\ell^{(t)})}\Re{([\mbf U_\ell^{(t)}]^\dagger)} - \Im{(\mbf U_\ell^{*})}\Re{([\mbf U_\ell^{*}]^\dagger)} \}
       \|_F \\
       &\leq  \|\Re{(\mbf U_\ell^{(t)})}\Re{([\mbf U_\ell^{(t)}]^\dagger)}  - \Re{(\mbf U_\ell^{*})}\Re{([\mbf U_\ell^{*}]^\dagger)} \|_F \\
       &\quad + \|\Im{(\mbf U_\ell^{(t)})}\Im{([\mbf U_\ell^{(t)}]^\dagger)} - \Im{(\mbf U_\ell^{*})}\Im{([\mbf U_\ell^{*}]^\dagger)}\|_F \\
       &\quad +  \|\Re{(\mbf U_\ell^{(t)})}\Im{([\mbf U_\ell^{(t)}]^\dagger)} - \Re{(\mbf U_\ell^{*})}\Im{([\mbf U_\ell^{*}]^\dagger)}\|_F \\
       &\quad + \|\Im{(\mbf U_\ell^{(t)})}\Re{([\mbf U_\ell^{(t)}]^\dagger)} - \Im{(\mbf U_\ell^{*})}\Re{([\mbf U_\ell^{*}]^\dagger)}\|_F \\
       &\overset{(i)}{\leq} 2\left\| 
                \begin{bmatrix}
                    \Re{(\mbf U_\ell^{(t)})} \\
                    \Im{(\mbf U_\ell^{(t)})}
                \end{bmatrix}    \begin{bmatrix} \Re{([\mbf U_\ell^{(t)}]^\dagger)} &  \Im{([\mbf U_\ell^{(t)}]^\dagger)} \end{bmatrix} \right. \\
                &\qquad\qquad \left. - 
                \begin{bmatrix}
                    \Re{(\mbf U_\ell^{*})} \\
                    \Im{(\mbf U_\ell^{*})}
                \end{bmatrix}    \begin{bmatrix} \Re{([\mbf U_\ell^{*}]^\dagger)} &  \Im{([\mbf U_\ell^{*}]^\dagger)} \end{bmatrix}
            \right\|_F \\
            &=  2 \mathcal{D}([\mathcal{U}^{(t)}(\omega_1)]^{(R)}, [\mathcal{U}^*(\omega_1)]^{(R)})
    \end{align*}
    Note that $(i)$ holds since 
    \begin{align*}
        \left\| 
                \begin{bmatrix}
                   \mbf A & \mbf B \\
                    \mbf B^\dagger & \mbf D
                \end{bmatrix}   
            \right\|_F = \sqrt{\|\mbf A\|_F^2 + \|\mbf B\|_F^2 + \|\mbf B^\dagger\|_F^2 + \|\mbf D\|_F^2}
    \end{align*}
    and 
    \[
    \|\mbf A\|_F+\|\mbf B\|_F+\|\mbf C\|_F+\|\mbf D\|_F\leq 2\sqrt{\|\mbf A\|_F^2 + \|\mbf B\|_F^2 + \|\mbf B^\dagger\|_F^2 + \|\mbf D\|_F^2}.
    \]
    Thus the result follows from Theorem (\ref{thm:ADMM}) with $\tilde{\tilde{C}}^\prime = 2\tilde{C}^\prime$ and $\tilde{\tilde{C}}^{\prime\prime} = 2\tilde{C}^{\prime\prime}$.
    
\end{proof}

\subsubsection{Proof of Theorem \ref{thm:consistency}, Part (II)} 
\label{subsec:proof_thm_consistency_f_theta}
\begin{remark}[Notation]
Recall that
    \begin{align*}
        \hat{\mbf f}^{(\theta)}(\omega_{\ell+1}) &= (1-\theta)\mbf f_n(\omega_{\ell+1}) + \theta\;\hat{\mbf \Pi}_\ell \mbf f_n(\omega_{\ell+1})\hat{\mbf \Pi}_\ell\\
        \mbf f_n(\omega_\ell) &= \sum_{t=-M}^{M}\hat{\mbf R}_t\exp\{-2\pi\di\omega_\ell t\} \\
        \hat{\mbf R}_t &= \frac{1}{n}\sum_{k=1}^{n-t}X(k+t)X(t)^\dagger,
    \end{align*} 
    with the dependence of $M$ on $n$ implicit, where $\hat{\mbf \Pi}_\ell$ is the estimated $d$-dimensional principal subspace of $\mbf f(\omega_\ell)$ obtained from the LSPCA algorithm.
\end{remark}

Let $\mathcal{U}^*_\ell$ be the $d$-dimensional principal subspace of $\mbf f(\omega_\ell)$ and $\mbf U^*_\ell$ be an orthonormal matrix such that its columns span $\mathcal{U}^*_\ell$. In addition, let $\mathcal{S}^*$ be the row-support of $\mbf U^*_\ell$ and $\mathcal{I}\subseteq\{1,\dots,p\}$ be an index set. Let $\hat{\mbf f}^{(\theta)}(\omega_\ell,\mathcal{I})$ be the restriction of $\hat{\mbf f}^{(\theta)}(\omega_\ell)$ onto the columns and rows indexed by $\mathcal{I}$ and $\hat{\mbf U}_\ell(\mathcal{I})$ be the orthonormal matrix with columns consisting of the top $d$ leading eigenvectors of $\hat{\mbf f}^{(\theta)}(\omega_\ell,\mathcal{I})$. Also, let $\hat{\mathcal{U}}(\omega_\ell,\mathcal{I})$ be the space spanned by the columns of $\hat{\mbf U}_\ell(\mathcal{I})$. 
By the QR decomposition step in the Algorithm SOAP, $\mbf f_n(\omega_{\ell+1},\mathcal{I})\mbf U^{(t)}(\omega_{\ell+1}) = \mbf V^{(t+1)}(\omega_{\ell+1}) \mathsf{\mbf R}_1^{(t+1)}(\omega_{\ell+1})$. Denote the column-space of $\mbf V^{(t+1)}(\omega_{\ell+1})$ by $\mathcal{V}^{(t+1)}(\omega_{\ell+1})$.
In addition, let $\|A\|_{op} = \max\{v^\dagger A v: \|v\|_2=1\}$, the $(p,q)$-norm of a matrix $A$, $\|A\|_{p,q}$ be the $\ell_q$ norm of the $\ell_p$ norms of the rows of $A$, and
\begin{equation}
    \|\hat{\mbf f}^{(\theta)}(\omega_{\ell+1}) - \mbf f(\omega_{\ell+1})\|_{op,|\mathcal{I}|} := \max_{\|v\|_2=1, |\mbox{supp}(v)|\leq\mathcal{I}} \|v^\dagger([\hat{\mbf f}^{(\theta)}(\omega_{\ell+1}) - \mbf f(\omega_{\ell+1})]_\mathcal{I})v\|_{2}. 
\end{equation}
Throughout the rest of the document, we denote the eigenvalues, in decreasing order, of a $p\times p$ Hermitian matrix $H$ by $\lambda_1(H),\dots,\lambda_p(H)$ and will use $\lambda_k(\omega)$ to denote $\lambda_k(\mbf f(\omega))$, for $k=1,\dots,p$.
 
\begin{lemma} \label{lemma:s_op_norm_theta}
    Under the conditions of Theorem (\ref{thm:SSD_consistency}), we have
    \begin{align}
    \begin{split}
        \|\hat{\mbf f}^{(\theta)}(\omega_{\ell+1})-\mbf f(\omega_{\ell+1})\|_{op,|\mathcal{I}|} &\leq \|\mbf f_n(\omega_{\ell+1})-\mbf f(\omega_{\ell+1})\|_{op,|\mathcal{I}|} \\
        &\qquad +2\theta\lambda_1(\omega_{\ell+1})\left[\mathcal{D}(\mathcal{U}^*_\ell,\mathcal{U}^*_{\ell+1})+\mathcal{D}(\mathcal{U}^*_\ell,\hat{\mathcal{U}}_{\ell}(\mathcal{I}))\right] 
        \end{split}
        \\
        \begin{split}
        &\leq \left(\exp(-c_0M)\vee M\sqrt{\frac{s^*\log(p)}{n}}\right) \\
        &\qquad +2\theta\lambda_1(\omega_{\ell+1})\left[\mathcal{D}(\mathcal{U}^*_\ell,\mathcal{U}^*_{\ell+1})+\mathcal{D}(\mathcal{U}^*_\ell,\hat{\mathcal{U}}_{\ell}(\mathcal{I}))\right],
        \end{split}
    \end{align}
    where the second inequality holds with high probability.
\end{lemma}
\begin{proof}
    Let
    \begin{equation}
        \mbf f^{(\theta)}(\omega_{\ell+1}) = (1-\theta)\mbf f(\omega_\ell) + \theta\mbf \Pi_\ell \mbf f(\omega_{\ell+1})\mbf \Pi_\ell. 
    \end{equation}
    Note that 
    \[
    \hat{\mbf f}^{(\theta)}(\omega_{\ell+1}) - \mbf f(\omega_{\ell+1}) = \underbrace{\hat{\mbf f}^{(\theta)}(\omega_{\ell+1}) - \mbf f^{(\theta)}(\omega_{\ell+1})}_{(I)} + \underbrace{\mbf f^{(\theta)}(\omega_{\ell+1}) - \mbf f(\omega_{\ell+1})}_{(II)}
    \]
    where
    \[
    (I) = \underbrace{(1-\theta)(\mbf f_n(\omega_{\ell+1})-\mbf f(\omega_{\ell+1}))}_{(I.1)} + \theta\underbrace{\left[\hat{\mbf \Pi}_\ell\mbf f_n(\omega_{\ell+1})\hat{\mbf \Pi}_\ell-\mbf \Pi_\ell \mbf f(\omega_{\ell+1})\mbf \Pi_\ell\right]}_{(I.2)}
    \]
    and
    \begin{align*}
        (I.2) = \hat{\mbf \Pi}_\ell\left[\mbf f_n(\omega_{\ell+1})-\mbf f(\omega_{\ell+1})\right]\hat{\mbf \Pi}_\ell + \hat{\mbf \Pi}_\ell \mbf f(\omega_{\ell+1})\left(\hat{\mbf \Pi}_\ell - \mbf \Pi_\ell\right) + \left(\hat{\mbf \Pi}_\ell - \mbf \Pi_\ell\right) \mbf f(\omega_{\ell+1})\mbf \Pi_\ell.
    \end{align*}
    By $(I.1)$ and $(I.2)$
    \begin{equation} \label{ineq:op_I_norm_f_theta}
        \|\hat{\mbf f}^{(\theta)}(\omega_{\ell+1})- \mbf f^{(\theta)}(\omega_{\ell+1})\|_{op,|\mathcal{I}|}\leq \|\mbf f_n(\omega_{\ell+1})-\mbf f(\omega_{\ell+1})\|_{op,|\mathcal{I}|} + 2\theta\lambda_1(\omega_{\ell+1})\mathcal{D}(\mathcal{U}^*_{\ell},\hat{\mathcal{U}}_{\ell}(\mathcal{I}))
    \end{equation}
    Also,
    \begin{align*}
        (II) &= \theta\left[\mbf \Pi_\ell \mbf f(\omega_{\ell+1})\mbf \Pi_\ell-\mbf f(\omega_{\ell+1})\right] \\
        &= \theta[\mbf \Pi_\ell \mbf f(\omega_{\ell+1})\mbf \Pi_\ell - \mbf \Pi_\ell \mbf f(\omega_{\ell+1})\mbf \Pi_{\ell+1} + \mbf \Pi_\ell \mbf f(\omega_{\ell+1})\mbf \Pi_{\ell+1} \\
        & \qquad-\mbf \Pi_{\ell+1}\mbf f(\omega_{\ell+1})\mbf \Pi_{\ell+1} - \mbf \Pi_{\ell+1}^\perp \mbf f(\omega_{\ell+1})\mbf \Pi_{\ell+1}^\perp] \\
        &= \theta\left[\mbf \Pi_\ell \mbf f(\omega_{\ell+1})(\mbf \Pi_\ell-\mbf \Pi_{\ell+1}) + (\mbf \Pi_\ell-\mbf \Pi_{\ell+1})\mbf f(\omega_{\ell+1})\mbf \Pi_{\ell+1} - \mbf \Pi_{\ell+1}^\perp \mbf f(\omega_{\ell+1})\mbf \Pi_{\ell+1}^\perp\right].
    \end{align*}
    Thus,
    \[
    \|\mbf f^{(\theta)}(\omega_{\ell+1}) - \mbf f(\omega_{\ell+1})\|_{op} \leq 2\theta\lambda_1(\omega_{\ell+1})\mathcal{D}(\mathcal{U}^*_{\ell},\mathcal{U}^*_{\ell+1}).
    \]
    Hence,
     \begin{align}
        \|\hat{\mbf f}^{(\theta)}(\omega_{\ell+1})-\mbf f(\omega_{\ell+1})\|_{op,|\mathcal{I}|} &\leq \|\mbf f_n(\omega_{\ell+1})-\mbf f(\omega_{\ell+1})\|_{op,|\mathcal{I}|} \notag\\
        &\qquad +2\theta\lambda_1(\mbf f(\omega_{\ell+1}))\left[\mathcal{D}(\mathcal{U}^*_\ell,\mathcal{U}^*_{\ell+1})+\mathcal{D}(\mathcal{U}^*_\ell,\hat{\mathcal{U}}_{\ell}(\mathcal{I}))\right] \notag\\
        &\leq \left(\exp(-c_0M)\vee M\sqrt{\frac{s^*\log(p)}{n}}\right) \notag\\
        &\qquad +2\theta\lambda_1(\mbf f(\omega_{\ell+1}))\left[\mathcal{D}(\mathcal{U}^*_\ell,\mathcal{U}^*_{\ell+1})+\mathcal{D}(\mathcal{U}^*_\ell,\hat{\mathcal{U}}_{\ell}(\mathcal{I}))\right]. \notag
    \end{align}
    Note the second inequality is followed by Theorem \ref{thm:SSD_consistency}.
\end{proof}

\begin{lemma} \label{Lemma:delta}
    Under the same conditions of Theorem \ref{thm:SSD_consistency}, with high probability, we have
    \begin{equation}
        \mathcal{D}(\mathcal{U}^*_{\ell+1},\hat{\mathcal{U}}_{\ell+1}(\mathcal{I})) \leq \frac{\mathcal{E}_1}{\mathcal{E}_2}=:\Delta_{\ell+1}(\mathcal{I}),
    \end{equation}
    where 
    \begin{align*}
        \mathcal{E}_1 &= \sqrt{2d}\left[\left(\exp(-c_0M)\vee M\sqrt{\frac{s^*\log(p)}{n}}\right) \right. \\
        &\qquad\qquad\left. +2\theta\lambda_1(\mbf f(\omega_{\ell+1}))\left[\mathcal{D}(\mathcal{U}^*_\ell,\mathcal{U}^*_{\ell+1})+\mathcal{D}(\mathcal{U}^*_\ell,\hat{\mathcal{U}}_{\ell}(\mathcal{I}))\right] \right] \\
        \mathcal{E}_2 &= \frac{1}{2}\left[\lambda_d(\omega_{\ell+1})-(1-\theta)\lambda_{d+1}(\omega_{\ell+1})\right]-2\theta\lambda_1(\omega_{\ell+1})\mathcal{D}(\mathcal{U}^*_{\ell},\hat{\mathcal{U}}_{\ell}(\mathcal{I}))
    \end{align*}
\end{lemma}
\begin{proof}
    Let $[\mbf U_\ell^*]^\perp$ be the orthogonal matrix whose columns are the eigenvectors corresponding to $\lambda_{d+1}(\mbf f(\omega_\ell)),\dots,\lambda_{p}(\mbf f(\omega_\ell))$ and let $[\hat{\mbf U}_\ell]^\perp$ be the orthogonal matrix whose columns are the eigenvectors corresponding to $\lambda_{d+1}(\hat{\mbf f}^{(\theta)}(\omega_\ell,\mathcal{I})),\dots,\lambda_{p}(\hat{\mbf f}^{(\theta)}(\omega_\ell,\mathcal{I}))$. The following holds 
    \begin{enumerate}
        \item \begin{align*}
            \mbf f(\omega_{\ell+1})U_{\ell+1}^* &= \mbf U_{\ell+1}^*\mbf \Lambda_0(\mbf f(\omega_{\ell+1})) \\
            \hat{\mbf f}^{(\theta)}(\omega_{\ell+1},\mathcal{I})[\hat{\mbf U}_{\ell+1}]^\perp &= [\hat{\mbf U}_{\ell+1}]^\perp\Lambda_1(\hat{\mbf f}^{(\theta)}(\omega_{\ell+1},\mathcal{I})),  
        \end{align*}
        where 
        \begin{align*}
            \mbf \Lambda_0(\mbf f(\omega_{\ell+1})) &= \mbox{diag}\{\lambda_{d+1}(\mbf f(\omega_{\ell+1})),\dots,\lambda_{p}(\mbf f(\omega_{\ell+1}))\} \\
            \mbf \Lambda_1(\hat{\mbf f}^{(\theta)}(\omega_{\ell+1},\mathcal{I})) &= \mbox{diag}\{\lambda_{d+1}(\hat{\mbf f}^{(\theta)}(\omega_{\ell+1},\mathcal{I})),\dots,\lambda_{p}(\hat{\mbf f}^{(\theta)}(\omega_{\ell+1},\mathcal{I})).\}
        \end{align*}
        \item 
        \begin{align*}
            [\mbf U_{\ell+1}^*]^\dagger[\hat{\mbf f}^{(\theta)}(\omega_{\ell+1},\mathcal{I})) - \mbf f(\omega_{\ell+1})][\hat{\mbf U}_{\ell+1}]^\perp &= [\mbf U_{\ell+1}^*]^\dagger[\hat{\mbf f}^{(\theta)}(\omega_{\ell+1},\mathcal{I})) ][\hat{U}_{\ell+1}]^\perp \\
            &\quad\quad\quad - [\mbf f(\omega_{\ell+1})\mbf U_{\ell+1}^*]^\dagger[\hat{\mbf U}_{\ell+1}]^\perp \\
            &= [\mbf U_{\ell+1}^*]^\dagger[\hat{\mbf U}_{\ell+1}]^\perp\Lambda_1(\hat{\mbf f}^{(\theta)}(\omega_{\ell+1},\mathcal{I})) \\
            &\quad\quad\quad - \mbf \Lambda_0(\mbf f(\omega_{\ell+1}))[\mbf U_{\ell+1}^*]^\dagger[\hat{\mbf U}_{\ell+1}]^\perp.
        \end{align*}
        \item 
        \begin{align*}
            \|[\mbf U_{\ell+1}^*]^\dagger[\hat{\mbf U}_{\ell+1}]^\perp\|_{op} &= \|\mbf \Lambda_0^{-1}(\mbf f(\omega_{\ell+1}))\mbf \Lambda_0(\mbf f(\omega_{\ell+1}))[\mbf U_{\ell+1}^*]^\dagger[\hat{\mbf U}_{\ell+1}]^\perp\|_{op} \\
            &\leq \|\mbf \Lambda_0^{-1}(\mbf f(\omega_{\ell+1}))\|_{op}\|\mbf \Lambda_0(f(\omega_{\ell+1}))[\mbf U_{\ell+1}^*]^\dagger[\hat{\mbf U}_{\ell+1}]^\perp\|_{op} \\
            &\leq\frac{1}{\mbf \lambda_{d}(\omega_{\ell+1})}\|\mbf \Lambda_0(\mbf f(\omega_{\ell+1}))[\mbf U_{\ell+1}^*]^\dagger[\hat{\mbf U}_{\ell+1}]^\perp\|_{op}.
        \end{align*}
        \item For any sub-multiplicative matrix norm,
        \[
        \|[\mbf U_{\ell+1}^*]^\dagger[\hat{\mbf U}_{\ell+1}]^\perp\mbf \Lambda_1(\hat{\mbf f}^{(\theta)}(\omega_{\ell+1},\mathcal{I}))\| \leq \|[\mbf U_{\ell+1}^*]^\dagger[\hat{\mbf U}_{\ell+1}]^\perp\| \|\mbf \Lambda_1(\hat{\mbf f}^{(\theta)}(\omega_{\ell+1},\mathcal{I}))\|.
        \]
        \item Since the row support of $[\hat{\mbf U}_{\ell+1}]^\perp$ is $\mathcal{I}$,
        \[
        [\mbf U_{\ell+1}^*]^\dagger[\hat{\mbf f}^{(\theta)}(\omega_{\ell+1},\mathcal{I})) - \mbf f(\omega_{\ell+1})][\hat{\mbf U}_{\ell+1}]^\perp = [\mbf U_{\ell+1}^*]^\dagger[\hat{\mbf f}^{(\theta)}(\omega_{\ell+1},\mathcal{I})) - \mbf f(\omega_{\ell+1})]_\mathcal{I}[\hat{\mbf U}_{\ell+1}]^\perp.
        \]
        \item 
        \[
        \|[\mbf U_{\ell+1}^*]^\dagger[\hat{\mbf f}^{(\theta)}(\omega_{\ell+1},\mathcal{I})) - \mbf f(\omega_{\ell+1})][\hat{\mbf U}_{\ell+1}]^\perp\|_{op} \leq \|[\hat{\mbf f}^{(\theta)}(\omega_{\ell+1},\mathcal{I})) - \mbf f(\omega_{\ell+1})]_\mathcal{I}\|_{op}
        \]
        \item 
        \begin{align*}
            \|[\hat{\mbf f}^{(\theta)}(\omega_{\ell+1},\mathcal{I}) - \mbf f(\omega_{\ell+1})]_\mathcal{I}\|_{op} &= \max_{\|v\|_2=1} \|v^\dagger([\hat{\mbf f}^{(\theta)}(\omega_{\ell+1}) - \mbf f(\omega_{\ell+1})]_\mathcal{I})v\|_{2} \\
            &= \max_{\|v\|_2=1, \mbox{supp}(v)\subseteq\mathcal{I}} \|v^\dagger([\hat{\mbf f}^{(\theta)}(\omega_{\ell+1}) - \mbf f(\omega_{\ell+1})]_\mathcal{I})v\|_{2} \\
            &\leq \max_{\|v\|_2=1, |\mbox{supp}(v)|\leq\mathcal{I}} \|v^\dagger([\hat{\mbf f}^{(\theta)}(\omega_{\ell+1}) - \mbf f(\omega_{\ell+1})]_\mathcal{I})v\|_{2} \\
            &= \|\hat{\mbf f}^{(\theta)}(\omega_{\ell+1}) - \mbf f(\omega_{\ell+1})\|_{op,|\mathcal{I}|}
        \end{align*}
        \item Following 2, 3, and 4,
        \begin{align*}
            \|[\mbf U_{\ell+1}^*]^\dagger[\hat{\mbf f}^{(\theta)}(\omega_{\ell+1},\mathcal{I})) - \mbf f(\omega_{\ell+1})][\hat{\mbf U}_{\ell+1}]^\perp\|_{op} &\geq \|\mbf \Lambda_0(\mbf f(\omega_{\ell+1}))[\mbf U_{\ell+1}^*]^\dagger[\hat{\mbf U}_{\ell+1}]^\perp\|_{op} \\
            &\qquad - \|[\mbf U_{\ell+1}^*]^\dagger[\hat{\mbf U}_{\ell+1}]^\perp\mbf \Lambda_1(\hat{\mbf f}^{(\theta)}(\omega_{\ell+1},\mathcal{I}))\|_{op} \\
            &\geq \lambda_{d}(\mbf f(\omega_{\ell+1}))\|[\mbf U_{\ell+1}^*]^\dagger[\hat{\mbf U}_{\ell+1}]^\perp\|_{op} \\
            &\qquad - \lambda_{d+1}(\hat{\mbf f}^{(\theta)}(\omega_{\ell+1},\mathcal{I}))\|[\mbf U_{\ell+1}^*]^\dagger[\hat{\mbf U}_{\ell+1}]^\perp\|_{op} \\
            &= [\lambda_{d}(\mbf f(\omega_{\ell+1})) - \lambda_{d+1}(\hat{\mbf f}^{(\theta)}(\omega_{\ell+1},\mathcal{I}))]\times \\
            &\qquad\qquad\|[\mbf U_{\ell+1}^*]^\dagger[\hat{\mbf U}_{\ell+1}]^\perp\|_{op}.
        \end{align*}
        \item Following 7 and 8,
        \[
        \|\hat{\mbf f}^{(\theta)}(\omega_{\ell+1},\mathcal{I})) - \mbf f(\omega_{\ell+1})\|_{op,|\mathcal{I}|} \geq |\lambda_{d}(\mbf f(\omega_{\ell+1})) - \lambda_{d+1}(\hat{\mbf f}^{(\theta)}(\omega_{\ell+1},\mathcal{I}))|\;\|[\mbf U_{\ell+1}^*]^\dagger[\hat{\mbf U}_{\ell+1}]^\perp\|_{op}.
        \]
        \item 
        \[
        \mathcal{D}(\mathcal{U}^*_{\ell+1},\hat{\mathcal{U}}_{\ell+1}(\mathcal{I})) = \sqrt{2}\|[\mbf U_{\ell+1}^*]^\dagger[\hat{\mbf U}_{\ell+1}]^\perp\|_F \leq \sqrt{2d}\|[\mbf U_{\ell+1}^*]^\dagger[\hat{\mbf U}_{\ell+1}]^\perp\|_{op}.
        \]
        \item Thus
        \[
        \mathcal{D}(\mathcal{\mbf U}^*_{\ell+1},\hat{\mathcal{\mbf U}}_{\ell+1}(\mathcal{I})) \leq\frac{\sqrt{2d}\|\hat{\mbf f}^{(\theta)}(\omega_{\ell+1},\mathcal{I})) - \mbf f(\omega_{\ell+1})\|_{op,|\mathcal{I}|}}{\lambda_{d}(\mbf f(\omega_{\ell+1})) - \lambda_{d+1}(\hat{\mbf f}^{(\theta)}(\omega_{\ell+1},\mathcal{I}))}.
        \]
    \end{enumerate}
    Note that,
        \begin{align*}
            \lambda_{d+1}(\hat{\mbf f}^{(\theta)}(\omega_{\ell+1},\mathcal{I})) &\leq \lambda_{d+1}(\mbf f^{(\theta)}(\omega_{\ell+1},\mathcal{I})) + \lambda_{1}(\hat{\mbf f}^{(\theta)}(\omega_{\ell+1},\mathcal{I})-\mbf f^{(\theta)}(\omega_{\ell+1},\mathcal{I})) \\
            \lambda_{d+1}(\mbf f^{(\theta)}(\omega_{\ell+1},\mathcal{I})) &\leq (1-\theta)\lambda_{d+1}(\mbf f(\omega_{\ell+1})) + \theta\lambda_{d+1}(\mbf \Pi_\ell \mbf f(\omega_{\ell+1})\mbf \Pi_\ell) \\
            &= (1-\theta)\lambda_{d+1}(\mbf f(\omega_{\ell+1})) \\
            \lambda_{1}(\hat{\mbf f}^{(\theta)}(\omega_{\ell+1},\mathcal{I})-\mbf f^{(\theta)}(\omega_{\ell+1},\mathcal{I})) &\leq \|\mbf f_n(\omega_{\ell+1})-\mbf f(\omega_{\ell+1})\|_{op,|\mathcal{I}|} +2\theta\lambda_1(\mbf f(\omega_{\ell+1}))\mathcal{D}(\mathcal{U}^*_\ell,\hat{\mathcal{U}}_{\ell}(\mathcal{I})),
        \end{align*}
        where the last inequality follows from (\ref{ineq:op_I_norm_f_theta}).
        Thus,
        \begin{align*}
        \lambda_{d+1}(\hat{\mbf f}^{(\theta)}(\omega_{\ell+1},\mathcal{I})) &\leq (1-\theta)\lambda_{d+1}(\mbf f(\omega_{\ell+1})) + \|\mbf f_n(\omega_{\ell+1})-\mbf f(\omega_{\ell+1})\|_{op,|\mathcal{I}|} \\
        &\qquad\qquad +2\theta\lambda_1(\mbf f(\omega_{\ell+1}))\mathcal{D}(\mathcal{U}^*_\ell,\hat{\mathcal{U}}_{\ell}(\mathcal{I}))
        \end{align*}
        and consequently
        \begin{align*}
            \lambda_{d}(\mbf f(\omega_{\ell+1})) - \lambda_{d+1}(\hat{\mbf f}^{(\theta)}(\omega_\ell,\mathcal{I})) &\geq \lambda_{d}(\mbf f(\omega_{\ell+1})) - (1-\theta)\lambda_{d+1}(\mbf f(\omega_{\ell+1})) \\
            &\qquad - \|\mbf f_n(\omega_{\ell+1})-\mbf f(\omega_{\ell+1})\|_{op,|\mathcal{I}|} \\
            &\qquad -2\theta\lambda_1(\mbf f(\omega_{\ell+1}))\mathcal{D}(\mathcal{U}^*_\ell,\hat{\mathcal{U}}_{\ell}(\mathcal{I})).
        \end{align*}
    If we choose $M$ large enough such that 
    \[
    \|\mbf f_n(\omega_{\ell+1})-\mbf f(\omega_{\ell+1})\|_{op,|\mathcal{I}|} \leq \frac{\lambda_{d}(\mbf f(\omega_{\ell+1})) - (1-\theta)\lambda_{d+1}(\mbf f(\omega_{\ell+1}))}{2},
    \]
    then
    \begin{align*}
    \lambda_{d}(\mbf f(\omega_{\ell+1})) - \lambda_{d+1}(\hat{\mbf f}^{(\theta)}(\omega_\ell,\mathcal{I})) &\geq \frac{\lambda_{d}(\mbf f(\omega_{\ell+1})) - (1-\theta)\lambda_{d+1}(\mbf f(\omega_{\ell+1}))}{2} \\
    &\qquad -2\theta\lambda_1(\mbf f(\omega_{\ell+1}))\mathcal{D}(\mathcal{U}^*_\ell,\hat{\mathcal{U}}_{\ell}(\mathcal{I})).
    \end{align*}
    Hence,
    \begin{align*}
        \mathcal{D}(\mathcal{U}^*_{\ell+1},\hat{\mathcal{U}}_{\ell+1}(\mathcal{I})) &\leq\frac{\sqrt{2d}\|\hat{\mbf f}^{(\theta)}(\omega_\ell,\mathcal{I})) - \mbf f(\omega_{\ell+1})\|_{op,|\mathcal{I}|}}{\frac{1}{2}\left[\lambda_{d}(\mbf f(\omega_{\ell+1})) - (1-\theta)\lambda_{d+1}(\mbf f(\omega_{\ell+1}))\right]-2\theta\lambda_1(\mbf f(\omega_{\ell+1}))\mathcal{D}(\mathcal{U}^*_\ell,\hat{\mathcal{U}}_{\ell}(\mathcal{I}))} 
    \end{align*}
    The result follows from Lemma \ref{lemma:s_op_norm_theta}.
\end{proof}

\begin{lemma}
    Let $\mathcal{I}$ be a superset of the row-support of $\mbf U^{(t)}$, and let 
    $N_1\in\mathbb{N}$ be such that for each $n>N_1$ and for all $\ell=1,\dots,n/2$
    \begin{align}\label{Lemma_assumption}
    \begin{split}
           \|\mbf f_n(\omega_{\ell+1})-\mbf f(\omega_{\ell+1})\|_{op,|\mathcal{I}|} +2\theta\lambda_1(\mbf f(\omega_{\ell+1}))&\left[\mathcal{D}(\mathcal{U}^*_\ell,\mathcal{U}^*_{\ell+1})+\mathcal{D}(\mathcal{U}^*_\ell,\hat{\mathcal{U}}_{\ell}(\mathcal{I}))\right] \\
           &\leq \inf_{\omega}(\lambda_d(\omega) - \lambda_{d+1}(\omega)))/4 .
           \end{split}
    \end{align}
    In addition, suppose $n,N_1,p,M$ are large enough so that 
    \begin{align}
    \begin{split}
    \lambda_d(\omega_{\ell+1}) &- \|\mbf f_n(\omega_{\ell+1})-\mbf f(\omega_{\ell+1})\|_{op,|\mathcal{I}|} \\
    &\qquad\qquad -2\theta\lambda_1(\mbf f(\omega_{\ell+1}))\left[\mathcal{D}(\mathcal{U}^*_\ell,\mathcal{U}^*_{\ell+1})+\mathcal{D}(\mathcal{U}^*_\ell,\hat{\mathcal{U}}_{\ell}(\mathcal{I}))\right]  > 0.
    \end{split}
    \end{align}

    Let 
    \[
    \gamma(\omega_{\ell+1}) = \frac{3\lambda_{d+1}(\omega_{\ell+1})+\lambda_d(\omega_{\ell+1})}{\lambda_{d+1}(\omega_{\ell+1})+3\lambda_d(\omega_{\ell+1})}.
    \]
    If $\mathcal{D}(\mathcal{U}^{(t)}(\omega_{\ell+1}),\hat{\mathcal{U}}(\omega_{\ell+1},\mathcal{I}))<\sqrt{2}$, then we have
    \begin{equation}\label{Lemma_conclusion}
            \mathcal{D}(\mathcal{V}^{(t+1)}(\omega_{\ell+1}),\hat{\mathcal{U}}(\omega_{\ell+1}, \mathcal{I}))\leq \frac{\mathcal{D}(\mathcal{U}^{(t)}(\omega_{\ell+1}),\hat{\mathcal{U}}(\omega_{\ell+1},\mathcal{I}))}{\sqrt{1-\mathcal{D}(\mathcal{U}^{(t)}(\omega_{\ell+1}),\hat{\mathcal{U}}(\omega_{\ell+1},\mathcal{I}))^2/(2d)}}.\gamma(\omega_{\ell+1}).
    \end{equation}
\end{lemma} 
\begin{proof}
    Recall $\hat{\mbf f}^{(\theta)}(\omega_{\ell+1},\mathcal{I})\mbf U^{(t)}(\omega_{\ell+1}) = \mbf V^{(t+1)}(\omega_{\ell+1}) \mathsf{\mbf R}_1^{(t+1)}(\omega_{\ell+1})$ and $\hat{\mbf f}^{(\theta)}(\omega_{\ell+1},\mathcal{I})$ is the restriction of $\hat{\mbf f}^{(\theta)}(\omega_{\ell+1})$ on the rows and columns indexed by $\mathcal{I}$. We denote $\hat{\mbf U}_{\ell+1}(\mathcal{I})^\perp$ to be the orthogonal matrix whose columns span the subspace corresponding to $\lambda_{d+1}(\hat{\mbf f}^{(\theta)}(\omega_{\ell+1},\mathcal{I})),\dots,\lambda_{p}(\hat{\mbf f}^{(\theta)}(\omega_{\ell+1},\mathcal{I}))$
    Note that
    \[
    \mathcal{D}(\mathcal{V}^{(t+1)}(\omega_{\ell+1}),\hat{\mathcal{U}}(\omega_{\ell+1}, \mathcal{I}))/\sqrt{2} = \|[\hat{\mbf U}_{\ell+1}( \mathcal{I})^\perp]^\dagger \mbf V_{\ell+1}^{(t+1)}\|_F=:\mathcal{D}
    \]
    In addition, following (B.21) of \cite{wang2014nonconvex}
    \begin{equation}
        \begin{split}
            \mathcal{D} \leq \underbrace{\|\mbf \Lambda_1(\hat{\mbf f}^{(\theta)}(\omega_{\ell+1},\mathcal{I}))\|_{op}}_{(i)}.&\underbrace{\|[\hat{\mathcal{U}}(\omega_\ell, \mathcal{I})^\perp]^\dagger \mbf U_{\ell+1}^{(t)}\|_F}_{(ii)}.\underbrace{\|[[\hat{\mathcal{U}}(\omega_\ell, \mathcal{I})^\perp]^\dagger \mbf U_{\ell+1}^{(t)}]^{-1}\|_{op}}_{(iii)} \\
            &.\underbrace{\|[\mbf \Lambda_0(\hat{\mbf f}^{(\theta)}(\omega_{\ell+1},\mathcal{I}))]^{-1}\|_{op}}_{(iv)}.\underbrace{\|[\hat{\mbf U}_{\ell+1}( \mathcal{I})]^\dagger \mbf V_{\ell+1}^{(t+1)}\|_{op}}_{(v)}
        \end{split}
    \end{equation}
    Now we analyze each term. 
    \begin{itemize}
        \item[(i)]  
        \begin{align*}
            \|\mbf \Lambda_1(\hat{\mbf f}^{(\theta)}(\omega_{\ell+1},\mathcal{I}))\|_{op} &\leq \lambda_{d+1}(\mbf f(\omega_{\ell+1})) + \|\hat{\mbf f}^{(\theta)}(\omega_{\ell+1}-\mbf f(\omega_{\ell+1})\|_{op,|\mathcal{I}|} \\
            &\leq \lambda_{d+1}(\mbf f(\omega_{\ell+1})) + \frac{\lambda_{d}(\mbf f(\omega_{\ell+1})) - \lambda_{d+1}(\mbf f(\omega_{\ell+1}))}{4} \\
            &= \frac{\lambda_{d}(\mbf f(\omega_{\ell+1})) + 3\lambda_{d+1}(\mbf f(\omega_{\ell+1}))}{4}
        \end{align*}
        \item[(ii),(iii)] 
        \begin{align}
            \|[\hat{\mathcal{U}}(\omega_\ell, \mathcal{I})^\perp]^\dagger \mbf U_{\ell+1}^{(t)}\|_F &= \mathcal{D}(\mathcal{U}^{(t)}(\omega_{\ell+1}),\hat{\mathcal{U}}(\omega_{\ell+1},\mathcal{I}))/\sqrt{2} \\
            \|[[\hat{\mathcal{U}}(\omega_\ell, \mathcal{I})^\perp]^\dagger \mbf U_{\ell+1}^{(t)}]^{-1}\|_{op} &\leq \frac{1}{\sqrt{1-\mathcal{D}(\mathcal{U}^{(t)}(\omega_{\ell+1}),\hat{\mathcal{U}}(\omega_{\ell+1},\mathcal{I}))^2/(2d)}}
        \end{align}
            \item[(iv)] 
    \begin{align}
        \lambda_d(\hat{\mbf f}^{(\theta)}(\omega_{\ell+1},\mathcal{I})) &\geq \lambda_d(\mbf f(\omega_{\ell+1})) -  \|\hat{\mbf f}^{(\theta)}(\omega_{\ell+1})-\mbf f(\omega_{\ell+1})\|_{op,|\mathcal{I}|} \\
        &\geq \lambda_d(\omega_{\ell+1}) - \|\mbf f_n(\omega_{\ell+1})-\mbf f(\omega_{\ell+1})\|_{op,|\mathcal{I}|}\notag \\
        &\qquad\qquad\qquad -2\theta\lambda_1(\mbf f(\omega_{\ell+1}))\left[\mathcal{D}(\mathcal{U}^*_\ell,\mathcal{U}^*_{\ell+1})+\mathcal{D}(\mathcal{U}^*_\ell,\hat{\mathcal{U}}_{\ell}(\mathcal{I}))\right]\notag 
    \end{align}
    Since by the assumption, $p,n,M$ are large enough so that the right hand side of the second inequality is positive, we have 
    \begin{align*}
        \|[\mbf \Lambda_0(\hat{\mbf f}^{(\theta)}(\omega_{\ell+1},\mathcal{I}))]^{-1}\|_{op} &\leq \frac{1}{\lambda_d(\hat{\mbf f}^{(\theta)}(\omega_{\ell+1},\mathcal{I}))} \\
        &\leq \frac{1}{\lambda_d(\mbf f(\omega_{\ell+1})) -  \|\hat{\mbf f}^{(\theta)}(\omega_{\ell+1})-\mbf f(\omega_{\ell+1})\|_{op,|\mathcal{I}|}} \\
        &\leq\frac{1}{\lambda_d(\mbf f(\omega_{\ell+1})) - \frac{\lambda_{d}(\mbf f(\omega_{\ell+1})) - \lambda_{d+1}(\mbf f(\omega_{\ell+1}))}{4}} \\
        &= \frac{4}{3\lambda_{d}(\mbf f(\omega_{\ell+1})) + \lambda_{d+1}(\mbf f(\omega_{\ell+1}))}
    \end{align*}
    \item[(v)]  
    \[
    \|[\hat{\mbf U}_{\ell+1}( \mathcal{I})]^\dagger \mbf V_{\ell+1}^{(t+1)}\|_{op} \leq 1.
    \]
    \end{itemize}
    The above argument together with (\ref{Lemma_assumption}) implies (\ref{Lemma_conclusion}).  
\end{proof}

\begin{lemma} \label{Lemma:initial_estimate}
    Assume that 
    \[
    \hat{s}=C\max\left\{\left[\frac{4d}{(\gamma^{-1/2}-1)^2}\right],1 \right\}.s^*, 
    \]
    and
    \[
    \mathcal{D}(\mathcal{U}^{(t)}(\omega_{\ell+1}), \mathcal{U}^*(\omega_{\ell+1}))\leq \min\left\{\sqrt{2d(1-\gamma^{1/2})},\sqrt{2}/2\right\},
    \]
    where $C\geq 1$ is an integer constant. Under conditions of Theorem (\ref{thm:SSD_consistency}), we can choose $N_3\in\mathbb{N}$ such that for all $n>N_3$, $\Delta(2\hat{s})\leq 1/24$. Then, for all $n>N_3$, 
    \[
     \mathcal{D}(\mathcal{U}^{(t+1)}(\omega_{\ell+1}), \mathcal{U}^*(\omega_{\ell+1}))\leq \gamma^{1/4}. \mathcal{D}(\mathcal{U}^{(t)}(\omega_{\ell+1}), \mathcal{U}^*(\omega_{\ell+1})) + 3\gamma^{1/2}\Delta(2\hat{s}),
    \]
    where 
    \begin{equation*}
        \Delta(s) := \sup_{\omega_\ell}\frac{\Upsilon}{\frac{1}{2}\left[\lambda_d(\omega_{\ell+1})-(1-\theta)\lambda_{d+1}(\omega_{\ell+1})\right]-2\theta\lambda_1(\omega_{\ell+1})\alpha}
    \end{equation*}
   
    and
    \begin{align*}
        \Upsilon &= \sqrt{2d}\left[\left(\exp(-c_0M)\vee M\sqrt{\frac{s^*\log(p)}{n}}\right) +2\theta\lambda_1(\omega_{\ell+1})\left[\mathcal{D}(\mathcal{U}^*_\ell,\mathcal{U}^*_{\ell+1})+\alpha)\right] \right] ,   \\
        \alpha &= \sup_{\omega_\ell} \frac{c_1\sqrt{2d}}{\lambda_d(\omega_\ell)-\lambda_{d+1}(\omega_\ell)}\left(\exp(-c_0M)\vee M\sqrt{\frac{s^*\log(p)}{n}}\right).
    \end{align*}
\end{lemma}

\begin{proof}
    Proof of the Lemma follows along the same lines as the proof of Lemma 5.6 of Wang et al. For the arguments to hold, we need $ \mathcal{D}(\mathcal{U}^*_{\ell+1},\hat{\mathcal{U}}_{\ell+1}(2\hat{s}))\leq \sqrt{2}/2$ and $\Delta(2\hat{s})\leq 1/24$ which by Lemma \ref{Lemma:delta} hold, when $n,p,M(n)$ are large enough. We also need to show that 
    \begin{equation} \label{eq:assumption_of_lemma}
            \mathcal{D}(\mathcal{U}^{(t)}(\omega_{\ell+1}), \mathcal{U}^*(\omega_{\ell+1}))\leq \sqrt{2}/2
    \end{equation}
    for all $\ell\in\{1,\dots,n/2\}$. Note that, for $\ell=1$, Theorem \ref{thm:consistency}, Part (I) guarantees that for sufficiently large $n,p,M(n),$ 
    $\mathcal{D}(\mathcal{U}^{(t)}(\omega_1), \mathcal{U}^*(\omega_1)) \leq \sqrt{2}/2$. In addition, Theorem (\ref{thm:consistency}), guarantees that $\mathcal{D}(\mathcal{U}^{(init)}(\omega_{2}), \mathcal{U}^*(\omega_{2}))\leq \sqrt{2}/2$. If we assume (\ref{eq:assumption_of_lemma}) holds for $\ell=3,\dots,L$, we can show the result of Theorem \ref{thm:consistency} holds for $\ell=L$, which guarantees $\mathcal{D}(\mathcal{U}^{(init)}(\omega_{L+1}), \mathcal{U}^*(\omega_{L+1}))\leq \sqrt{2}/2$ for sufficiently large $n,p,M(n)$.
\end{proof}

\begin{theorem} \label{thm:SOAP_theta}
    Let $\{X(t): t=1,\dots,n\}$ be a realization of a weakly stationary time series that follows $\mathcal{M}_d(f,d,s^*)$ with $n>n_{min}$. Suppose the sparsity parameter $\hat{s}$ in Algorithm SOAP is chosen such that  
    \[
     \hat{s}=C\max\left\{\left[\frac{4d}{(\gamma^{-1/2}-1)^2}\right],1\right\}.s^*
    \]
    for some integer constant $C\geq 1$. If the column space $\mathcal{U}^{init}$ of the initial estimator $U^{init}$ of Algorithm SOAP at each frequency $\omega_\ell$ satisfies
    \[
    \mathcal{D}(\mathcal{U}^{init}(\omega_\ell), \mathcal{U}^*(\omega_\ell))\leq R = \min\left\{\sqrt{\frac{d\gamma(1-\gamma^{1/2})}{2}},\frac{\sqrt{2\gamma}}{4}\right\},
    \]
    then the iterative sequence $\{\mathcal{U}^{(t)}\}_{t=T+1}^\infty$ obtained for $\hat{f}^{(\theta)}_M$ satisfies
    \[
    \mathcal{D}(\mathcal{U}^{(t)}(\omega_\ell), \mathcal{U}^*(\omega_\ell))\leq \gamma^{(t-T-1)/4}\min\left\{\sqrt{\frac{d(1-\gamma^{1/2})}{2}},\frac{\sqrt{2}}{4}\right\} + \frac{3\gamma^{1/2}}{1-\gamma^{1/4}}\Delta(2\hat{s}).
    \]
    with probability at least $C^\prime(6p)^{s^*} M\exp\{-\frac{(ns^*\log(p))^{\gamma/2} }{\tilde{C}}\}$ for some constants $C^\prime$ and $\tilde{C}$, where $\Delta(s)$ is as defined in (\ref{eq:Delta_theta}).
\end{theorem}
\begin{proof}
    Proof follows along the same lines as in the proof of Theorem 4.2 in \cite{wang2014nonconvex} with the probability bound derived in Theorem \ref{thm:SSD_consistency} and using previous Lemmas.
\end{proof}

\begin{proof}[Proof of Part (II) of Theorem \ref{thm:consistency}]
    Proof follows from Part (I) and Theorem \ref{thm:SOAP_theta}.
\end{proof}

\subsection{Preliminary Theorems and Lemmas} \label{subsec:preliminary_thm_lemma}

Let 
\begin{equation*}
    \|\mbf f_n(\omega_{\ell+1}) - \mbf f(\omega_{\ell+1})\|_{op,|\mathcal{I}|} := \max_{\|v\|_2=1, |\mbox{supp}(v)|\leq\mathcal{I}} \|v^\prime([\mbf f_n(\omega_{\ell+1}) - \mbf f(\omega_{\ell+1})]_\mathcal{I})v\|_{2}. 
\end{equation*}
Below, we adapted the proof technique of \cite{lusparse} to establish the following concentration result. 
\begin{theorem} \label{thm:SSD_consistency}
    Under the Assumptions \ref{as:mixing}, \ref{as:concentration} and \ref{as:gamma}, if $M\to\infty$ and $M\sqrt{\frac{\log(p)}{n}}\to 0$ as $n\to\infty$, with the dependence of $M$ on $n$ implicit, then
    \begin{align}
        \underset{\omega\in[0,1)}{\sup} \|\mbf f_n(\omega)-\mbf f(\omega)\|_{op,s^*} &= O_P\left(\exp(-c_0M)\vee M\sqrt{\frac{s^*\log(p)}{n}}\right). 
    \end{align}
\end{theorem}

\begin{proof} 
    Let $\hat{\mbf R}_t=\frac{1}{n}\sum_{k=1}^{n-t}X(k+t)X(t)^\prime, \mbf R_t = \mathbb{E}[X(t)X(0)], Y_v^t(k):=v^\prime X(k+t)X(t)^\prime v - \mathbb{E}[v^\prime X(k+t)X(t)^\prime v]$. Note that $v^\prime(\hat{\mbf R}_t-\mbf R_t)v=\frac{1}{n}\sum_{k=0}^{n-t}Y_v^t(k)$. In addition, by Lemma \ref{lemma:decay_rate_cov} $\mathbb{E}[v^\prime X(k+t)X(t)v]\leq \|\mbf R_t\|\leq \frac{8\sqrt{2}}{\gamma_2}$ and by Assumption \ref{as:mixing} and the fact that $\alpha_{Y_v^t}(m)\leq\alpha_X (m(t+1)+1)$ for all $m\geq 1$, we have $\alpha_{Y_v^t(n)}\leq \exp\{-c_1 k^{\gamma_1}\}$. By Assumption \ref{as:concentration} and Lemma \ref{lemma:sub_g_products}, there exist a constant $c_2^{\prime}$ which only depends on $c_2$ such that for every $v\in\mathbb{S}^{p-1}(\mathbb{C}), t\in\mathbb{Z}, k\in\{1,\dots,n-t\}, \lambda\geq 0$, we have 
    \begin{equation}
        \mathbb{P}\left(\left|Y_v^t(k)-\mathbb{E}[Y_v^t(k)]\right|\geq \lambda\right)\leq 2\exp\left\{-c_2^\prime\lambda^{\gamma_2/2}\right\}.
    \end{equation}
    Thus by Theorem 1.1 of \cite{merlevede2011bernstein}, there exists $V<\infty$ and $C_1,\dots,C_4$ depending on $c,\gamma_1,\gamma_2$ such that for all $n$ and $\lambda>\frac{1}{n-t}$, we have
    \begin{align} \label{ineq:Bernstein}
        \mathbb{P}\left(\left|v^\prime(\hat{\mbf R}_t-\mbf R_t)v\right| \geq \lambda\right) &\leq n\exp\left\{-\frac{n^\gamma\lambda^\gamma}{C_1}\right\} + \exp\left\{-\frac{n^2\lambda^2}{C_2(1+nV)}\right\} \notag \\
        & \qquad\qquad + \exp\left\{-\frac{n\lambda^2}{C_3}\exp\left\{\frac{n^{\gamma(1-\gamma)}\lambda^{\gamma(1-\gamma)}}{C_4(\log(n\lambda))^\gamma}\right\}\right\}.
    \end{align}
    We then apply the $\varepsilon$-net type argument. Let $\mathcal{N}_{1/2}$ be a $\frac{1}{2}$-net of $\mathbb{S}^{p-1}(\mathbb{C})\cap\mathbb{B}_0(s^*)$. Note that, it contains $\binom{p}{s^*}6^{s^*}$ points. In addition, note that for any Hermitian matrix $\mbf A$, $\|\mbf A\|_{op,s^*}\leq C\max_{v\in\mathcal{N}_{1/2}}|v^\prime A v|$. Let $\mbf f_n(\omega)=\sum_{t=-M}^{M}\hat{\mbf R}_t\exp\{-2\pi\di\omega t\}$, then
    \begin{align*}
        \mathbb{P}\left(\|\mbf f_n(\omega)-\mbf f(\omega)\|_{op,s^*}\geq C\lambda\right) &\leq \sum_{v\in\mathcal{N}_{1/2}}\mathbb{P}\left(\left|v^\prime(\mbf f_n(\omega)-\mbf f(\omega))v\right|\geq\lambda\right) \\
        &\leq \binom{p}{s^*}6^{s^*}\left[\mathbb{P}\left(\left|\sum_{t=-M}^{M}v^\prime(\hat{\mbf R}_t-\mbf R_t)v\exp\{-2\pi\di\omega t\}\right|\geq\frac{\lambda}{2}\right) \right. \\ 
        &\qquad\qquad\qquad\left. + \mathbb{P}\left(\left|\sum_{|t|>M}v^\prime \mbf R_t v\exp\{-2\pi\di\omega t\}\right|\geq\frac{\lambda}{2}\right)\right].
    \end{align*}
    Note that 
    \begin{align*}
        Q_1 = \left|\sum_{|t|>M}v^\prime \mbf R_t v\exp\{-2\pi\di\omega t\}\right| &\leq 2\sum_{t>M} \| \mbf R_t \exp\{-2\pi\di\omega t\} + \mbf R_{-t}\exp\{2\pi\di\omega t\}\|_{op} \\
        &\overset{Lemma(\ref{lemma:decay_rate_cov})}{\leq} 2c_3\sum_{t>M}\exp\{-c_4t^{-\gamma_1}\}\leq 2c_3\exp\{-c_4M\},
    \end{align*}
    which for large $M$ is smaller than $\lambda/2$. Hence the second term is $0$. For the first term
    \begin{align*}
        \mathbb{P}\left(\left|\sum_{t=-M}^{M}v^\prime(\hat{\mbf R}_t-\mbf R_t)v\exp\{-2\pi\di\omega t\}\right|\geq\frac{\lambda}{2}\right) \leq \sum_{t=-M}^{M}\mathbb{P}\left(\left|v^\prime(\hat{\mbf R}_t-\mbf R_t)v\right|\geq\frac{\lambda}{2(2M+1)}\right) =:Q2.
    \end{align*}
    Let $\lambda=M\sqrt{\frac{s^*\log(p)}{n}}$. By utilizing (\ref{ineq:Bernstein}) We show that $\binom{p}{s^*}6^{s^*}Q_2\leq (I)+(II)+(III)$ where $(I),(II),(III)$ are as follows.
    \begin{align*}
        (I) = \binom{p}{s^*}6^{s^*}\sum_{t=0}^{M}n\exp\left\{-\frac{n^\gamma\lambda^\gamma}{C_1(4M+2)^\gamma}\right\} &\leq C_1^\prime (6p)^{s^*}Mn\exp\left\{-\frac{n^\gamma\lambda^\gamma}{\tilde{C}_1M^\gamma}\right\} \\
        &\leq C_1^\prime(6p)^{s^*}Mn\exp\left\{-\frac{n^\gamma M^\gamma (s^{*})^{\gamma/2}(\log(p))^{\gamma/2}}{\tilde{C}_1n^{\gamma/2}M^\gamma}\right\} \\
        &= C_1^\prime(6p)^{s^*}Mn\exp\left\{-\frac{n^{\gamma/2}(s^*\log(p))^{\gamma/2}}{\tilde{C}_1} \right\}
    \end{align*}
    \begin{align*}
        (II) = \binom{p}{s^*}6^{s^*}\sum_{t=0}^{M}\exp\left\{-\frac{n^2\lambda^2}{C_2(1+nV)(4M+2)^2}\right\} &\leq C_2^\prime(6p)^{s^*}M\exp\left\{-\frac{n^2\lambda^2}{\Tilde{C}_2M^2}\right\} \\
        &\leq C_2^\prime (6p)^{s^*} M\exp\left\{-\frac{n^2M^2s^*\log(p)}{\tilde{C}_2M^2 n}\right\} \\
        &= C_2^\prime (6p)^{s^*} M\exp\left\{-\frac{n s^* \log(p)}{\tilde{C}_2}\right\}
    \end{align*}
    \begin{align*}
        (III) = \binom{p}{s^*}6^{s^*}\sum_{t=0}^{M}\exp\left\{-\frac{n\lambda^2}{C_3}\exp\left\{\frac{n^{\gamma(1-\gamma)}\lambda^{\gamma(1-\gamma)}}{C_4(\log(n\lambda))^\gamma}\right\}\right\} \leq (II)
    \end{align*}
    If $M\to\infty$ and $M\sqrt{\frac{s^*\log(p)}{n}}\to 0$ and $n\to\infty$, then the $Q_2\to 0$. the upper bound holds for all $\omega\in[0,1)$, therefore
    \[
    \underset{\omega\in[0,1)}{\sup} \|\mbf f_n(\omega)-\mbf f(\omega)\|_{op,s^*} = O_P\left(\exp(-c_0M)\vee M\sqrt{\frac{s^*\log(p)}{n}}\right).
    \]

\end{proof}

\begin{lemma}
    Suppose that $\sqrt{s^*/\hat{s}}\leq 1$ and there exists $N_2\in\mathbb{N}$ such that for all $n>N_2$ and all $\ell=1,\dots,n/2$, $\mathcal{D}(\mathcal{V}^{(t+1)}(\omega_\ell), \mathcal{U}^*(\omega_\ell))\leq 1$. Then, for all $n>N_2$ and $\ell=1,\dots,n/2$,
    \[
    \mathcal{D}(\mathcal{U}^{(t+1)}(\omega_\ell), \mathcal{U}^*(\omega_\ell))\leq\left(1+2\sqrt{\frac{d.s^*}{\hat{s}}}\right)\mathcal{D}(\mathcal{V}^{(t+1)}(\omega_\ell), \mathcal{U}^*(\omega_\ell)).
    \]
\end{lemma}
\begin{proof}
    Proof follows along the same lines as in Lemma 5.5 of \cite{wang2014nonconvex}.
\end{proof}

Let 
\begin{align}\label{eq:f^{(R)}}
    {^{(R)}}{ \hat{\mbf {f}}^{(\theta)}}(\omega_\ell):=
    \begin{bmatrix}
        \Re{\left\{ \hat{\mbf {f}}^{(\theta)}(\omega_\ell) \right\}} & \Im{ \left\{ \hat{\mbf {f}}^{(\theta)}(\omega_\ell) \right\}} \\
        -\Im{ \left\{ \hat{\mbf {f}}^{(\theta)}(\omega_\ell) \right\}} & \Re{ \left\{ \hat{\mbf {f}}^{(\theta)}(\omega_\ell) \right\}} 
    \end{bmatrix}.
\end{align}
If the eigenvalues and corresponding eigenvectors of $ \hat{\mbf {f}}^{(\theta)}(\omega_\ell)$ are $\lambda_j(\omega_\ell)$ and $u_j(\omega_\ell), j=1,\dots,p$, then the $\left[2\left(j+1\right) + 1\right]$st and $\left[2\left(j+1\right) + 2\right]$st  eigenvalues and corresponding eigenvectors of ${^{(R)}}{ \hat{\mbf {f}}^{(\theta)}}(\omega_\ell)$ are 
$
\lambda_j(\omega_\ell)\;$, $\;\left[ \Re{\left\{u_j(\omega_\ell) \right\}}, \Im{\left\{u_j(\omega_\ell)\right\} }  \right]^\dagger$ and $\lambda_j(\omega_\ell)\;$, $\;\left[-\Im{\left\{u_j(\omega_\ell)\right\}}, \Re{\left\{u_j(\omega_\ell)\right\}} \right]^{\dagger}$, respectively.
\begin{theorem} \label{thm:ADMM}
    Let $\{X(t): t=1,\dots,n\}$ be a realization of a weakly stationary time series that follows $\mathcal{M}_d(f,d,s^*)$ with $n>n_{min}$ and ${^{(R)}}{ \hat{\mbf {f}}^{(\theta)}}(\omega_\ell)$ be as defined in (\ref{eq:f^{(R)}}). In addition, let the regularization parameter in (\ref{eq:relaxed_convex_opt}) be $\varrho=C\lambda_1(\omega_1)\sqrt{\log(p)/n}$ for a sufficiently large constant $C$, and the penalty parameter $\beta$ in (\ref{eq:augmented_Lagrangian}) be $\beta=\sqrt{2}p.\varrho/\sqrt{d}$. Then the iterative sequence of $d$-dimensional subspace $\{\mathcal{U}^{(t)}(\omega_1)\}_{t=1}^{T}$ satisfies
    \begin{align*}
    \mathcal{D}(\mathcal{U}^{(t)}(\omega_1), \mathcal{U}^*(\omega_1)) &\leq \frac{\tilde{C}^\prime\lambda_1(\omega_1)}{\lambda_d(\omega_1)-\lambda_{d+1}(\omega_1)}.s^*\sqrt{\frac{\log(p)}{n}} \\
    &\qquad\qquad + \frac{\tilde{C}^{\prime\prime}\sqrt{M\lambda_1(\omega_1)}}{\sqrt{\lambda_d(\omega_1)-\lambda_{d+1}(\omega_1)}}\left(\frac{d.p^2\log(p)}{n}\right)^{1/4}.\frac{1}{\sqrt{t}}
    \end{align*}
    with high probability, where $\tilde{C}^\prime=4C$ a $\tilde{C}^{\prime\prime}$ are constants. 
\end{theorem}
\begin{proof}
    Proof follows along the same lines as in the proof of Theorem 4.3 of \cite{wang2014nonconvex} with $\varrho=C\lambda_1(\omega_1)\sqrt{\log(p)/n}$.
\end{proof}


\subsection{Background Definitions and Lemmas}
\begin{definition}\label{def:alpha_mixing}
    Given two $\sigma$-algebras $\mathcal{A}$ and $\mathcal{B}$, the $\alpha$-mixing coefficient between $\mathcal{A}$ and $\mathcal{B}$, denoted by $\alpha(\mathcal{A}, \mathcal{B})$, is defined as 
    \begin{equation}\label{eq:alpha_mixing}
        \alpha(\mathcal{A},\mathcal{B}) = \sup_{A\in\mathcal{A}, B\in\mathcal{B}} \mathbb{P}(A\cap B) - \mathbb{P}(A)\mathbb{P}(B)
    \end{equation}
\end{definition}

\begin{definition}\label{def:strong_mixing-coeff}
    Given a weakly stationary time series $X_t$, the strong-mixing coefficient at lag $h$ is defined as 
    \begin{equation}\label{eq:strong_mising_coeff}
        \alpha(h) = \alpha(\sigma(X_t,t\leq 0), \sigma(X_t,t\geq h).
    \end{equation}
\end{definition}

\begin{lemma}\label{lemma:decay_rate_cov}
    Let $\mbf R_t = \mathbb{E}[X(t)X(0)]$. Under the assumptions (\ref{as:mixing}, \ref{as:concentration}, \ref{as:gamma}), there exist constants $c_3$ and $c_4$ which only depends on $c_1,c_2$ and $\gamma_2$, such that for all $t\in\mathbb{Z}$ and all $\omega\in[0,1)$,  
    \begin{equation}
        \|\mbf R_t\exp\{-2\pi\di\omega t\}+\mbf R_{-t}\exp\{2\pi\di\omega t\}\|_{op} \leq c_3\exp\{-c_4t^{\gamma_1}\}
    \end{equation}
\end{lemma}

\begin{lemma}\label{lemma:sub_g_products}
    If $X$ and $Y$ are two sub-Gaussian random variables, then $XY$ is a complex sub-exponential random variable. 
\end{lemma}

\begin{lemma}\label{Lemma:isomorphism}
    To any $J\times K$ matrix $\mbf Z$ with complex entries there corresponds a $(2J)\times (2K)$ matrix $\mbf Z^R$ with real entries such that 
    \begin{enumerate}
        \item if $\mbf Z=\mbf X+\mbf Y$, then $\mbf Z^R=\mbf X^R+\mbf Y^R$
        \item if $\mbf Z=\mbf X\mbf Y$, then $\mbf Z^R=\mbf X^R\mbf Y^R$
        \item if $\mbf Y=\mbf Z^{-1}$, then $\mbf Y^R=(\mbf Z^R)^{-1}$
        \item $\det(\mbf Z^R)=|\det(\mbf Z)|^2$
        \item if $\mbf Z$ is Hermitian, then $\mbf Z^R$ is symmetric
        \item if $\mbf Z$ is unitary, then $\mbf Z^R$ is orthogonal
        \item if the latent values and vectors of $\mbf Z$ are $\mu_j,\alpha_j,j=1,\dots,J$, then those of $\mbf Z^R$ are, respectively,
        \[
\lambda_j, 
\begin{bmatrix}
    \Re{(u_j)} \\
    \Im{(u_j)}
\end{bmatrix}
,\mbox{ and },
\lambda_j, 
\begin{bmatrix}
    -\Im{(u_j)} \\
    \Re{(u_j)}
\end{bmatrix}, j=1,\dots,J.
\]
        providing that dimensions of matrices appearing throughout the lemma are appropriate. 
        In addition, the $\mbf Z^R$ may be taken to be 
        \[
        Z^R=
    \begin{bmatrix}
        \Re{(\mbf Z)} & \Im{(\mbf Z)} \\
        -\Im{(\mbf Z)} & \Re{(\mbf Z)} 
    \end{bmatrix}.
        \]
    \end{enumerate}
\end{lemma}

Let $\mathcal{U}$ and $\mathcal{U}^\prime$ be two $d$-dimensional subspaces of $\mathbb{R}^p$ with projection matrices $\Pi$ and $\mbf \Pi^\prime$, respectively. In addition, let columns of $\mbf U=[u_1|\dots|u_d]$ and $\mbf U^\prime=[u_1^\prime|\dots|u_d^\prime]$ be orthonormal basis of $\mathcal{U}$ and $\mathcal{U}^\prime$, respectively. Also, let $\mathcal{U}^\perp$ be the orthogonal complement of $\mbf U$ and $\mbf U^\perp=[u_{d+1}|\dots|u_p]$ be an orthonormal matrix whose columns are orthogonal to the columns of $U$. The following lemma, adopted from \cite{wang2014nonconvex}, characterizes useful properties of the distance between subspaces.
\begin{lemma} \label{lemma:subspace_distance}
    Let $\mathcal{U}$ and $\mathcal{U}^\prime$ be two $d$-dimensional subspaces of $\mathbb{R}^p$. Then,
    \begin{equation*}
        \mathcal{D}(\mathcal{U},\mathcal{U}^\prime)=\sqrt{2}\|\mbf U^\dagger \mbf U^{\prime\perp}\|_F=\sqrt{2}\|[\mbf U^\perp]^\dagger \mbf U^\prime\|_F = \sqrt{2}\|\mbf \Pi^\perp \mbf \Pi^\prime\|_F=\sqrt{2}\|\mbf \Pi\mbf \Pi^{\prime\perp}\|_F\leq\sqrt{2d} 
    \end{equation*}
    and
    \begin{equation*}
        \mathcal{D}(\mathcal{U},\mathcal{U}^\prime)=\sqrt{2}\left[d-\|\mbf U^\dagger \mbf U^\prime\|_F^2\right]^{1/2}=\sqrt{2}\left[d-\frac{\mathcal{D}(\mathcal{U},\mathcal{U}^{\prime\perp})^2}{2}\right]^{1/2}.
    \end{equation*}
\end{lemma}
\begin{proof}
    See \cite{stewart1990perturbation} and \cite{bhatia2013matrix} for details.
\end{proof}

\section{Simulations} \label{Appendix:Simulation}

This appendix provides four classes of additional results from simulation studies: illustrations of the role of the smoothness penalty, stability of the tuning parameter selection routine, computational runtime, and performance in sparsity and frequency selection.  

\subsection{Illustration}

The bottom two panels of Figure 2 in the manuscript help illustrate the benefits of the smoothness penalty in improved estimation and interpretation through the analysis of one realization of a simulated time series.  To further illustrate these benefits, here we take a closer look at information that is displayed in this figure for just one coordinate. Estimates from the modulus of the principal component for the first coordinate without sharing of information across frequencies $\theta =0$, $\theta =0.6$ and the true population value are displayed in Figure \ref{fig:evec_coordinate}.  We can see that estimates at certain frequencies within the range of support $\Omega$ are estimated as zero when $\theta = 0$.  This is a result of the sparsity selection procedure  falsely selecting other coordinates as having non-zero support at a given frequency compared to this first coordinate. The result is not only poorer finite sample estimation compared to when $\theta = 0.6$, but also results in an estimate where there is not a continuous band of support frequency. 

\begin{figure}[H] 
\centering
\begin{tabular}{c}
 \includegraphics[width=5.5in, height=3in, bb = 50 50 600 600]{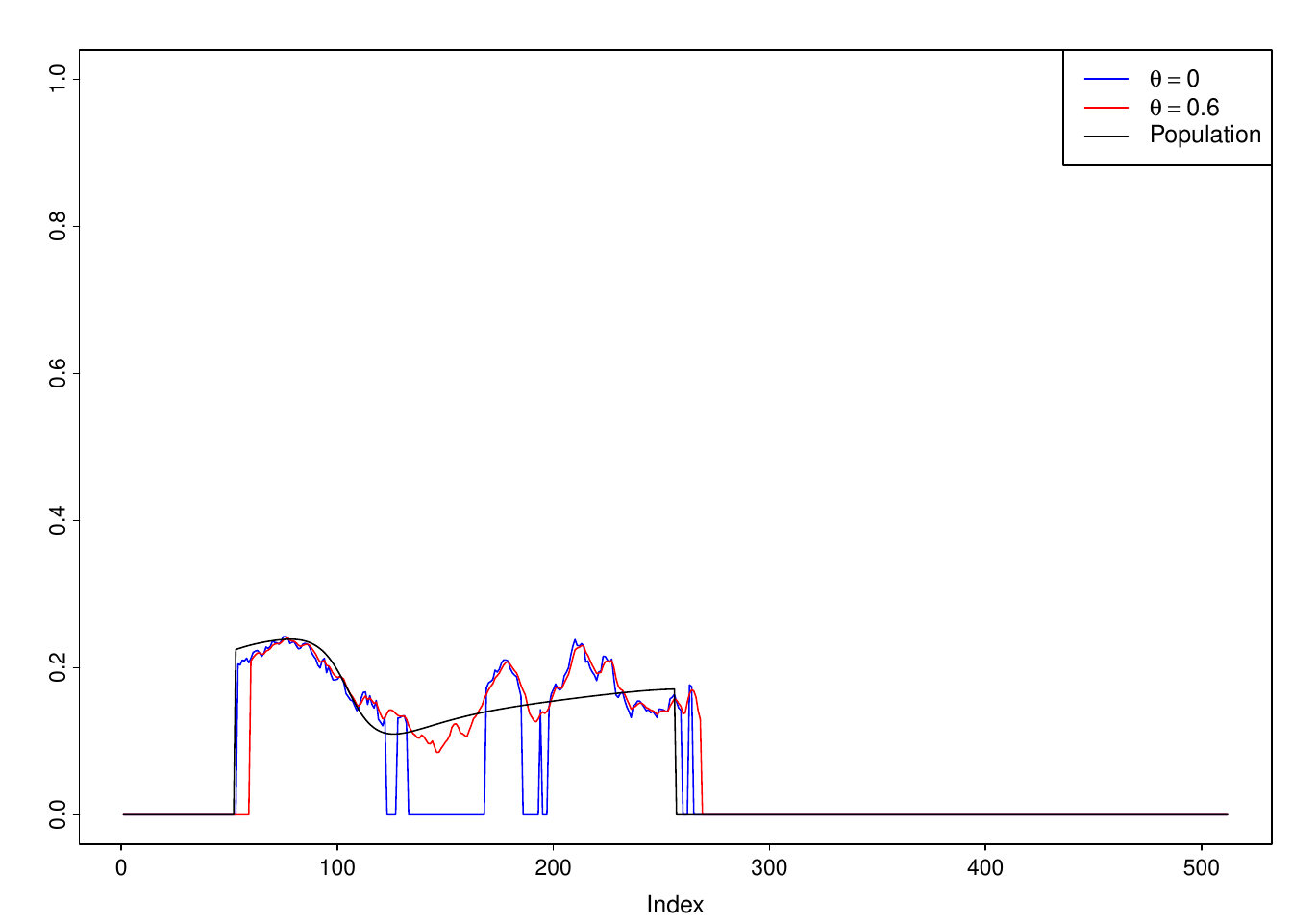} 
\end{tabular}
\vspace{10pt}
         \caption{A coordinate of the leading eigenvector of the spectral density matrix of the process $\{X_t\}$.}
         \label{fig:evec_coordinate}
\end{figure}

\subsection{Tuning Parameter Stability}

We explored changing the order of parameter selection on the final result by comparing the sequential order of initializing $s$ followed by selecting $\eta$ then $\theta$ and $s$, continuing iteratively, with the ordering of initializing with $\eta$ followed by selecting $s$ and then $\theta$ and $\eta$, continuing iteratively. In the case where the parameter selection begins by the initial value $\eta^{init}=3n/8$, which is an overestimate of the frequency parameter in our case, followed by selection of $s$ and $\theta$, iteratively for total of $5$ iterations, we observed that the parameter section and final estimates are not sensitive to the order of parameter selection. Side by side boxplots comparing the average distances between the estimated and the population 1-dimensional principal subspaces of the process are illustrated in Figure \ref{fig:Comparing_Initializations}.

\begin{figure}[H] 
\centering
\begin{tabular}{c}
 \includegraphics[width=5.75in, height=3in, bb = 50 50 600 600]{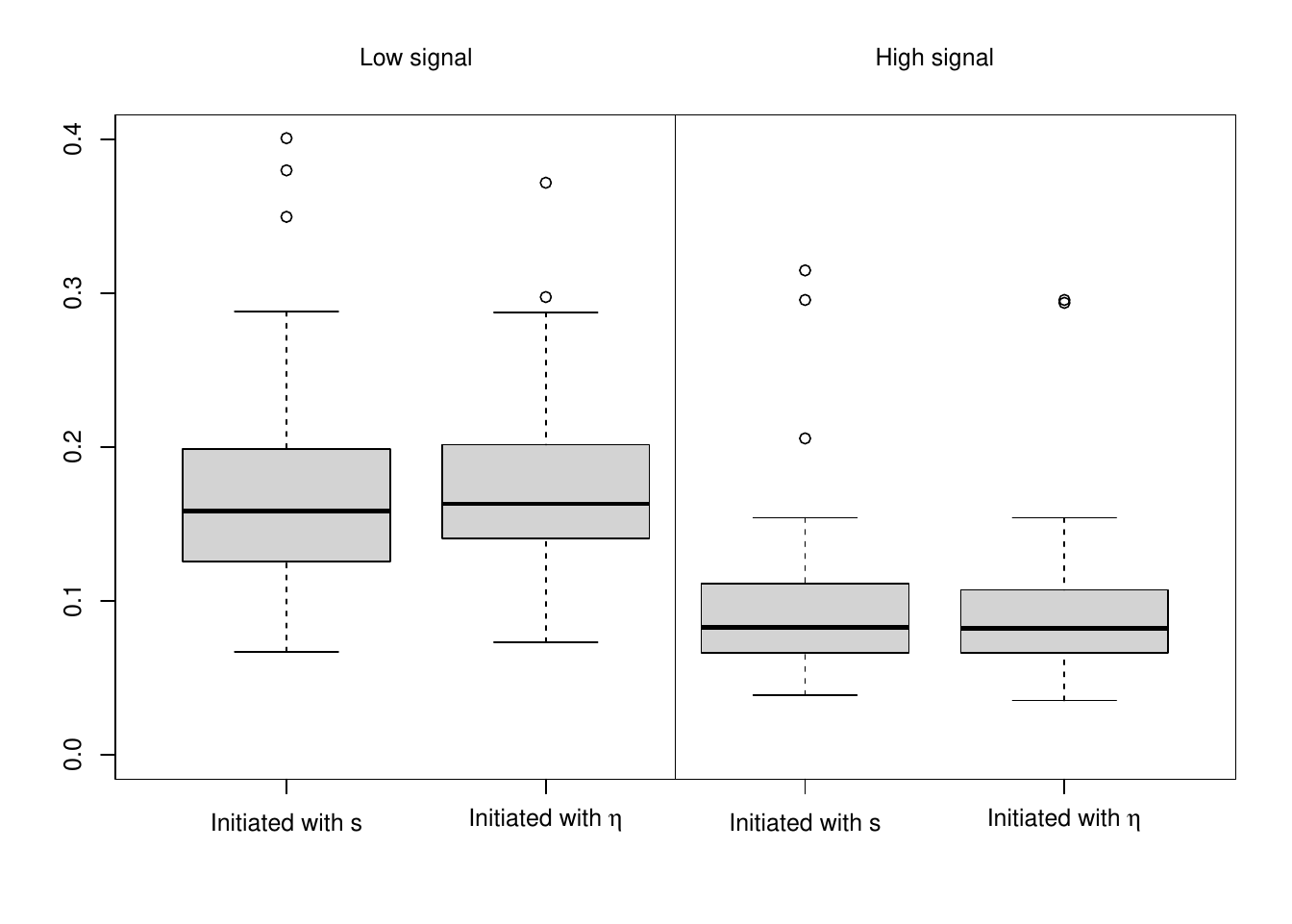}
\end{tabular}
         \caption{Side by side box plots of the mean estimation error of the LSPCA illustrating the effect of changing the order of parameter selection on estimation of the 1-dimensional principal subspaces over $\Omega$.}
         \label{fig:Comparing_Initializations}

\end{figure}

\subsection{Runtime}

Table \ref{table:time} is a summary of the run time of the LSPCA algorithm 
for all combinations of $p=64,n=1024$; $p=64,n=2048$; $p=128,n=1024$; and $p=128,n=2048$.  
Computations were performed in parallel in R on an HPC running Redhat Linux and where each core used has at least 196GB of RAM. For the settings with $p=64$,  4GB of memory and for $p=128$ 8GB of memory were requested. 

\begin{center}
\begin{table}[H] 
\centering
\begin{tabular}{|>{\centering\arraybackslash}m{4.5cm}|>{\centering\arraybackslash}m{.95cm}|>{\centering\arraybackslash}m{.95cm}|>{\centering\arraybackslash}m{.95cm}|>{\centering\arraybackslash}m{.95cm}|>{\centering\arraybackslash}m{.95cm}|>{\centering\arraybackslash}m{.95cm}|>{\centering\arraybackslash}m{.5cm}|>{\centering\arraybackslash}m{.5cm}|>{\centering\arraybackslash}m{.5cm}|>{\centering\arraybackslash}m{.90cm}|>{\centering\arraybackslash}m{.5cm}|>{\centering\arraybackslash}m{.5cm}|>{\centering\arraybackslash}m{.5cm}|>{\centering\arraybackslash}m{.5cm}|>{\centering\arraybackslash}m{.5cm}|>{\centering\arraybackslash}m{.5cm}|>{\centering\arraybackslash}m{.5cm}|>{\centering\arraybackslash}m{.5cm}|>{\centering\arraybackslash}m{.5cm}|>{\centering\arraybackslash}m{.5cm}|
}
    \hline
	$\hat{\eta}$ &min&$Q_1$&$Q_2$&Mean&$Q_3$&max\\
	\hline
	$p = 64, n=1024$ & 1.186 &  1.349 &  1.386  & 1.380 &  1.441 &  1.543     \\
	\hline
	$p = 64, n=2048$ & 2.300 &  2.642 &  2.728  & 2.699 &  2.823  & 3.040   \\
	\hline
	$p = 128, n=1024$ & 1.429 &  1.824 &  1.970 &  1.905 &  2.015 &  2.164     \\
	\hline
	$p = 128, n=2048$ & 2.883 &  3.432 &  3.623 &  3.545 &  3.702 &  4.048   \\
	\hline
	\end{tabular}
    \caption{Run time of the LSPCA algorithm in minutes.}
\label{table:time}
\end{table}
\end{center}

\subsection{Sparsity and Frequency Parameter Selection}
In this experiment the sparsity parameter is selected in an iterative scheme starting with an initial sparsity level of $s^{init} = 16$ followed by selecting $\eta$ iterating $2$ times. Results of the selected sparsity level in 100 runs of the algorithm is as follows. It can be seen that the sparsity parameter selection does not underestimate the true sparsity level more than about 80\% of the times.  

\begin{center}
\begin{table}[H] 
\centering
\begin{tabular}{|>{\centering\arraybackslash}m{4.5cm}|>{\centering\arraybackslash}m{.5cm}|>{\centering\arraybackslash}m{.5cm}|>{\centering\arraybackslash}m{.5cm}|>{\centering\arraybackslash}m{.5cm}|>{\centering\arraybackslash}m{.5cm}|>{\centering\arraybackslash}m{.5cm}|>{\centering\arraybackslash}m{.5cm}|>{\centering\arraybackslash}m{.5cm}|>{\centering\arraybackslash}m{.5cm}|>{\centering\arraybackslash}m{.9cm}|>{\centering\arraybackslash}m{.5cm}|>{\centering\arraybackslash}m{.5cm}|>{\centering\arraybackslash}m{.5cm}|>{\centering\arraybackslash}m{.5cm}|>{\centering\arraybackslash}m{.5cm}|>{\centering\arraybackslash}m{.5cm}|>{\centering\arraybackslash}m{.5cm}|>{\centering\arraybackslash}m{.5cm}|>{\centering\arraybackslash}m{.5cm}|>{\centering\arraybackslash}m{.5cm}|
}
    \hline
	$\hat{s}$ &1&2&3&4&5&6&7&8&9&$\geq 10$\\
	\hline
	$p = 64, n=1024, c=3$ &1&3 &0 & 7& 41 & 19 & 9 & 7 & 2  & 11   \\
	\hline
	$p = 64, n=1024, c=1$ & 4& 3 & 1 & 14 & 45 & 10 & 3 & 3 & 3 & 14  \\
	\hline
	$p = 128, n=1024, c=3$ &2&4 & 6& 6 & 30 & 17 & 4 & 4 & 8 & 19   \\
	\hline
	$p = 128, n=1024, c=1$ & 4& 5 & 0 & 9 & 36 & 10 & 6 & 5 & 4 &  21  \\
	\hline
\end{tabular}
\caption{Sparsity parameter selection  }
\label{table:s-p64n1024}
\end{table}
\end{center}

The following table represents the percentage of frequencies selected by the frequency selection procedure. For instance when $p=64,n=1024,c=1$, in 100 runs of the simulation, the minimum, first quartile, median, mean, third quartile, and maximum  number of selected frequencies are $40,41.2, 41.6, 43.1, 42$, and $90$ percent of the fundamental frequencies. Note that the frequency support of the truncated AR(4) processes considered in the simulation study was $\Omega=[0.05,0.25]$ which contains $40\%$ of the fundamental frequencies. Results presented in Table \ref{table:eta_hat} indicates that when the signal strength is high (in this case when $c=1$) the frequency selection selects the frequency support of the underlying truncated AR(4) process. When the signal strength is low compare to the strength of the noise the frequency selection procedure selects a subset of the frequency support that contains strong power.

\begin{center}
\begin{table}[H] 
\centering
\begin{tabular}{|>{\centering\arraybackslash}m{4.5cm}|>{\centering\arraybackslash}m{.95cm}|>{\centering\arraybackslash}m{.95cm}|>{\centering\arraybackslash}m{.95cm}|>{\centering\arraybackslash}m{.95cm}|>{\centering\arraybackslash}m{.95cm}|>{\centering\arraybackslash}m{.95cm}|>{\centering\arraybackslash}m{.5cm}|>{\centering\arraybackslash}m{.5cm}|>{\centering\arraybackslash}m{.5cm}|>{\centering\arraybackslash}m{.90cm}|>{\centering\arraybackslash}m{.5cm}|>{\centering\arraybackslash}m{.5cm}|>{\centering\arraybackslash}m{.5cm}|>{\centering\arraybackslash}m{.5cm}|>{\centering\arraybackslash}m{.5cm}|>{\centering\arraybackslash}m{.5cm}|>{\centering\arraybackslash}m{.5cm}|>{\centering\arraybackslash}m{.5cm}|>{\centering\arraybackslash}m{.5cm}|>{\centering\arraybackslash}m{.5cm}|
}
    \hline
	$\hat{\eta}$ &min&$Q_1$&$Q_2$&Mean&$Q_3$&max\\
	\hline
	$p = 64, n=1024, c=3$ & 0.20& 19.6& 22.6& 23.3&  25.2& 33.6     \\
	\hline
	$p = 64, n=1024, c=1$ & 40&  41.2&  41.6&  43.1&  42&  90   \\
	\hline
	$p = 128, n=1024, c=3$ & 0.20& 11.2& 14.5& 17.1&  18.5& 82     \\
	\hline
	$p = 128, n=1024, c=1$ & 38.5&  40.6&  41.2&  41.6&  41.8&  64.5   \\
	\hline
	\end{tabular}
    \caption{Percentage of fundamental frequencies selected   }
\label{table:eta_hat}
\end{table}
\end{center}

\section{Data Analysis} \label{Appendix:Application}

This appendix provides two additional analyzes from the LSPCA estimates of the EEG data for the FEP and HC subjects discussed in Section 7 of the manuscript. In addition, we provide the plot of the estimated, in the sense of \cite{Brillinger1964a}, leading eigenvector of the spectral density matrices over the frequency domain of $[0,32]$ Hz and the corresponding spatial plot of the average power over the commonly adopted delta and theta bands. Finally, the estimated band coherence, based on LSPCA, over the two selected frequency bands are displayed and interpreted. 

First, we display the real and imaginary parts of the $d=2$ components as functions of coordinate/channel and frequency for the FEP participant in Figure \ref{fig:Subject3_evecs} and for the healthy control in Figure \ref{fig:Subject43_evecs}. These figures provide more nuanced information compared to the estimated modulus provided in the manuscript, although are of less direct clinical interest. 

\begin{figure}[H]
    \centering
    \includegraphics[width=6.0in, height=3.5in]{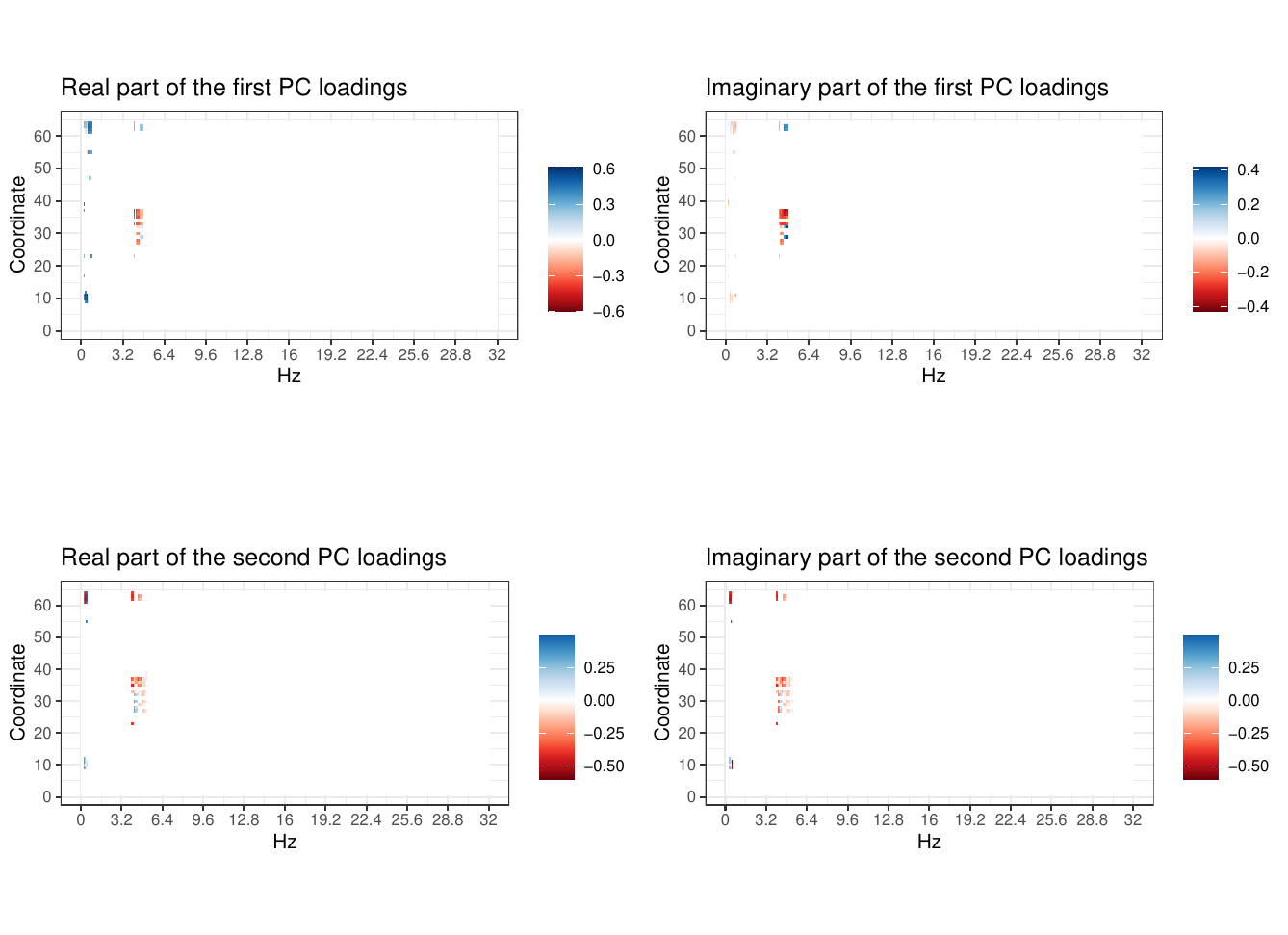}
    \caption{Subject FEP: \textbf{Top left panel} Real part of the first PC loadings; \textbf{Top right pane}: Imaginary part of the first PC loadings. \textbf{Bottom left panel} Real part of the second PC loadings; \textbf{Bottom right pane}: Imaginary part of the second PC loadings.}
    \label{fig:Subject3_evecs}
\end{figure}

\begin{figure}[H]
    \centering
    \includegraphics[width=6.0in, height=3.5in]{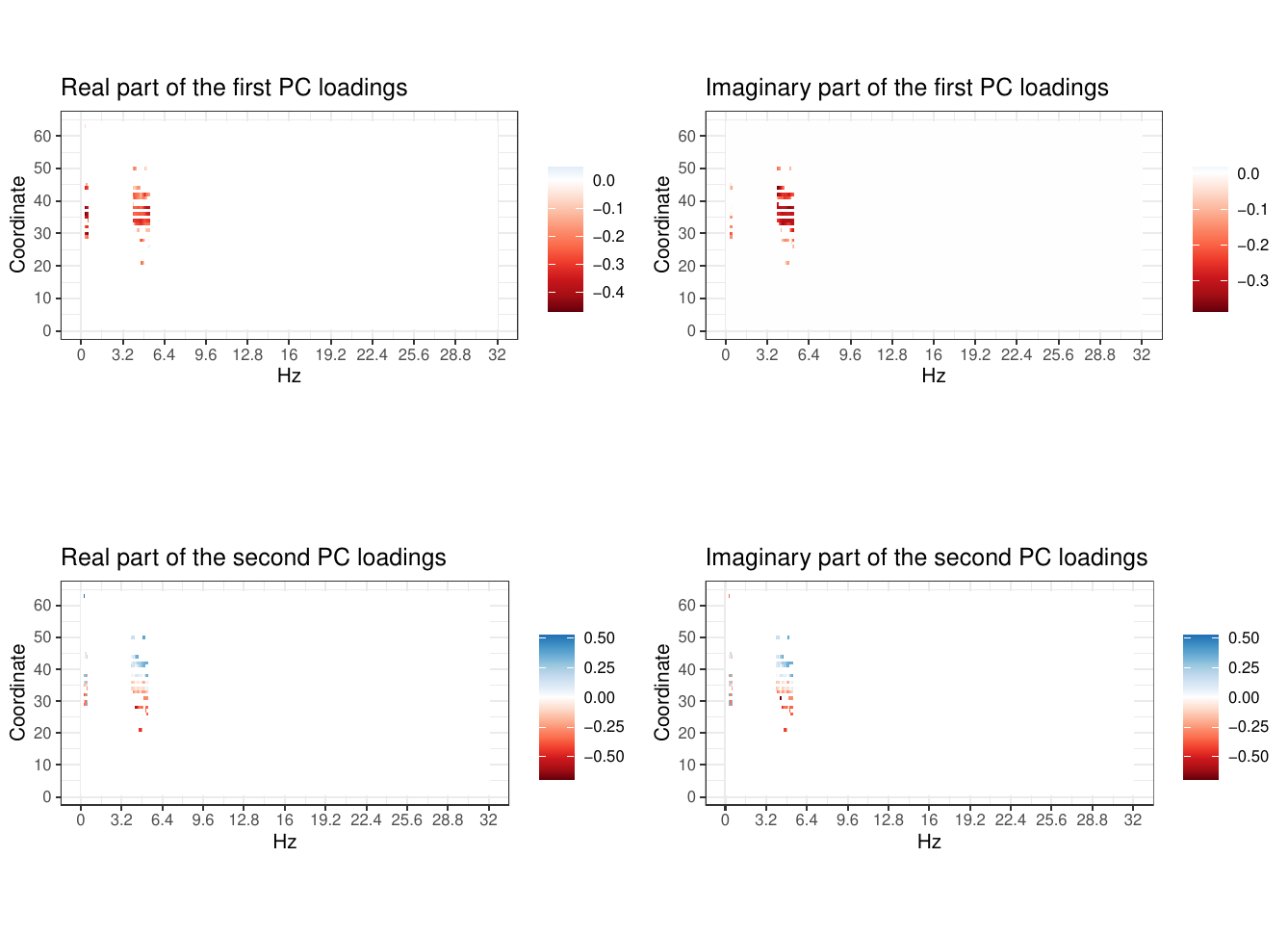} 
    \caption{Subject HC: \textbf{Top left panel} Real part of the first PC loadings; \textbf{Top right pane}: Imaginary part of the first PC loadings. \textbf{Bottom left panel} Real part of the second PC loadings; \textbf{Bottom right pane}: Imaginary part of the second PC loadings.}
    \label{fig:Subject43_evecs}
\end{figure}

To provide a visual comparison between the estimated leading eigenvector obtained by (1) the LSPCA, illustrated in Figures 5 and 6, and (2) the classical eigenvector estimates, we illustrate the plot of the modulus of the estimated leading eigenvector, in the sense of \cite{Brillinger1964a}, in Figure \ref{fig:subjects_evecs_classic} and the average power within the delta band of less than 4 Hz and theta band of 4 to 8 Hz in Figure \ref{fig:eegcap_classic}. Note that the scale of the color map in Figure \ref{fig:eegcap_classic} is not the same as that obtained from the LSPCA, where in Figure \ref{fig:eegcap_classic} bolder red color represents a higher power relative to the other channels. To better illustrate how not regularizing the estimation of eigenvectors affects power estimation, we matched the scale of the color map to the one obtained by the LSPCA, illustrated in Figure \ref{fig:eegcap_classic_color_scaled}. The estimation of the eigenvectors results in the noisy spread of power across all channels and in underestimation of the power within each channel.


\begin{figure}[H] 
\begin{center}
\begin{tabular}{cc}
\includegraphics[scale = .32]{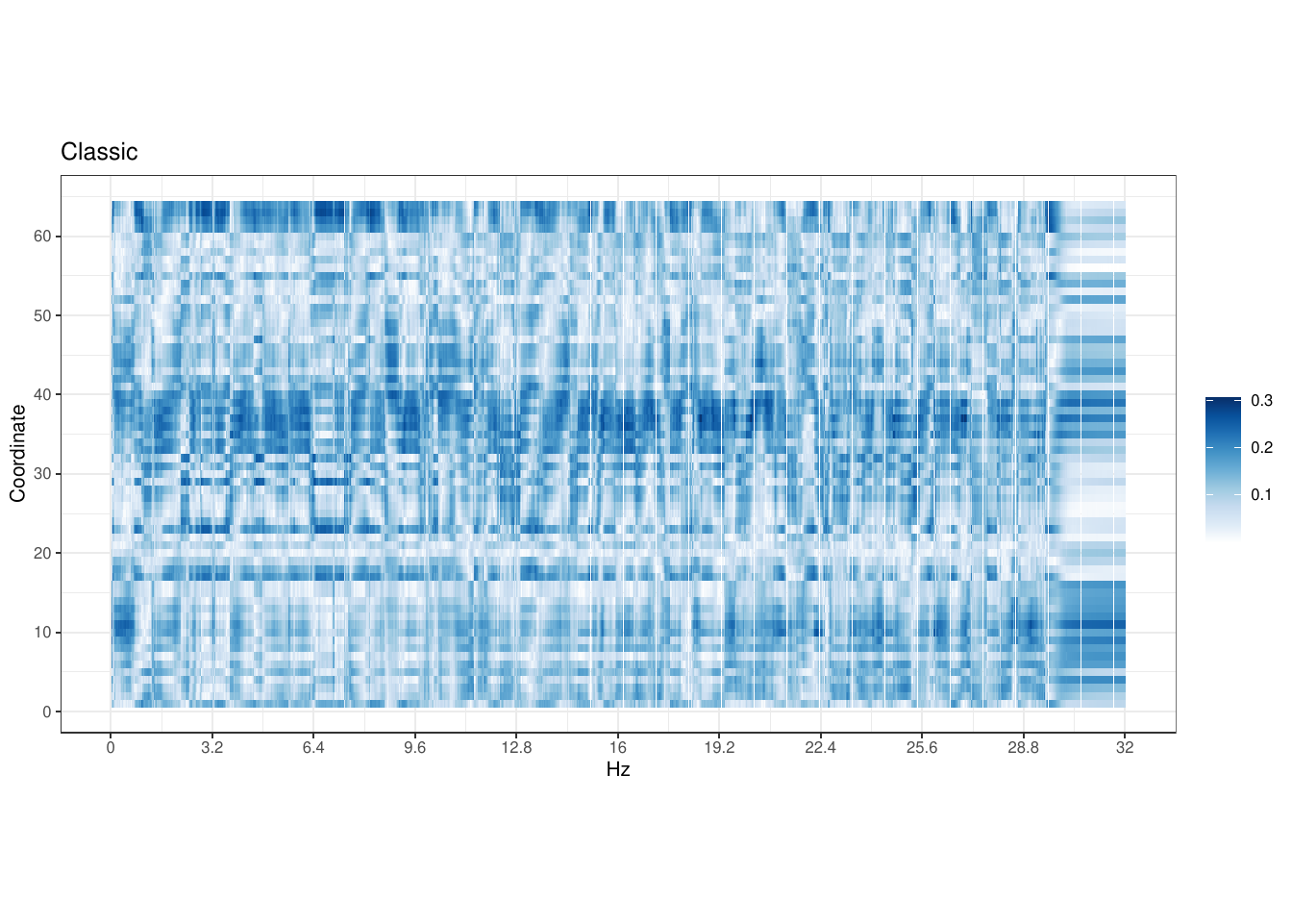} 
& \includegraphics[scale = .32]{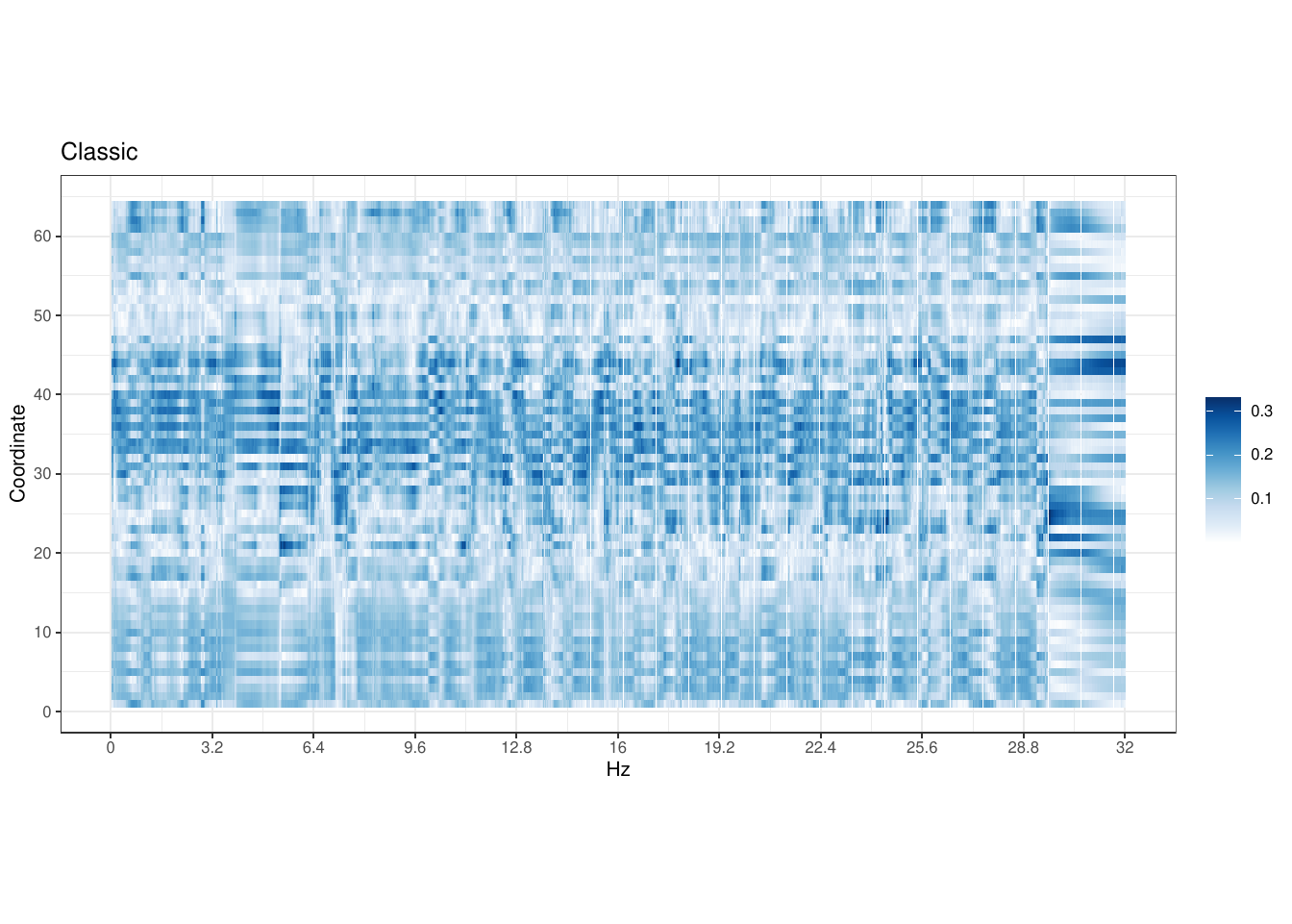} 
\end{tabular}
         \caption{Plots represent the estimated modulus of the first PC loadings by the classical method, corresponding to \textbf{left panel} the FEP subject; \textbf{ right pane}:  the HC subject.}
         \label{fig:subjects_evecs_classic}
\end{center}
\end{figure}

\begin{figure}[H] 
\begin{center}
\begin{tabular}{cc}
\includegraphics[scale = .3]{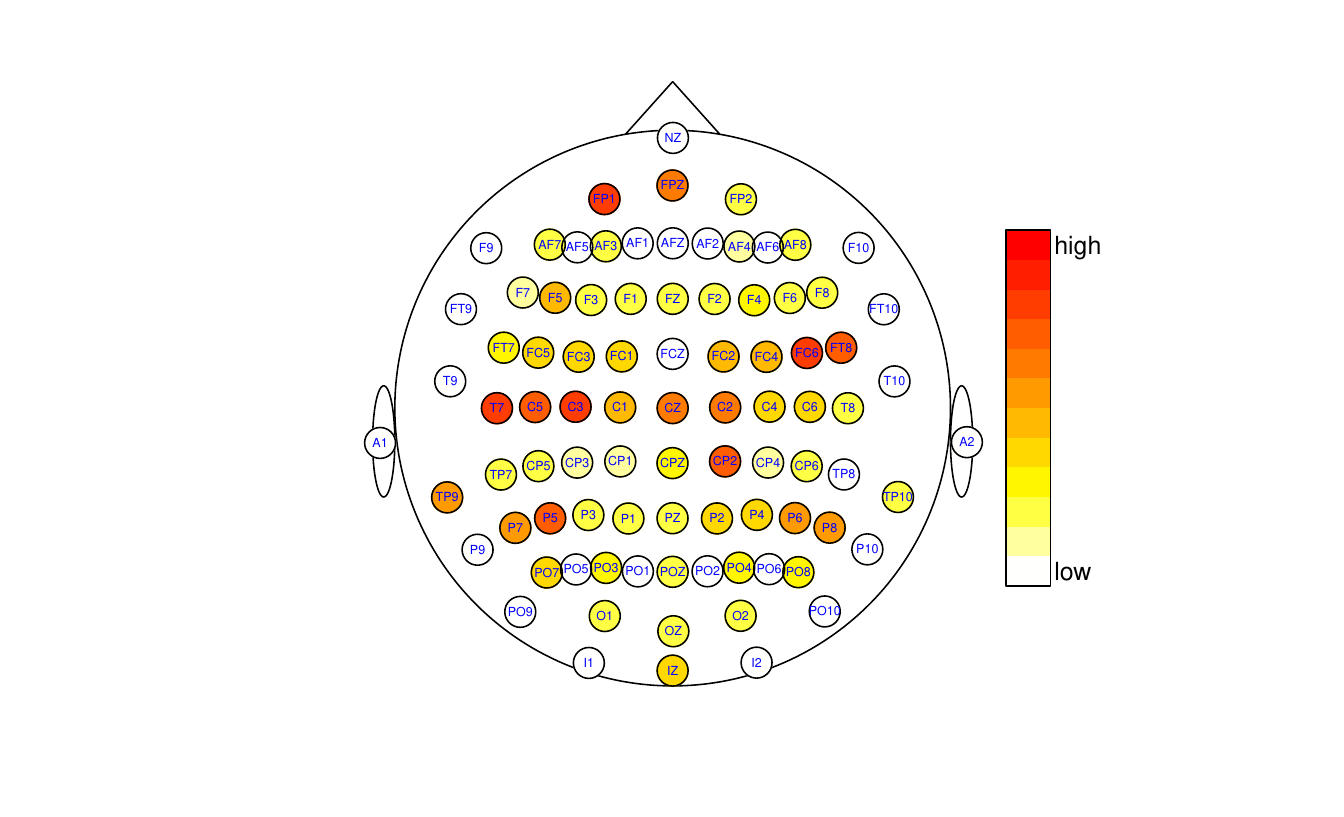} 
& \includegraphics[scale = .3]{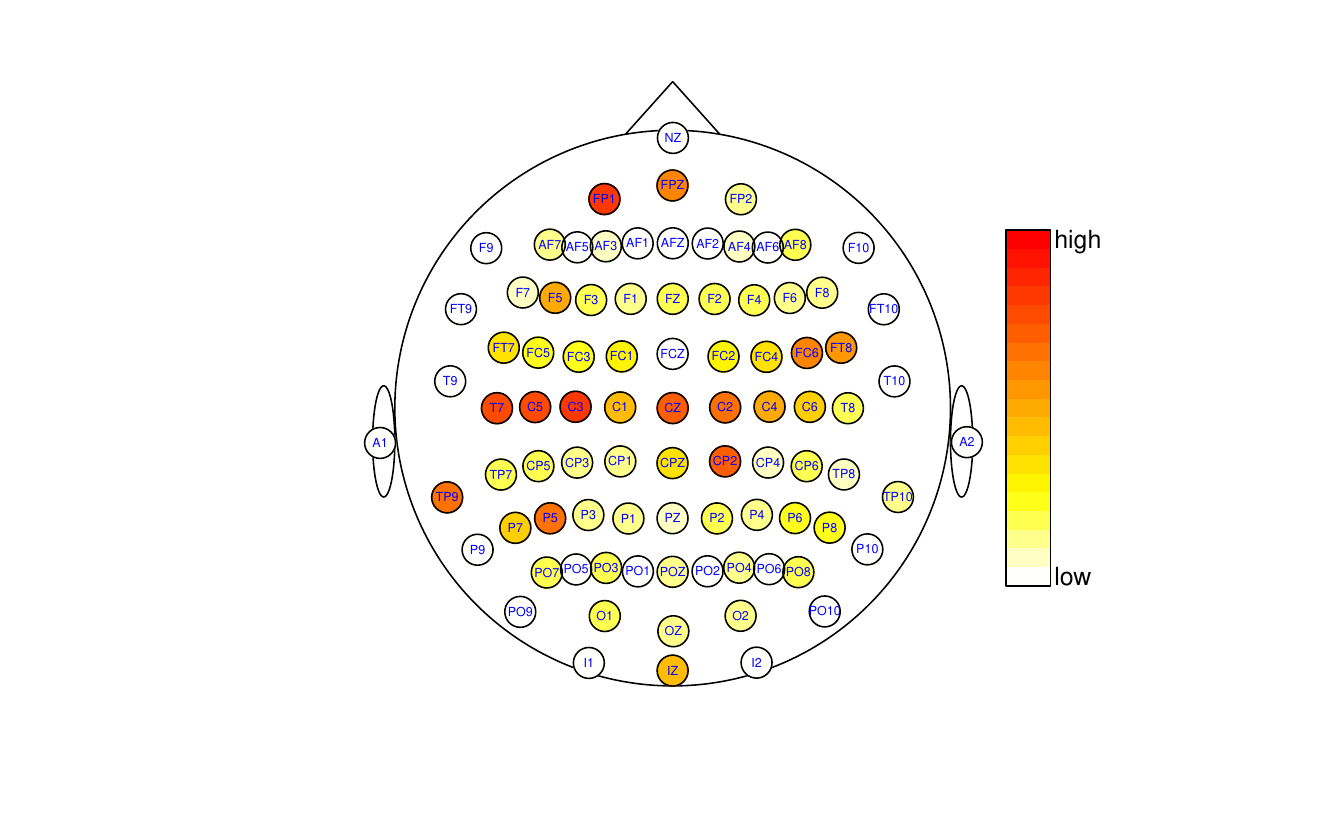} \\
\includegraphics[scale = .3]{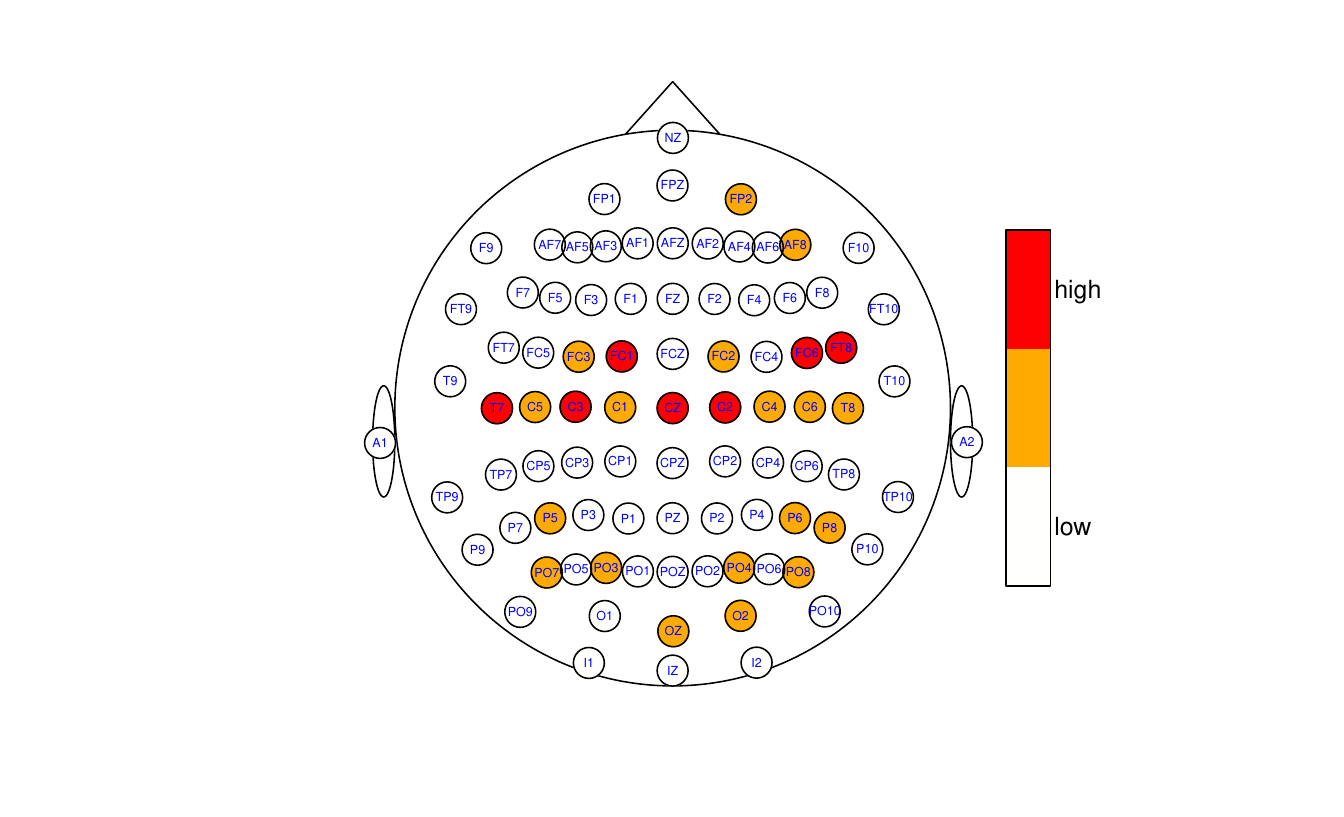}
& \includegraphics[scale = .3]{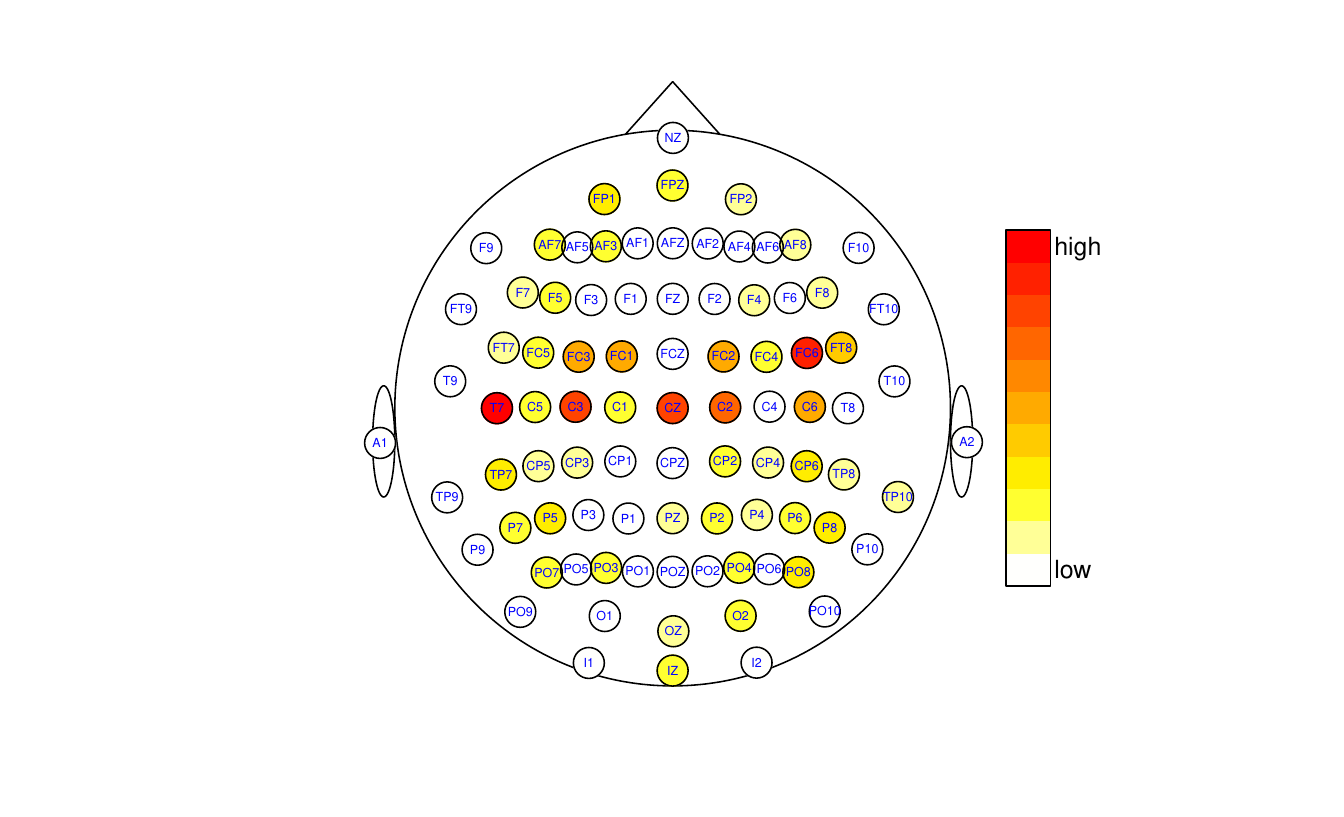}
\end{tabular}
         \caption{Band-collapsed power estimated using the classical approach for the FEP (top row) and HC (bottom row) participants for the lower frequency band (left column) and higher frequency band (right column); Scale of the color map is not matched with the one obtained by the LSPCA and presented in the manuscript.}
        \label{fig:eegcap_classic}
\end{center}
\end{figure}

\begin{figure}[H] 
\begin{center}
\begin{tabular}{cc}
\includegraphics[scale = .3]{Subject3_band1_V0.pdf} 
& \includegraphics[scale = .3]{Subject3_band2_V0.pdf} \\
\includegraphics[scale = .3]{Subject43_band1_V0.pdf}
& \includegraphics[scale = .3]{Subject43_band2_V0.pdf}
\end{tabular}
         \caption{Band-collapsed power estimated using the classical approach for the FEP (top row) and HC (bottom row) participants for the lower frequency band (left column) and higher frequency band (right column); Scale of the color map is matched with the one obtained by the LSPCA presented in the manuscript.}
        \label{fig:eegcap_classic_color_scaled}
\end{center}
\end{figure}

Next, we analyze coherence. Figure \ref{fig:Coherence} represents the estimated band coherence, where band coherence is defined as suggested in \cite{ombao2006coherence}, as follows.
\begin{align*} \label{eq:band_coh}
    \mathcal{K}_{k,\ell}(\mathcal{B}) = \frac{\int_{\mathcal{B}}f_{k,\ell}(\omega)\;d\omega}{\sqrt{\int_{\mathcal{B}}f_{k,k}(\omega)\;d\omega \;\int_{\mathcal{B}}f_{\ell,\ell}(\omega)\;d\omega}},    
\end{align*}
where $\mathcal{B}\subset[0,.5]$. Estimate of $\mathcal{K}_{k,\ell}(\mathcal{B})$ were obtained by 
\begin{align*}
    \hat{\mathcal{K}}_{k,\ell}(\mathcal{B}) = \frac{\sum_{\{j:\omega_j\in\mathcal{B}\}}\hat{f}_{k,\ell}(\omega_j)}{\sqrt{\sum_{\{j:\omega_j\in\mathcal{B}\}}\hat{f}_{k,k}(\omega_j) \;\sum_{\{j:\omega_j\in\mathcal{B}\}}\hat{f}_{\ell,\ell}(\omega_j)}},
\end{align*}
where $\hat{f}$ is the rank 1 estimate of the spectral density matrices computed as follows
\begin{align*}
    \hat{f}(\omega) = \hat{\lambda}(\omega)\hat{V}(\omega)\hat{V}(\omega)^\dagger,
\end{align*}
and $\hat{V}(\omega)$ is the estimated sparse eigenvector of $f(\omega)$ by the LSPCA and $\hat{\lambda}(\omega)$ is the estimated leading eigenvalue. Plots represent the coherence between the selected channels by the LSPCA over the two selected frequency bands, where the lower delta frequency band is denoted by $\mathcal{B}_1$ and the higher theta frequency band is denoted by $\mathcal{B}_2$.

The coherence maps reveal some expected similarities among the two subjects.  Both subjects displayed highly coherent delta and theta power in frontal and central blocks that span both hemispheres of the brain.  However, the FEP participant displays an additional coherent parietal region (Iz, PO7, P8, P6, P4) across both hemispheres that is not present in the HC. Increased inter- and intra-hemisphere parietal delta synchronization have been reported in people with schizophrenia compared to healthy controls \citep{kam2013, perrotelli2021}.  Identifying characteristics in FEP patients that are are similar to those without mental health conditions and characteristics that are similar to people with schizophrenia are of clinical interest, as they can help understand the course and development of the disease.  Our analysis suggests that increased parietal theta coherence could be a marker that can be explored in future studies.

\begin{figure}[H] 
\begin{center}
\begin{tabular}{cc}
\includegraphics[width=2.9in, height=3.5in, bb = 50 50 650 650]{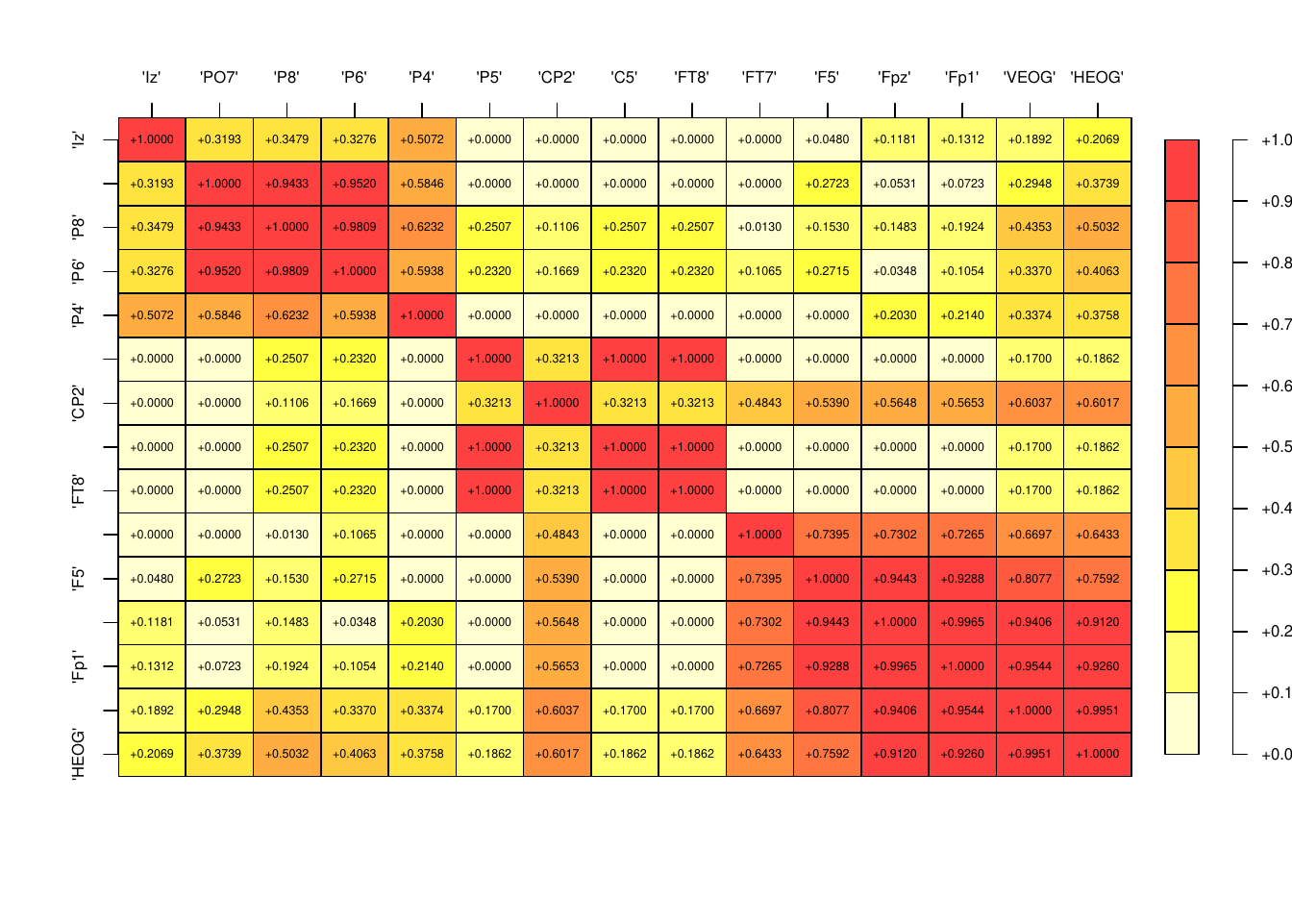} 
& \includegraphics[width=2.9in, height=3.5in, bb = 50 50 650 650]{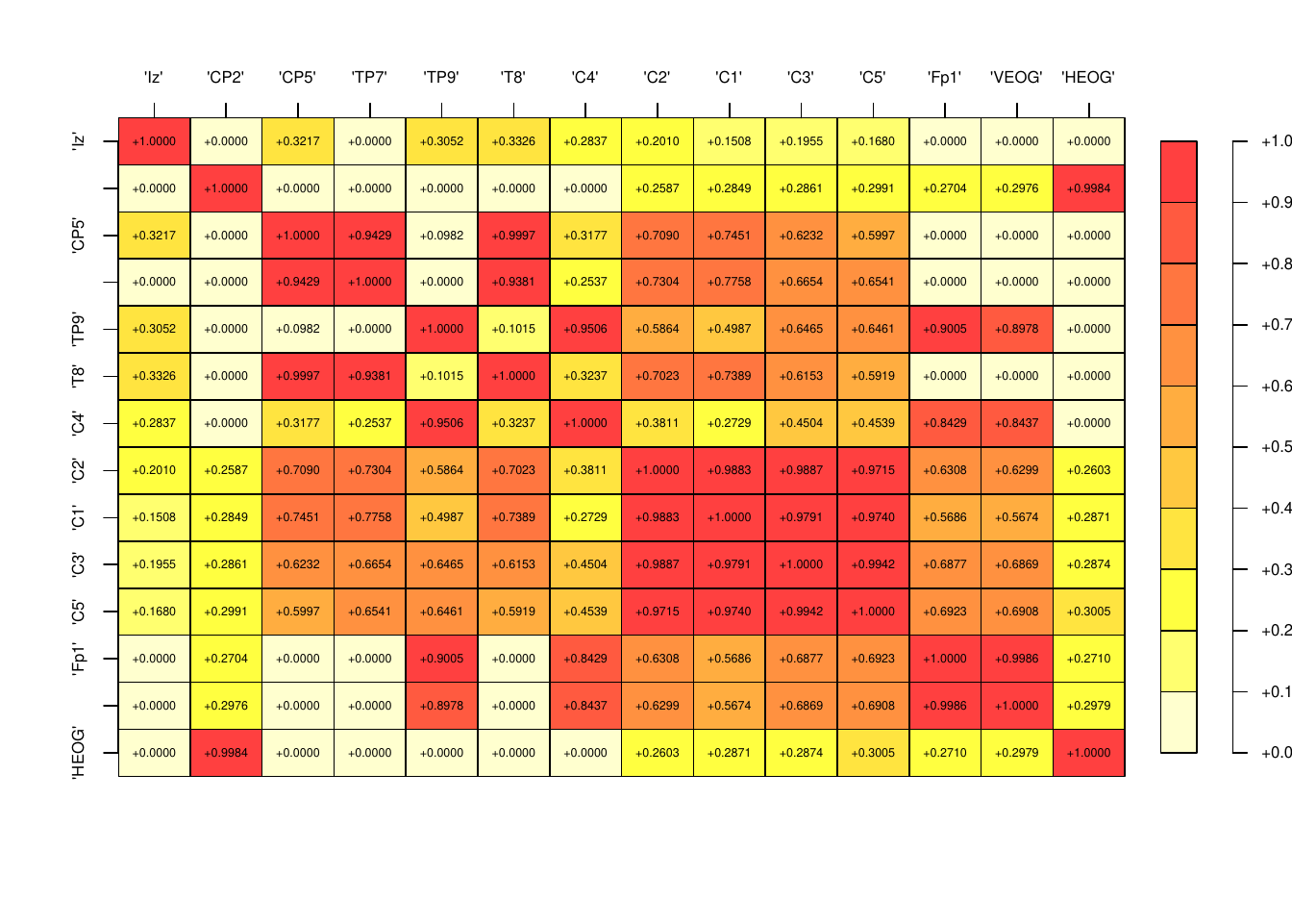} \vspace{-50pt}\\
\includegraphics[width=2.9in, height=3.5in, bb = 50 50 650 650]{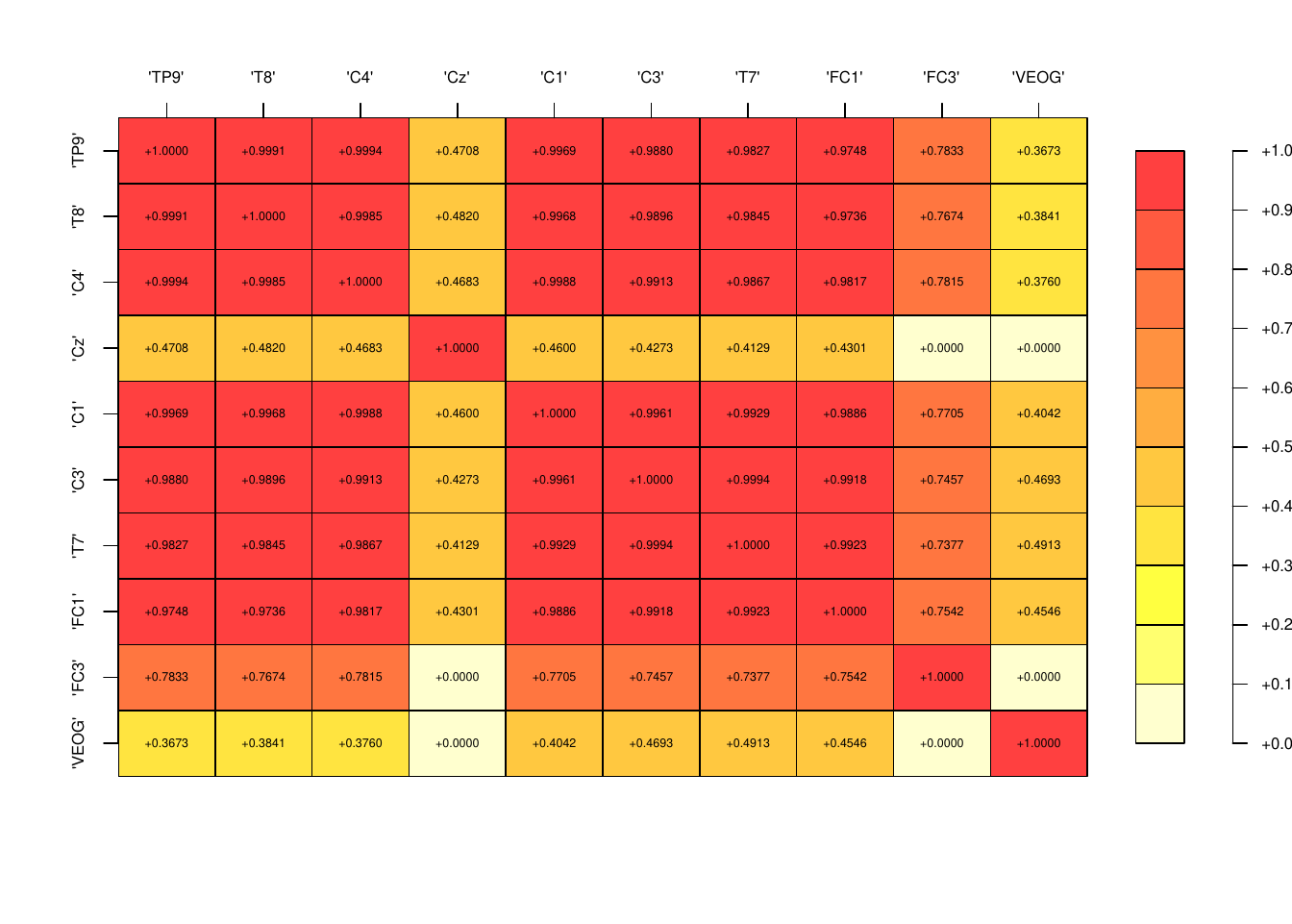}
& \includegraphics[width=2.9in, height=3.5in, bb = 50 50 650 650]{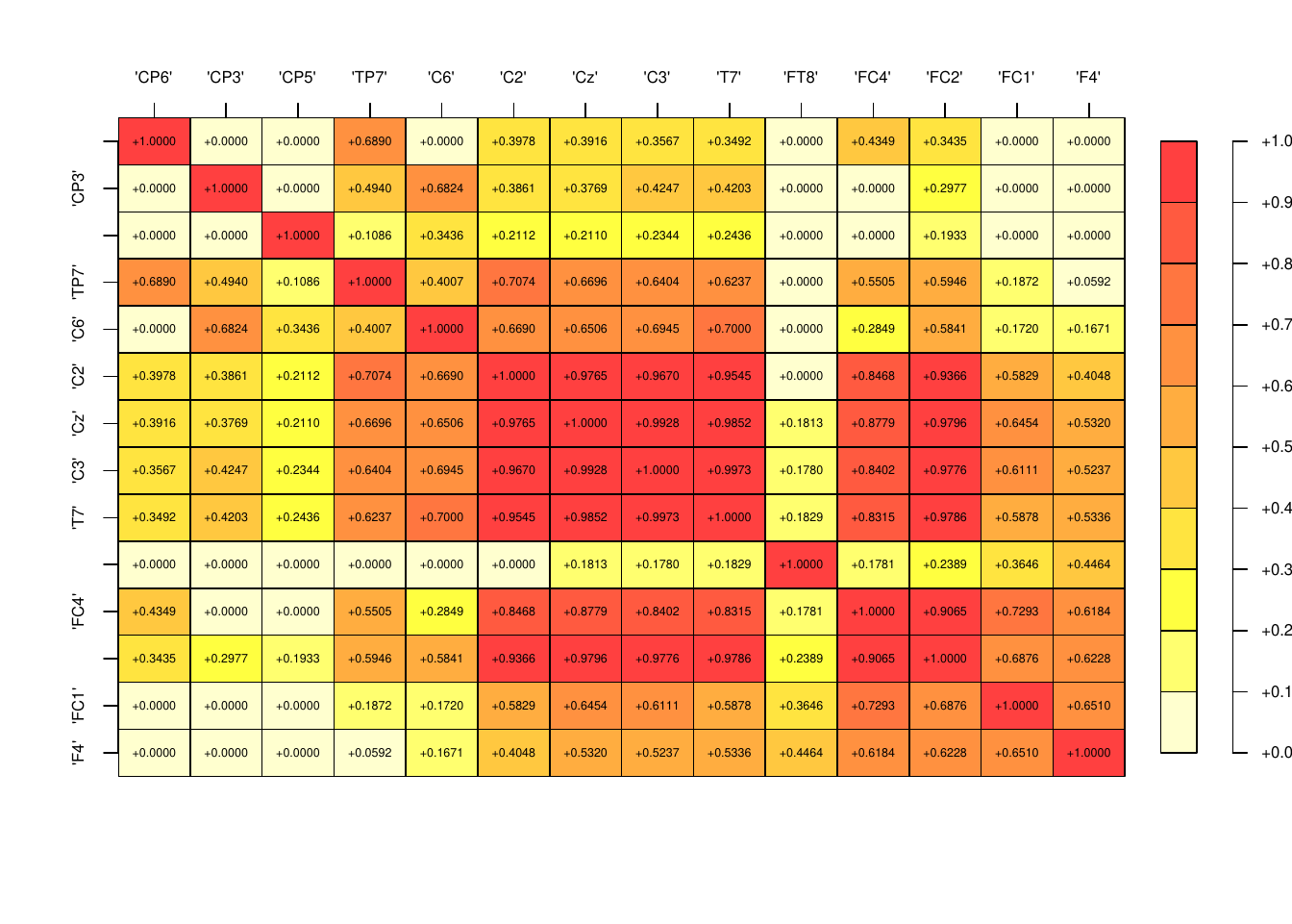}
\end{tabular}
         \caption{Plots represent the Coherence between the selected channels by the LSPCA where \textbf{Top left} panel corresponds to the FEP subject over $\mathcal{B}_1$ (lower frequency band); \textbf{Top right} panel corresponds to the FEP subject over $\mathcal{B}_2$ (higher frequency band); \textbf{Bottom left} panel corresponds to the HC subject over $\mathcal{B}_1$ (lower frequency band); \textbf{Bottom right} panel corresponds to the HC subject over $\mathcal{B}_2$ (higher frequency band).}
         \label{fig:Coherence}
\end{center}
\end{figure}

\bibliographystyle{asa}

\bibliography{LSPCA}

@article{forni2000,
    author = {Forni, Mario and Hallin, Marc and Lippi, Marco and Reichlin, Lucrezia},
    title = {The Generalized Dynamic-Factor Model: Identification and Estimation},
    journal = {The Review of Economics and Statistics},
    volume = {82},
    number = {4},
    pages = {540-554},
    year = {2000},
    month = {11},
    issn = {0034-6535},
    doi = {10.1162/003465300559037},
    url = {https://doi.org/10.1162/003465300559037},
    eprint = {https://direct.mit.edu/rest/article-pdf/82/4/540/1613058/003465300559037.pdf},
}

@article{pearson1901liii,
  title={LIII. On Lines and Planes of Closest Fit to Systems of Points in Space},
  author={Pearson, Karl},
  journal={The London, Edinburgh, and Dublin Philosophical Magazine and Journal of Science},
  volume={2},
  number={11},
  pages={559--572},
  year={1901},
  publisher={Taylor \& Francis}
}

@article{kam2013,
  title={Resting state EEG power and coherence abnormalities in bipolar disorder and schizophrenia},
  author={Kam, JW and Bolbecker, AR and O'Donnell, BF and Hetrick, WP and Brenner, CA},
  journal={Journal of Psychiatric Research},
  year={2013},
  volume={47},
  issue = {12}, 
  pages={1893-1901}
}

@ARTICLE{perrotelli2021,  
AUTHOR={Perrottelli, Andrea  and Giordano, Giulia Maria  and Brando, Francesco  and Giuliani, Luigi  and Mucci, Armida },         
TITLE={EEG-Based Measures in At-Risk Mental State and Early Stages of Schizophrenia: A Systematic Review},        
JOURNAL={Frontiers in Psychiatry},        
VOLUME={12},
YEAR={2021},
ISSN={1664-0640}
}

@article{johnstone2009consistency,
  title={On Consistency and Sparsity for Principal Components Analysis in High Dimensions},
  author={Johnstone, Iain M and Lu, Arthur Yu},
  journal={Journal of the American Statistical Association},
  volume={104},
  number={486},
  pages={682--693},
  year={2009},
  publisher={Taylor \& Francis}
}

@article{paul2007asymptotics,
  title={Asymptotics of Sample Eigenstructure for a Large Dimensional Spiked Covariance Model},
  author={Paul, Debashis},
  journal={Statistica Sinica},
  pages={1617--1642},
  year={2007},
  volume = {14}
}

@article{johnstone2018pca,
  title={PCA in High Dimensions: An Orientation},
  author={Johnstone, Iain M and Paul, Debashis},
  journal={Proceedings of the IEEE},
  volume={106},
  number={8},
  pages={1277--1292},
  year={2018},
  publisher={IEEE}
}

@article{zou2006sparse,
  title={Sparse Principal Component Analysis},
  author={Zou, Hui and Hastie, Trevor and Tibshirani, Robert},
  journal={Journal of Computational and Graphical Statistics},
  volume={15},
  number={2},
  pages={265--286},
  year={2006},
  publisher={Taylor \& Francis}
}

@article{zou2018selective,
  title={A Selective Overview of Sparse Principal Component Analysis},
  author={Zou, Hui and Xue, Lingzhou},
  journal={Proceedings of the IEEE},
  volume={106},
  number={8},
  pages={1311--1320},
  year={2018},
  publisher={IEEE}
}

@article{shen2008sparse,
  title={Sparse Principal Component Analysis via Regularized Low Rank Matrix Approximation},
  author={Shen, Haipeng and Huang, Jianhua Z},
  journal={Journal of Multivariate Analysis},
  volume={99},
  number={6},
  pages={1015--1034},
  year={2008},
  publisher={Elsevier}
}

@article{d2004direct,
  title={A Direct Formulation for Sparse PCA Using Semidefinite Programming},
  author={d'Aspremont, Alexandre and Ghaoui, Laurent and Jordan, Michael and Lanckriet, Gert},
  journal={Advances in Neural Information Processing Systems},
  volume={17},
  year={2004}
}

@article{vu2013fantope,
  title={Fantope Projection and Selection: A Near-Optimal Convex Relaxation of Sparse PCA},
  author={Vu, Vincent Q and Cho, Juhee and Lei, Jing and Rohe, Karl},
  journal={Advances in Neural Information Processing Systems},
  volume={26},
  year={2013}
}

@article{vu2013minimax,
author = {Vincent Q. Vu and Jing Lei},
title = {{Minimax Sparse Principal Subspace Estimation in High Dimensions}},
volume = {41},
journal = {The Annals of Statistics},
number = {6},
publisher = {Institute of Mathematical Statistics},
pages = {2905 -- 2947},
year = {2013},
doi = {10.1214/13-AOS1151},
URL = {https://doi.org/10.1214/13-AOS1151}
}

@article{ma2013sparse,
  title={Sparse Principal Component Analysis and Iterative Thresholding},
  author={Ma, Zongming},
  journal={The Annalls of Statistics},
  pages={772--801},
  year={2013},
  volume = {41},
  number = {2}
}

@article{wang2014nonconvex,
  title={Nonconvex Statistical Optimization: Minimax-Optimal Sparse PCA in Polynomial Time},
  author={Wang, Zhaoran and Lu, Huanran and Liu, Han},
  journal={arXiv preprint arXiv:1408.5352},
  year={2014}
}

@inproceedings{wang2013sparse_s,
  title={Sparse Principal Component Analysis for High Dimensional Multivariate Time Series},
  author={Wang, Zhaoran and Han, Fang and Liu, Han},
  booktitle={Artificial Intelligence and Statistics},
  pages={48--56},
  year={2013},
  organization={PMLR}
}

@article{yuan2013truncated,
  title={Truncated Power Method for Sparse Eigenvalue Problems.},
  author={Yuan, Xiao-Tong and Zhang, Tong},
  journal={Journal of Machine Learning Research},
  volume={14},
  number={4},
  year={2013}, 
  pages = {899-925}
}

@article{witten2009penalized,
  title={A Penalized Matrix Decomposition, with Applications to Sparse Principal Components and Canonical Correlation Analysis},
  author={Witten, Daniela M and Tibshirani, Robert and Hastie, Trevor},
  journal={Biostatistics},
  volume={10},
  number={3},
  pages={515--534},
  year={2009},
  publisher={Oxford University Press}
}

@article{Brillinger1964a,
  title={The Generalization of Techniques of Factor Analysis Canonical Correlation and Principal Component to Stationary Time Series},
  author={Brillinger, Davis R},
  journal={Invited paper at Royal Statistical Society Conference in Cardiff, Wales},
  year={1964},
}

@book{brillinger2001time,
  title={Time Series: Data Analysis and Theory},
  author={Brillinger, David R},
  year={2001},
  publisher={SIAM}
}

@article{lusparse,
  title={Sparse Principal Component Analysis in Frequency Domain for Time Series},
  author={Lu, Junwei and Chen, Yichen and Zhu, Xiuneng and Han, Fang and Liu, Han},
  journal = {https://junwei-lu.github.io/papers/FourierPCA.pdf}, 
  url={https://junwei-lu.github.io/papers/FourierPCA.pdf},
  year = {2016}
}

@article{moghaddam2005spectral,
  title={Spectral Bounds for Sparse PCA: Exact and Greedy Algorithms},
  author={Moghaddam, Baback and Weiss, Yair and Avidan, Shai},
  journal={Advances in Neural Information Processing Systems},
  volume={18},
  year={2005}
}

@book{golub2013matrix,
  title={Matrix Computations},
  author={Golub, Gene H and Van Loan, Charles F},
  year={2013},
  publisher={JHU press}
}

@article{james2009functional,
author = {Gareth M. James and Jing Wang and Ji Zhu},
title = {{Functional Linear Regression That’s Interpretable}},
volume = {37},
journal = {The Annals of Statistics},
number = {5A},
publisher = {Institute of Mathematical Statistics},
pages = {2083 -- 2108},
year = {2009},
doi = {10.1214/08-AOS641},
URL = {https://doi.org/10.1214/08-AOS641}
}

@book{stewart1990perturbation,
  title={Matrix Pertrubation Theory},
  author={Stewart, G.W. and Sun, Ji-guang},
  year={1990},
  publisher={Academic Press}
}

@book{bhatia2013matrix,
  title={Matrix Analysis},
  author={Bhatia, Rajendra},
  volume={169},
  year={2013},
  publisher={Springer Science \& Business Media}
}

@article{renaldi2019predicting,
  title={Predicting Symptomatic and Functional Improvements over 1 Year in Patients with First-Episode Psychosis Using Resting-State Electroencephalography},
  author={Renaldi, Rinvil and Kim, Minah and Lee, Tak Hyung and Kwak, Yoo Bin and Tanra, Andi J and Kwon, Jun Soo},
  journal={Psychiatry Investigation},
  volume={16},
  number={9},
  pages={695},
  year={2019},
  publisher={Korean Neuropsychiatric Association}
}

@article{bruce2020,
  title={Empirical Frequency Band Analysis of Nonstationary Time Series},
  author={Bruce, SA and Tang, CY and Hall, MH and Krafty, RT},
  journal={Journal of the American Statistical Association},
  volume={115},
  number={532},
  pages={1933--1945},
  year={2020}
}

@article{tuft2023,
  title={Spectra in Low-Rank Localized Layers (SpeLLL) for Interpretable Time--Frequency Analysis},
  author={Tuft, Marie and Hall, Martica H and Krafty, Robert T},
  journal={Biometrics},
  volume={79},
  number={1},
  pages={304--318},
  year={2023},
  publisher={Wiley Online Library}
}

@article{zhang2021,
    author = {Zhang, Jun and Siegle, Greg J and Sun, Tao and D’andrea, Wendy and Krafty, Robert T},
    title = "{Interpretable Principal Component Analysis for Multilevel Multivariate Functional Data}",
    journal = {Biostatistics},
    volume = {24},
    number = {2},
    pages = {227-243},
    year = {2021},
    month = {09},
    issn = {1465-4644},
}

@article{krafty2011,
  title={Functional Mixed Effects Spectral Analysis},
  author={Krafty, RT and Hall, MH and Guo, W},
  journal={Biometrika},
  volume={98},
  number={3},
  pages={583--598},
  year={2011}
}

@article{krafty2015,
  title={Discriminant Analysis of Time Series in the Presence of Within-Group Spectral Variability},
  author={Krafty, RT},
  journal={Journal of Time Series Analysis},
  volume={37},
  number={4},
  pages={435--450},
  year={2015}
}

@article{delorme2004,
  title={EEGLAB: An Open Source Toolbox for Analysis of Single-Trial EEG Dynamics Including Independent Component Analysis},
  author={Delorme, A and Makeig, S},
  journal={Journal of Neuroscience Methods},
  volume={134},
  number={1},
  pages={9--21},
  year={2004}
}

@article{maullinsapey2023,
    author = {Maullin-Sapey, Thomas and Schwartzman, Armin and Nichols, Thomas E},
    title = {Spatial Confidence Regions for Combinations of Excursion Sets in Image Analysis},
    journal = {Journal of the Royal Statistical Society Series B: Statistical Methodology},
    volume = {86},
    number = {1},
    pages = {177-193},
    year = {2023},
    month = {09}
}

@article{jiao2021,
    author = {Jiao, Shuhao and Shen, Tong and Yu, Zhaoxia and Ombao, Hernando},
    title = {Change-Point Detection Using Spectral PCA for Multivariate Time Series},
    journal = {Journal of the Royal Statistical Society Series B: Statistical Methodology},
    year = {2021},
    url = {https://arxiv.org/abs/2101.04334}
}

@article{sundararajan2021,
    author = {Raanju R Sundararajan},
    title = "Principal Component Analysis Using Frequency Components of Multivariate Time Series",
    journal = {Compuational Statistics and Data Analysis},
    volume = {157},
    number = {107164},
    year = {2021},
}

@article{ombao2005,
author = {Hernando Ombao and Rainer von Sachs and Wensheng Guo},
title = {SLEX Analysis of Multivariate Nonstationary Time Series},
journal = {Journal of the American Statistical Association},
volume = {100},
number = {470},
pages = {519--531},
year = {2005}
}

@article{merlevede2011bernstein,
  title={A Bernstein Type Inequality and Moderate Deviations for Weakly Dependent Sequences},
  author={Merlev{\`e}de, Florence and Peligrad, Magda and Rio, Emmanuel},
  journal={Probability Theory and Related Fields},
  volume={151},
  pages={435--474},
  year={2011},
  publisher={Springer}
}

@article{ombao2006coherence,
  title={Coherence analysis of nonstationary time series: a linear filtering point of view},
  author={Ombao, Hernando and Van Bellegem, S{\'e}bastien and others},
  journal={IEEE Transactions on Signal Processing},
  volume={56},
  pages={2259--2266},
  year={2006}
}

\end{document}